\documentclass[12pt]{article}
\usepackage{amsmath}
\usepackage{graphicx}
\usepackage{enumerate}
\usepackage{url} 

\newcommand{\blind}{1}

\usepackage{graphicx}
\usepackage{amssymb}
\usepackage{amsmath}
\usepackage{amsfonts}
\usepackage{amsthm}
\usepackage{mathtools}
\usepackage{algorithm}
\usepackage{algcompatible}
\usepackage{thmtools,thm-restate}
\usepackage[affil-it]{authblk}
\usepackage[square,numbers,compress]{natbib}
\usepackage{nicefrac}
\usepackage{comment}
\usepackage{ifpdf}
\usepackage{bm}
\usepackage[table]{xcolor}
\usepackage{color}
\usepackage[displaymath,mathlines]{lineno}
\usepackage{scalerel}
\usepackage{stackengine,wasysym}
\usepackage{hyperref}
\usepackage[stable]{footmisc}
\usepackage{enumerate}
\usepackage{tikz-cd}
\usetikzlibrary{arrows,positioning}
\tikzset{
  shift left/.style ={commutative diagrams/shift left={#1}},
  shift right/.style={commutative diagrams/shift right={#1}}
}
\usetikzlibrary{matrix}
\usepackage{subcaption}
\usepackage{enumitem}
\usepackage{booktabs}
\usepackage{rotating}
\usepackage{soul,xcolor}
\usepackage{cancel}

\usepackage{hhline} 
\RequirePackage{etoolbox}
\RequirePackage{pgf} 
\RequirePackage{xcolor} 
\definecolor{high}{gray}{0.3}  
\definecolor{low}{gray}{0.9}   
\newcommand*{\opacity}{90}
\newcommand*{\minval}{0.0}
\newcommand*{\maxval}{1.0}
\newcommand{\gradient}[1]{
    \ifdimcomp{#1pt}{>}{\maxval pt}{#1}{
    \ifdimcomp{#1pt}{<}{\minval pt}{#1}{
         \pgfmathparse{int(round(100*(#1/(\maxval-\minval))-(\minval*(100/(\maxval-\minval)))))}
        \xdef\tempa{\pgfmathresult}
        \cellcolor{high!\tempa!low!\opacity} #1
    }}
 }

\newtheorem{thm}{Theorem}[section]

\newtheorem{lem}{Lemma}[section]
\newtheorem{cor}{Corollary}[section]
\newtheorem{prop}{Proposition}[section]

\newtheorem{dfn}{Definition}
\newtheorem{asump}{Assumption}

\newtheorem{rk}{Remark}

\newcommand{\vect}[1]{\boldsymbol{#1}}
\newcommand{\tp}[1]{{#1}^{\mathsf T}}
\renewcommand{\bar}{\overline}
\newcommand{\tr}{\mathrm{tr}}

\DeclareMathAlphabet\mathbfcal{OMS}{cmsy}{b}{n}
\DeclareMathOperator{\diag}{diag}
\newcommand{\half}{\frac12}
\newcommand{\mbP}{\mathbb P}
\newcommand{\mbE}{\mathbb E}
\newcommand{\E}{\mathrm E}

\newcommand{\mCv}{\mathrm{Cov}}

\newcommand{\eps}{\epsilon}

\renewcommand{\eps}{\varepsilon}
\renewcommand{\epsilon}{\varepsilon}
\renewcommand{\Sigma}{\varSigma}

\newcommand{\bmu}{\vect\mu}
\newcommand{\bxi}{\vect\xi}
\newcommand{\bXi}{{\boldsymbol{\Xi}}}

\newcommand{\bPhi}{{\boldsymbol{\Phi}}}
\newcommand{\bgamma}{{\boldsymbol{\gamma}}}

\newcommand{\bzeta}{\vect\zeta}
\newcommand{\Zeta}{\mathrm{Z}}
\newcommand{\bZeta}{\vect\Zeta}

\newcommand{\bzero}{{\bf 0}}
\newcommand{\bu}{{\bf u}}
\newcommand{\bU}{{\bf U}}
\newcommand{\by}{{\bf y}}
\newcommand{\bY}{{\bf Y}}
\newcommand{\bx}{{\bf x}}
\newcommand{\bX}{{\bf X}}
\newcommand{\bt}{{\bf t}}
\newcommand{\bz}{{\bf z}}
\newcommand{\bZ}{{\bf Z}}

\newcommand{\bC}{{\bf C}}
\newcommand{\bI}{{\bf I}}
\newcommand{\bL}{{\bf L}}

\newcommand{\mC}{{\mathcal C}}
\newcommand{\mG}{{\mathcal G}}
\newcommand{\mK}{\mathcal{K}}
\newcommand{\mH}{\mathcal{H}}
\newcommand{\mI}{{\mathcal I}}

\newcommand{\mN}{{\mathcal N}}

\newcommand{\mS}{{\mathcal S}}
\newcommand{\mX}{{\mathcal X}}
\newcommand{\mT}{{\mathcal T}}
\newcommand{\mZ}{{\mathcal Z}}

\newcommand{\mbX}{{\mathbb X}}
\newcommand{\mbY}{{\mathbb Y}}
\newcommand{\mbT}{{\mathbb T}}

\newcommand{\mbN}{{\mathbb N}}
\newcommand{\mbQ}{{\mathbb Q}}
\newcommand{\mbR}{{\mathbb R}}

\newcommand{\qED}{\mathrm{q}\!-\!\mathrm{ED}}

\newcommand{\qEP}{\mathrm{q}\!-\!\mathcal{EP}}
\newcommand{\ep}{\mathrm{EP}}

\newcommand{\GP}{\mathcal{GP}}
\newcommand{\mB}{{\mathcal B}}
\newcommand{\STBP}{\mathcal{STBP}}
\newcommand{\STGP}{\mathcal{STGP}}

\ifpdf
   \graphicspath{{./figure/PNG/}{./figure/PDF/}{./figure/}}
\else
   \graphicspath{{./figure/EPS/}{./figure/}}
\fi

\begin{document}

\def\spacingset#1{\renewcommand{\baselinestretch}%
{#1}\small\normalsize} \spacingset{1}


\if1\blind
{
  \title{\bf Spatiotemporal Besov Priors for Bayesian Inverse Problems}
  \author[1]{Shiwei Lan \thanks{The corresponding author. Email: \texttt{slan@asu.edu}.}}
  \author[2]{Mirjeta Pasha}
  \author[1]{Shuyi Li}
  \author[3]{Weining Shen}
  \affil[1]{School of Mathematical \& Statistical Sciences, Arizona State University}
  \affil[2]{Department of Mathematics, Virginia Polytechnic and State University}
  \affil[3]{Department of Statistics, University of California - Irvine}
    \date{}
  \maketitle
} \fi

\if0\blind
{
  \bigskip
  \bigskip
  \bigskip
  \begin{center}
    {\LARGE\bf Spatiotemporal Besov Priors for Bayesian Inverse Problems}
\end{center}
  \medskip
} \fi

\bigskip
\begin{abstract}
Fast development in science and technology has driven the need for proper statistical tools to capture special data features such as abrupt changes or sharp contrast.
Many inverse problems in data science require spatiotemporal solutions derived from a sequence of time-dependent objects with these spatial features, 
e.g., the dynamic reconstruction of computerized tomography (CT) images with edges.
Conventional methods based on Gaussian processes (GP) often fall short in providing satisfactory solutions since they tend to offer oversmooth priors. Recently, the Besov process (BP), defined by wavelet expansions with random coefficients, has emerged as a more suitable prior for Bayesian inverse problems of this nature.  While BP excels in handling spatial inhomogeneity, it does not automatically incorporate temporal correlation inherited in the dynamically changing objects.
In this paper, we generalize BP to a novel spatiotemporal Besov process (STBP) by replacing the random coefficients in the series expansion with stochastic time functions as Q-exponential process (Q-EP) which governs the temporal correlation structure.
We thoroughly investigate the mathematical and statistical properties of STBP.  
Simulations, two limited-angle CT reconstruction examples, a highly non-linear inverse problem involving Navier-Stokes equation, and a spatiotemporal temperature imputation problem are used to demonstrate the advantage of the proposed STBP 
compared with the classic STGP and a time-uncorrelated approach.
\end{abstract}

\noindent%

\spacingset{1.67} 

\section{Introduction}
Many modern science and engineering applications are presented as inverse problems whose main goal is to recover parameters of interest from observed data. These data may possess inhomogeneity in the sense that certain portions differ from others significantly. For example, some medical images exhibit sharp edges where properties undergo dramatic changes. Furthermore, these complex datasets may be spatiotemporal, with solutions extending across both space and time. One of the key challenges in solving these types of inverse problems is to  effectively capture the distinctive characteristics of high-dimensional (potentially infinite-dimensional) objects using limited data. Surging need has been posed for statistical methodology to appropriately impose regularization or fill in prior information for these ill-posed inverse problems in order to construct meaningful solutions. 

In nonparametric statistics, Gaussian process (GP) \citep{Rasmussen_2005} 
has been widely used as an $L_2$ penalty 
or a prior on the function space. However, despite their flexibility, random functions generated from GPs often exhibit excessive smoothing, which is not ideal for modeling heterogeneous objects such as images with sharp edges (See Figure \ref{fig:simulation_MAP_whiten_q} for illustration). 
To address this issue, researchers have proposed a class of $L_1$ penalty based priors including Laplace random field \citep{
KOZUBOWSKI_2013,
Atchade_2017} and Besov process (BP) \citep{Lassas_2009,Dashti_2012,
dashti2017}. 
There are also many heavy-tailed priors 
such as Cauchy \citep{
Suuronen2022}, total variation (TV) \citep{
Yao_2016}, and those constructed by normal variance mixture \citep{
Bolin_2020}, and data-informed priors based on level set functions \citep{
dunlop2016map} 
proposed for handling inhomogeneity.
These approaches have found significant applications in 
signal processing \citep{
KOZUBOWSKI_2013}, imaging analysis \citep{
Pasha_2023} 
and inverse problems \citep{Dashti_2012}.

In spatiotemporal modeling, GP (STGP) has long been used as a flexible prior to capture space-time interactions. A large class of models bear a non-separable kernel structure constructed by parametric functions \citep{Cressie_1999, Gneiting_2002}, spectral representation \citep{Fuentes_2008}, kernel convolution \citep{marco2015,Wang_2020} or mixing \citep{Fonseca_2011}, and nonparametric hierarchical modeling \citep{Zhang_2020, Lan_2022}. While capable of characterizing the spatiotemporal relationship in data, these GP based priors tend to oversmooth spatial features due to their $L_2$ nature.

On the other hand, most of the sparsity-promoting and edge-preserving priors in the literature work well in characterizing spatial inhomogeneity, but few are tailored to specifically address spatiotemporal targets and their temporal correlations. 
In this paper, we focus on the BP proposed by \cite{Lassas_2009} for imaging analysis and generalize it to the spatiotemporal domain. Historically, \cite{lassas2004can} discovered that the TV prior degenerates to a Gaussian prior as the discretization mesh becomes denser and loses the edge-preserving properties in high-dimensional applications. Therefore, \cite{Lassas_2009} proposed the BP prior defined using wavelet basis and random coefficients following a (univariate) $q$-exponential distribution and proved its discretization-invariant property. 
Recently, \cite{Li_2023} introduced a stochastic process based on a consistent multivariate generalization of the $q$-exponential distribution, hence named the $Q$-exponential process (Q-EP). The Q-EP can be viewed as an explicit probabilistic definition of BP with direct control on the correlation structure and tractable prediction formula. 
In this paper, we propose a novel \emph{spatiotemporal Besov process (STBP)} by replacing the (univariate) $q$-exponential random coefficients in the series definition of BP with stochastic time functions as Q-EP.
The proposed STBP offers a flexible prior in modeling functional data with spatial features while controlling the temporal correlations explicitly through a covariance kernel. Similarly as BP, STBP also includes spatiotemporal GP (STGP) as a special case for $q=2$ (See Figure \ref{fig:relationship} for their relationship).
Since our motivation is to model heterogeneous data with priors imposing sharp $L_q$ regularization, we focus on $1\leq q\leq 2$ in this paper (See Figure \ref{fig:simulation_MAP_whiten_q} for the regularization effect of parameter $q$).
To the best of our knowledge, this is by far the first spatiotemporal generalization of BP. 

To justify STBP as a working prior for spatiotemporal inverse problems, Bayes theorem in this setting is re-examined based on \cite{Dashti_2012,dashti2017}. With proper assumptions on the likelihood, the posterior contraction theorems for Bayesian inverse models with STBP priors are established based on \cite{ghosal2007,vanderVaart08,Agapiou_2021}. Posterior contraction properties with Gaussian priors have been extensively studied by \cite{vanderVaart08,vandervaart09,
Ghosal_2017} for regression and classification, density estimation, white noise models, and Bayesian linear inverse problems \citep{knapik2011bayesian} and nonlinear inverse problems \citep{Vollmer_2013}. 
\cite{Gine_2011} also studied the posterior contraction of density estimation for a class of $L^r$-metrics ($1\leq r\leq\infty$) based priors  including GP, wavelet series and normal mixture.
\cite{Agapiou_2021} studied the posterior contraction theorems for a re-branded BP named $p$-exponential process (which only differs from BP by a constant in the univariate $q$-exponential distribution) for density estimation and white noise model.
Other works on the posterior contraction of Besov-type priors include \cite{Rivoirard_2012,
Gao_2020}. 
Compared with the existing literature, our theoretic results are novel in terms of: 1) generalization to spatiotemporal models for inverse problems; 2) simplified contraction rates given in a class of Besov-type spaces contained in $L^q$ spaces.

To facilitate Bayesian inference for models with STBP priors, we introduce a novel white noise representation \citep{chen2018dimension} of the random function drawn from STBP and take advantage of dimension-independent MCMC algorithms \citep{beskos2017} for efficient implementation.
The numerical advantages of the proposed STBP over STGP have been supported by multiple experiments in dynamic CT reconstruction, highly nonlinear inverse problems, and spatiotemporal imputation.
Our proposed work on STBP has multiple contributions to the literature of spatiotemporal inverse problems:
\begin{enumerate}[itemsep=1pt]
\item It generalizes BP to the spatiotemporal domain to simultaneously model the spatial inhomogeneity and the temporal correlations.
\item It provides theoretic characterization on the posterior contraction in the infinite data limit, justifying its validity as a nonparametric learning tool.
\item It demonstrates utility in spatiotemporal modeling inhomogeneous data (dynamic CT reconstruction) and indicates broader impact on imaging analysis.
\end{enumerate}

The rest of the paper is organized as follows. Section \ref{sec:bkgd} provides a background review on the Bayesian inverse problems and BP used as a flexible edge-preserving prior. Section \ref{sec:QEP} introduces Q-EP as random coefficient functions on the time domain. We then formally define STBP and study its theoretic properties in Section \ref{sec:STBP}. In Section \ref{sec:inference} we describe a white noise representation of STBP that facilitates the inference for models with STBP prior. In Section \ref{sec:numerics} we demonstrate the advantage of the proposed STBP prior in retaining spatial features and capturing temporal correlations for the spatiotemporal inverse problems using a simulated regression, two dynamic CT reconstruction examples, a nonlinear inverse problem involving Navier-Stokes equation, and a spatiotemporal temperature imputation. Finally we conclude with some discussion on future research in Section \ref{sec:conclusion}.

\section{Background on Besov priors for Inverse Problems}\label{sec:bkgd}
The Bayesian approach to inverse problems \citep{dashti2017} has gained increasing popularity because it provides a natural framework for model calibration and uncertainty quantification (UQ). In this section, we review some background about Bayesian inverse problems and BP as a flexible prior for modeling objects with spatial features.

\subsection{Bayesian Inverse Problems}
We consider the inverse problem of recovering an unknown parameter $u\in \mbX$ from a noisy observation $y\in \mbY$ based on the following Bayesian model
\begin{equation}\label{eq:bip}
\begin{aligned}
y &= \mG(u) + \eta , \quad \eta \sim \mbQ_0 ,\\
u &\sim \Pi ,
\end{aligned}
\end{equation}
where both $\mbX$ and $\mbY$ are separable Banach spaces,
$\mG:\mbX \to \mbY$ is a forward mapping from the parameter space $\mbX$ to the data space $\mbY$, and $\eta\in \mbY$ denotes the random noise whose distribution $\mbQ_0$ is independent from the prior $\Pi$.
We assume the conditional $y|u$ is distributed according to the measure $\mbQ_u\ll \mbQ_0$ for $u$, $\Pi$-almost surely (a.s.), and hence define the potential (negative log-likelihood) function $\Phi:\mbX\times\mbY\to\mbR$:
\begin{equation}\label{eq:potential}
\frac{d\mbQ_u}{d\mbQ_0}(y) = \exp(-\Phi(u;y)).
\end{equation}
The objective of Bayesian inverse problem is to seek the posterior solution of $u|y$ whose distribution, denoted as $\Pi(\cdot|y)$, according to the Bayes' theorem \citep{
dashti2017}, satisfies the requirement that if $0< Z:=\int_{\mbX} \exp(-\Phi(u; y)) \Pi(d u) < +\infty$ for $y$ $\mbQ_0$-a.s., then
\begin{equation} \label{eq:Bayes}
\frac{d\Pi(\cdot|y)}{d\Pi}(u) = \frac{1}{Z}\,\exp(-\Phi(u; y)). 
\end{equation}
The forward operator $\mG$ could be linear or nonlinear, possibly encoding physical information represented by a system of ordinary or partial differential equations (ODE/PDE). The resulted posterior $\Pi(\cdot|y)$ is usually non-Gaussian with a complicated geometric structure even if a Gaussian prior $\Pi=\GP(0,\mC)$ is adopted. 
When there are sparse data but the targets are high-dimensional, the inverse problems are ill-posed.
Proper prior information is crucial to induce well-defined solutions. 

\subsection{Besov Process}\label{sec:Besov}
Let $\mX\subset\mbR^d$ be the spatial domain, e.g., a $d$-dimensional torus, $\mX=\mbT^d=(0,1]^d$ for $d\leq3$.
Consider a separable Banach space $(\mbX, \Vert\cdot\Vert)$ with a Schauder basis $\{\phi_\ell\}_{\ell=1}^\infty$, e.g., a square integrable function space $L^2(\mX):=\{u:\mX \to\mbR | \int_\mX|u(\bx)|^2 d\bx<\infty\}$ with Fourier basis. Any function $u\in \mbX$ can then be represented by the following series:
\begin{equation}\label{eq:series_repn}
    u(\bx) = \sum_{\ell =1}^{\infty}u_\ell\phi_\ell(\bx).
\end{equation}
Based on \eqref{eq:series_repn}, we consider a norm $\|\cdot\|_{s,q}$ defined with a smoothness parameter $s>0$ and an integrability parameter $q\geq 1$ \citep{Lassas_2009, Dashti_2012}:
\begin{equation}\label{eq:Besov_norm}
    \Vert u(\cdot) \Vert_{s,q} = \left(\sum_{\ell=1}^\infty\ell^{\tau_q(s)q}|u_\ell|^q\right)^{\frac{1}{q}}, \quad \tau_q(s) = \frac{s}{d}+\half-\frac{1}{q}.
\end{equation}
We define the Banach space $B^{s,q}(\mX):=\{u: \mX\to\mbR \,|\, \Vert u(\cdot)\Vert_{s,q}<\infty\}$.
If $\{\phi_\ell\}_{\ell=1}^{\infty}$ is an $r$-regular wavelet basis for $r > s$, then $B^{s,q}(\mX)$ becomes the Besov space $B^s_{qq}$ \citep{Triebel_1983}. 
In particular, if $q=2$ and $\{\phi_\ell\}_{\ell=1}^\infty$ form the Fourier basis, then $B^{s,2}(\mX)$ reduces to the Sobolev space $H^s(\mX)$ with the special case $B^{0,2}(\mX)=L^2(\mX)$ assuming $s=0$.

For a given basis $\{\phi_\ell\}_{\ell=1}^\infty$, based on the the series expansion \eqref{eq:series_repn}, there is a one-to-one correspondence between the $B^{s,q}(\mX)$ function $u(\cdot)$ and the infinite sequence $u:=\{u_\ell\}_{\ell=1}^\infty$ in a weighted $\ell^q$ space, $\ell^{q,\tau}:=\{u\in\mbR^\infty \,|\, \Vert u\Vert_{\tau, q}=\left(\sum_{\ell=1}^\infty\ell^{\tau q}|u_\ell|^q\right)^{\frac{1}{q}} <\infty\}$, which reduces to the regular $\ell^q$ space when $\tau=0$.
Hence $\Vert u(\cdot)\Vert_{s,q}=\Vert u\Vert_{\tau_q(s),q}$.
In the following, we will use $u$ to refer to both notations when there is no confusion.

Now we define a \emph{Besov process (BP)} $u(\cdot)$ based on \eqref{eq:series_repn} 
by randomizing the coefficients $\{u_\ell\}_{\ell=1}^\infty$.
More specifically, we set for $\ell\in\mbN$
\begin{equation}\label{eq:q_exp}
 u_\ell := \gamma_\ell \xi_\ell, \quad \gamma_\ell = \kappa^{-\frac{1}{q}} \ell^{-\tau_q(s)}, \quad  \xi_\ell\overset{i.i.d.}{\sim}\pi_{\xi}(\cdot) \propto \exp\left(- \half|\cdot|^q\right),
\end{equation}
where $\kappa>0$ is a scaling factor,
and $\pi_{\xi}$ denotes the probability density function of the \emph{$q$-exponential distribution} \citep{Dashti_2012,dashti2017}.
Though not spelled out, such $q$-exponential distribution is actually a special case of the following exponential power (EP, a.k.a. generalized normal) distribution $\ep(\mu, \sigma, q)$ with $\mu=0$, $\sigma=1$:
\begin{equation}\label{eq:epd}
    p(\xi|\mu, \sigma, q) = \frac{q}{2^{1+1/q}\sigma\Gamma(1/q)}\exp\left\{-\half \left|\frac{\xi-\mu}{\sigma}\right|^q\right\}.
\end{equation}
When $q=2$, this is just a normal distribution $\mN(\mu,\sigma^2)$. When $q=1$, it becomes a Laplace distribution $L(\mu, b)$ with $\sigma=2^{-1/q} b$.

Denote infinite sequences $\gamma=\{\gamma_\ell\}_{\ell=1}^\infty$ and $\xi=\{\xi_\ell\}_{\ell=1}^\infty$. 
Then $\xi$ is a random element of the probability space $(\Omega, \mB(\Omega), \mbP)$ with $\Omega=\mbR^\infty$, product $\sigma$-algebra $\mB(\Omega)$ and probability measure $\mbP$ defined by extending the finite product of $\pi_\xi$ to the infinite product by the Kolmogorov extension theorem \citep[c.f. Theorem 29 in section A.2.1 of][]{dashti2017}. Then we define the \emph{Besov measure} as the pushforward of $\mbP$ as follows.
\begin{dfn}[Besov Measure]
Let $\mbP$ be the measure of random sequences $\xi\in\Omega$. 
Suppose we have the following map
\begin{equation}\label{eq:fun_expn}
f_\gamma: \Omega \to B^{s,q}(\mX), \quad \xi \mapsto u = \sum_{\ell=1}^\infty u_\ell\phi_\ell = \sum_{\ell =1}^\infty \gamma_\ell \xi_\ell \phi_\ell, 
\end{equation}
where $\gamma_\ell$ and $\xi_\ell$ are defined in \eqref{eq:q_exp}.
Then the pushforward $f_\gamma^\sharp\mbP$ is Besov measure on $B^{s,q}(\mX)$, denoted as $\mB(\kappa, B^{s,q}(\mX))$, and we say BP $u$ follows the Besov measure, i.e., $u\sim \mB(\kappa, B^{s,q}(\mX))$.
\end{dfn}
\begin{rk}
If $q = 2$ and $\{\phi_\ell\}_{\ell=1}^{\infty}$ is either a wavelet or Fourier basis, we obtain a Gaussian measure with the Cameron-Martin space $B_{22}^s$ \citep{
Triebel_1983}, which is the Hilbert space $H^s(\mX)$, and 
\eqref{eq:fun_expn} is reduced to a zero mean Gaussian random element based on Karhunen-Lov\`eve representation (refer to Figure \ref{fig:relationship}):
    $u(\bx) = \kappa^{-\half} \sum_{\ell =1}^\infty \ell^{-\frac{s}{d}}\xi_\ell \phi_\ell(\bx), \, \xi_\ell \overset{iid}{\sim} \mN(0,1)$.
\end{rk}

Define $\Vert \xi\Vert_q := \left(\sum_{\ell=1}^\infty |\xi_\ell|^q\right)^{\frac{1}{q}}$ for $\xi\in\Omega$. Then we have $\Vert u\Vert_{s,q} = \kappa^{-\frac{1}{q}}\Vert \xi\Vert_q$. 
The following formal Lebesgue density can be made rigorous by Fernique theorem \citep{Dashti_2012}:
\begin{equation}\label{eq:lebesgue}
\mbP(d\xi) = p(\xi) d\xi, \quad p(\xi) = \prod_{\ell=1}^\infty \pi_\xi(\xi_\ell) \propto \exp\left\{-\half\Vert \xi\Vert_q^q\right\} = \exp\left\{-\frac{\kappa}{2}\Vert u\Vert_{s,q}^q\right\}.
\end{equation}
When used in the optimization to obtain a parameter estimate, the logarithm of Besov prior density \eqref{eq:lebesgue} serves as an $L_q$ regularization term. Larger regularization parameter ($q>0$) promotes smoother solutions, as demonstrated in a simulated regression example in Figure \ref{fig:simulation_MAP_whiten_q}. In practice, Besov prior is often adopted for $q=1$ to preserve edges in imaging analysis.

Define the Banach space of $q$-integrable functions in $(\Omega, \mB(\Omega), \mbP)$ as $L^q_\mbP(\Omega; B^{s,q}(\mX)) = \{u: \mX\times \Omega\to \mbR | \mbE(\Vert u\Vert_{s,q}^q)<\infty \}$.
We notice that for $u \sim \mB(\kappa, B^{s,q}(\mX))$, $u\notin L^q_\mbP(\Omega; B^{s,q}(\mX))$ because $\mbE(\Vert u\Vert_{s,q}^q)=\kappa^{-1}\mbE(\Vert \xi\Vert_q^q)=\infty$ due to the iid assumption on $\xi$ in \eqref{eq:q_exp}.
However, the following theorem states that Besov random draw as in \eqref{eq:fun_expn} has limit in a proper $q$-integrable function space \citep[Thorem 4 of][]{dashti2017}. 
\begin{restatable}{thm}{Lqbound}
\label{thm:Lqbound}
    If $u \sim \mB(\kappa, B^{s,q}(\mX))$ as in \eqref{eq:fun_expn}, then $u\in L^q_\mbP(\Omega; B^{s',q}(\mX))$ for all $s'<s-\frac{d}{q}$.
\end{restatable}
\begin{proof}
See Supplement \ref{apx:Lqbound}.
\end{proof}
\begin{rk}
This theorem implies that for any random draw from $\mB(\kappa, B^{s,q}(\mX))$, we need to consider its $q$-integrability in a larger ambient space $B^{s',q}(\mX)$ for some $s'<s-\frac{d}{q}$ ($B^{s,q}(\mX)\subset B^{s',q}(\mX)$, refer to Proposition \ref{prop:embeddings}).
\end{rk}

Next, we generalize the series representation of a Besov random function \eqref{eq:fun_expn} to a representation for STBP by replacing the random variable $\xi_\ell$ with a stochastic process $\xi_\ell(\cdot)$ on the temporal domain $\mT\subset \mbR_+$.
For this purpose, in the following we will first introduce a properly defined process Q-EP for $\xi_\ell(\cdot)$ that generalizes the $q$-exponential random variable $\xi_\ell$ and has the capability of capturing the temporal dependence in data.

\section{$Q$-exponential Process Valued Random Coefficients}\label{sec:QEP}
\subsection{Multivariate Generalization of $Q$-exponential Distribution}

To generalize the aforementioned $q$-exponential \eqref{eq:q_exp} (or univariate EP \eqref{eq:epd}) random variable, $\xi_\ell$, to a multivariate random vector, $\bxi_\ell$, and further a stochastic process, $\xi_\ell(\cdot)$, we have two important requirements by the Kolmogorov' extension theorem \citep{Oksendal_2003}: i) {\bf exchangeability} of the joint distribution, i.e., $p(\bxi_{1:J}) = p(\bxi_{\tau(1:J)})$ for any finite permutation $\tau$; and ii) {\bf consistency} of the marginalization, i.e., $p(\bxi_1) = \int p(\bxi_1, \bxi_2)d\bxi_2$.

Consider the process $\xi_\ell(\cdot)$ defined in a finite temporal domain $\mT\subset \mbR_+$, i.e., there exists $T<\infty$ such that $\mT\subset [0,T]$.
Suppose we observe $\xi_\ell(t)$ at $J$ time points, $t_1,\cdots,t_J\in \mT$, then we need to define the distribution of $\bxi_\ell=(\xi_\ell(t_1),\cdots, \xi_\ell(t_J))$.
\cite{Li_2023} investigate the family of elliptic contour distributions \citep{Johnson_1987} 
and propose the following consistent \emph{multivariate $q$-exponential distribution} for $\bxi_\ell$. 
\begin{dfn}\label{dfn:qED}
    A multivariate $q$-exponential distribution, $\qED_J(\bmu, \bC)$, has the density
    \begin{equation}\label{eq:qED}
        p(\bxi|\bmu, \bC, q) = \frac{q}{2} (2\pi)^{-\frac{J}{2}} |\bC|^{-\half} r^{(\frac{q}{2}-1)\frac{J}{2}} \exp\left\{-\frac{r^\frac{q}{2}}{2}\right\}, \quad r(\bxi) = \tp{(\bxi-\bmu)} \bC^{-1} (\bxi-\bmu).
    \end{equation}
\end{dfn}

\cite{Li_2023} prove 
that the multivariate $q$-exponential random vector following distribution \eqref{eq:qED} satisfies the conditions of Kolmogorov's extension theorem (both {\bf exchangeability} and {\bf consistency}) \citep[Theorem 3.3 of][]{Li_2023} hence can be generalized to a stochastic process.

To generate random vectors $\bxi \sim \qED_J(\bmu, \bC)$, one can take advantage of the stochastic representation, 
as defined in the following proposition \citep[c.f. Theorem 2.1 and Proposition A.1][]{Li_2023}. This will be needed for the Bayesian inference in Section \ref{sec:inference}.
\begin{prop}\label{prop:QED_stochrep}
    If $\bxi \sim \qED_J(\bmu, \bC)$, then we have
    \begin{equation}\label{eq:stoch_QED}
        \bxi = \bmu + R \bL S, 
    \end{equation}
where $S\sim \mathrm{Unif}(\mS^{J+1})$ uniformly distributed on the unit-sphere $\mS^{J+1}$, $\bL$ is the Cholesky factor of $\bC$ such that $\bC=\bL\tp{\bL}$, $R\perp S$ and $R^q \overset{d}{=} r^\frac{q}{2} \sim \Gamma\left(\alpha=\frac{J}{2}, \beta=\half \right) = \chi^2(J)$.
\end{prop}

\subsection{$Q$-exponential Process}

With a covariance (symmetric and positive-definite) kernel $\mC : \mT\times \mT\to \mbR$,
we define \emph{$q$-exponential process (Q-EP)} based on the multivariate $q$-exponential distribution \eqref{eq:qED}. 
\begin{dfn}\label{dfn:qEP}
    A (centered) $q$-exponential process $\xi(t)$ with kernel $\mC$, $\qEP(0, \mC)$, is a collection of random variables such that any finite set, $\bxi:=(\xi(t_1),\cdots, \xi(t_J))$, 
    follows a multivariate $q$-exponential distribution $\qED_J(\bzero, \bC)$, where $\bC=[\mC(t_j,t_{j'})]_{J\times J}$. 
\end{dfn}
\begin{rk}
When $q=2$, $\qED_J(\bmu, \bC)$ reduces to $\mN_J(\bmu, \bC)$ and thus $\qEP(0, \mC)$ becomes $\GP(0,\mC)$. When $q\in[1,2)$, $\qEP(0, \mC)$ lends flexibility to modeling functional data with more regularization than GP.
See Figures \ref{fig:relationship} and \ref{fig:simulation_MAP_whiten_q}, and more details in Section \ref{sec:numerics}.
\end{rk}

The covariance kernel $\mC$ is associated with a Hilbert-Schmidt (HS) integral operator $T_\mC: L^2(\mT) \to L^2(\mT), \xi(\cdot)\mapsto\int_\mT \mC(\cdot, t') \xi(t')\mu(dt')$ which has eigen-pairs $\{\lambda_\ell, \psi_\ell(\cdot)\}_{\ell=1}^\infty$ such that for $\forall\ell\in\mbN$, $T_\mC \psi_\ell(t)=\psi_\ell(t) \lambda_\ell$ and $\Vert\psi_\ell\Vert_2=1$. Then $\{\psi_\ell\}_{\ell=1}^\infty$ serves as a basis of $L^2(\mT)$.
Denote $\lambda:=\{\lambda_\ell\}_{\ell=1}^\infty$.
We assume $T_\mC$ is a trace-class operator, i.e. $\tr(T_\mC):=\Vert \lambda\Vert_1=\sum_{\ell=1}^\infty \lambda_\ell<\infty$.
\cite[Theorem 3.4 of][]{Li_2023} shows that we have a Karhunen-Lo\'eve type of theorem on the series representation of random function $\xi(\cdot)$ drawn from Q-EP.
\begin{thm}[Karhunen-Lo\'eve]\label{thm:KL}
If $\xi(\cdot)\sim \qEP(0,\mC)$ with a trace-class HS operator $T_\mC$ having eigen-pairs $\{\lambda_\ell, \psi_\ell(\cdot)\}_{\ell=1}^\infty$, then we have the following series representation for $\xi(t)$:
\begin{equation}\label{eq:KL}
\xi(t) = \sum_{\ell=1}^\infty \xi_\ell \psi_\ell(t), \quad \xi_\ell:=\int_\mT \xi(t) \psi_\ell(t) \mu(dt) \sim \qED(0, \lambda_\ell),
\end{equation}
where $\E[\xi_\ell]=0$ and $\mCv(\xi_\ell, \xi_{\ell'})=\lambda_\ell\delta_{\ell\ell'}$ with Dirac function $\delta_{\ell\ell'}=1$ if $\ell=\ell'$ and $0$ otherwise.
Moreover, we have
$\E[\Vert \xi(\cdot)\Vert_2^2] = \sum_{\ell=1}^\infty \E[\xi_\ell^2] = 
\tr(T_\mC) < \infty$.
\end{thm}
\begin{rk}
By re-scaling $\xi_\ell$ in \eqref{eq:KL}, we have the series representation of Q-EP $\xi(\cdot)$ in the same format as BP in \eqref{eq:q_exp}:
 $\xi_\ell := \gamma_\ell \xi_\ell^\star, \quad \gamma_\ell = \sqrt{\lambda_\ell}, \quad \xi_\ell^\star \overset{iid}{\sim} \qED(0,1) \sim \pi_{\xi}(\cdot)$. 
If we choose $\sqrt{\lambda_\ell}=\ell^{-\tau_q(s)}$, 
then $\qEP(0, \mC)$ process becomes equivalent to $\mB(1, B^{s,q}(\mT))$ process.
From this perspective, we can view Q-EP as a probabilistic definition of BP (See Figure \ref{fig:relationship}).
\end{rk}


Theorem \ref{thm:KL} states that for the given basis $\{\psi_\ell\}_{\ell=1}^\infty$ on the time domain $\mT$, we can identify any random draw $\xi(\cdot) \sim \qEP(0,\mC)$ with the associated infinite sequence $\xi=\{\xi_\ell\}_{\ell=1}^\infty$ as in \eqref{eq:KL}.
Similarly as in Section \ref{sec:Besov}, we can define $\Vert \xi(\cdot)\Vert_{s,q} = \Vert \xi\Vert_{\tau_q(s),q} = \left(\sum_{\ell=1}^\infty \ell^{\tau_q(s)q}|\xi_\ell|^q\right)^{\frac{1}{q}}$ and 
$B^{s,q}(\mT) = \{\xi: \mT\to\mbR \,|\, \Vert \xi(\cdot)\Vert_{s,q}<\infty\}$.
In the probability space $(\Omega, \mB(\Omega), \mbP)$ with $\Omega=\mbR$ and $\mbP$ defined by $\qEP(0, \mC)$, we consider a Banach space defined as $L^p_\mbP(\Omega, L^p(\mT))=\{\xi:\mT\times\Omega\to\mbR \,|\, \mbE(\Vert \xi\Vert_p^p)<\infty\}$.
From Theorem \ref{thm:KL}, we immediately have that if $\xi(\cdot) \sim \qEP(0,\mC)$ with a trace-class HS operator $T_\mC$, then $\xi(\cdot) \in L^2_\mbP(\Omega, L^2(\mT))$.
More general integrability of $\xi(\cdot)$ relates to the summability of eigenvalues $\lambda$ of the HS operator $T_\mC$, as expressed in the following assumption.
\begin{asump}\label{asmp:eigen_summable}
Suppose $\lambda=\{\lambda_\ell\}_{\ell=1}^\infty$ are eigenvalues of HS operator $T_\mC$ for the kernel $\mC$ 
in Definition \ref{dfn:qEP}. We assume
\begin{enumerate}[label=(\roman*),nosep]
\item 
    $\sqrt{\lambda}\in \ell^{q,\tau_q(s')}, \quad i.e. \quad \Vert\sqrt{\lambda}\Vert_{\tau_q(s'),q}^q = \sum_{\ell=1}^\infty \ell^{\tau_q(s')q} \lambda_\ell^{\frac{q}{2}}<\infty$ for $s'<s-\frac{d}{q}$.
\item 
    $\lambda\in \ell^{\frac{q}{2}}, \quad i.e. \quad \Vert\lambda\Vert_{\frac{q}{2}}^{\frac{q}{2}}= \sum_{\ell=1}^\infty \lambda_\ell^{\frac{q}{2}}<\infty$.
\end{enumerate}
\end{asump}
\begin{restatable}{thm}{qepint}
\label{thm:qepint}
If $\xi(\cdot) \sim \qEP(0,\mC)$ with a trace-class HS operator $T_\mC$ 
satisfying Assumption \ref{asmp:eigen_summable}-(i),
then $\xi(\cdot)\in L^q_\mbP(\Omega, B^{s',q}(\mT))$.
If Assumption \ref{asmp:eigen_summable}-(ii) holds instead, then $\xi(\cdot)\in L^q_\mbP(\Omega, L^q(\mT))$ and in particular, $\mbE[\Vert \xi(\cdot)\Vert_{q}^q] = \Vert\lambda\Vert_{\frac{q}{2}}^{\frac{q}{2}} <\infty$.
\end{restatable}
\begin{proof}
See Supplement \ref{apx:qepint}.
\end{proof}


Under Assumption \ref{asmp:eigen_summable}-(i), e.g., $\sqrt{\lambda_\ell}=\ell^{-\tau_q(s)}$, the $\qEP(0, \mC)$ process becomes equivalent to the BP, $\mB(1, B^{s,q}(\mT))$, which only differs from the \emph{$p$-exponential process} \citep{Agapiou_2021} by a constant in the definition $\pi_\xi(\cdot)\propto \exp\left(-\frac{1}{p}|\cdot|^p\right)$. See Figure \ref{fig:relationship} for more illustration of their relationship.
Therefore, the posterior concentration theory \citep[Theorem 3.1 and Lemma 5.14 of][]{Agapiou_2021} developed for the $p$-exponential process applies to the $\qEP(0, \mC)$ process.
This result (see Theorem \ref{thm:p-exp_contr} in Supplement \ref{apx:posthm}) will be used in the proof of posterior contraction Theorem \ref{thm:postcontr} for STBP priors.

\section{Spatiotemporal Besov Process}\label{sec:STBP}
Now we generalize the Banach space $B^{s,q}(\mX)$ to include the temporal domain.
Let the coefficients $\{u_\ell\}_{\ell=1}^\infty$ in \eqref{eq:series_repn} be $L^p(\mT)$ functions over some bounded temporal domain $\mT\subset \mbR_+$. 
Denote $\mZ=\mX\times \mT$ and $\bz=(\bx, t)$. Then we obtain a spatiotemporal function $u(\bz)=u(\bx, t)$ on $\mZ$ by the following series expansion with an infinite sequence of $L^p(\mT)$ functions:
\begin{equation}\label{eq:STseries_repn}
    u(\bz) = \sum_{\ell =1}^{\infty}u_\ell(t)\phi_\ell(\bx),\quad u_\ell(\cdot)\in L^p(\mT),\; \forall\ell\in\mbN.
\end{equation}
Denote the infinite sequence $u_\mT:=\{u_\ell(\cdot)\}_{\ell=1}^\infty$. We define the following norm $\Vert \cdot \Vert_{\tau,q,p}$ for $u_\mT$ with a spatial (BP) index $q\geq 1$ and a temporal (Q-EP) index $p\geq 1$: 
\begin{equation}\label{eq:STBSV_norm}
    \Vert u_\mT \Vert_{\tau,q,p} = \left(\sum_{\ell=1}^\infty\ell^{\tau q}\Vert u_\ell(\cdot)\Vert^q_p\right)^{\frac{1}{q}},
    \quad \Vert u_\ell(\cdot)\Vert_p = \left(\int_\mT |u_\ell(t)|^p dt\right)^{\frac{1}{p}}.
\end{equation}
Denote the space of such infinite sequences as $\ell^{q,\tau}(L^p(\mT)) := \{u_\mT \,|\, \Vert u_\mT \Vert_{\tau,q,p} < \infty \}$.
For a fixed spatial basis $\{\phi_\ell(\bx)\}_{\ell=1}^\infty$, 
we can identify $u$ with $u_\mT$ based on the series representation \eqref{eq:STseries_repn}. 
Let $\Vert u \Vert_{s,q,p}=\Vert u_\mT \Vert_{\tau_q(s),q,p}$ with $\tau_q(s)=\frac{s}{d}+\half-\frac{1}{q}$ as in \eqref{eq:Besov_norm}.
Then we define the Banach space of spatiotemporal functions $B^{s,q,p}(\mZ) := \{u : \mZ \to \mbR \,|\, \Vert u \Vert_{s,q,p} < \infty\}$. 

Next we generalize BP $u(\bx)\sim \mB(\kappa, B^{s,q}(\mX))$ as in \eqref{eq:fun_expn} to be spatiotemporal by letting the random coefficients $\{\xi_\ell\}_{\ell=1}^\infty$ vary in time domain according to $\qEP(0,\mC)$. For this purpose, we make the following assumption.
\begin{asump}\label{asmp:decaying_constant}
In \eqref{eq:STseries_repn}, we let
\begin{equation}\label{eq:q_EP}
    u_\ell(t) = \gamma_\ell \xi_\ell(t), \quad \gamma_\ell = \kappa \ell^{-\tau_q(s)}, \quad \xi_\ell(\cdot)\overset{\text{i.i.d.}}{\sim} \qEP(0,\mC).
\end{equation}
\end{asump}
Compared with \eqref{eq:q_exp}, we absorb the scaling factor $\kappa>0$ into the covariance kernel $\mC$ and set $\kappa=1$ 
except in Section \ref{sec:post_contr}.
Under Assumption \ref{asmp:decaying_constant}, we have $u$ in \eqref{eq:STseries_repn} as a stochastic process termed \emph{spatiotemporal Besov process (STBP)}, denoted as $\STBP(\mC, B^{s,q,p}(\mZ))$. 
Similarly as in Section \ref{sec:Besov}, the infinite random sequence $\xi_\mT:=\{\xi_\ell(\cdot)\}_{\ell=1}^\infty$ is a random element of the probability space $(\Omega, \mB(\Omega), \mbP)$ with $\Omega=(L^p(\mT))^\infty$, product $\sigma$-algebra $\mB(\Omega)$ and probability measure $\mbP$ defined by the infinite product of $\qEP(0,\mC)$ measures. 
Then we can define a spatiotemporal Besov measure on $B^{s,q,p}(\mZ)$ as the law of STBP.
\begin{dfn}[Spatiotemporal Besov Measure]\label{dfn:STBP}
Let $\mbP$ be the measure of random sequences $\xi_\mT\in\Omega$.
Suppose we have the following map
\begin{equation}\label{eq:STBP_randf}
f_\gamma: \Omega \to B^{s,q,p}(\mZ), \quad \xi_\mT \mapsto u(\bz) = \sum_{\ell=1}^{\infty} u_\ell(t)\phi_\ell(\bx) = \sum_{\ell=1}^{\infty} \gamma_\ell \xi_\ell(t)\phi_\ell(\bx), 
\end{equation}
where $\gamma_\ell$ and $\xi_\ell(\cdot)$ are defined in \eqref{eq:q_EP}.
Then the pushforward $f_\gamma^\sharp\mbP$ is a spatiotemporal Besov measure $\Pi$ on $B^{s,q,p}(\mZ)$.
\end{dfn}


Based on \eqref{eq:STBP_randf}, we need to bound $\Vert u_\ell(\cdot)\Vert_p^q$ (or $\Vert\xi_\ell(\cdot)\Vert_p^q$) so the norm \eqref{eq:STBSV_norm} is well defined.
By Theorem \ref{thm:qepint}, $\Vert\xi_\ell(\cdot)\Vert_p^q$ has a bounded mean for $1\leq p\leq q$.
For the convenience of exposition and theoretical investigation, in the following we only consider $p=q\in[1,2]$. This is also when most interesting applications happen (See more details in Section \ref{sec:numerics}). 
Denote $\Vert u\Vert_{s,q}:=\Vert u\Vert_{s,q,q} = \Vert \xi_\mT\Vert_{q,q} = \left(\sum_{\ell=1}^\infty \Vert \xi_\ell(\cdot)\Vert^q_q\right)^{\frac{1}{q}}$ and $B^{s,q,q}(\mZ):=B^{s,q}(\mZ)$.
Figure \ref{fig:relationship} summarizes the relationship between GP, BP, Q-EP and their spatiotemporal variants.

\begin{figure}[t]
\centering
\begin{tikzpicture}\fontsize{9.5pt}{10.25pt}\selectfont
\matrix (m) [matrix of math nodes,row sep=3.5em,column sep=4em,minimum width=2em]
{
 & p\textrm{-exponential} \textrm{\citep{Agapiou_2021}} & & \\
 \GP(0,\mC_\bx) & \mB(\kappa, B^{s,q}(\mX)) \textrm{\citep{Lassas_2009}} & \STBP(\mC_t, B^{s,q}(\mZ)) & \STGP(0,\mC_\bz) \\
 & \qEP(0, \mC_t) \textrm{\citep{Li_2023}} & \textrm{time-uncorrelated} & \\};
\path[-stealth]
(m-2-2) edge node[right,transform canvas={shift={(-40pt,0pt)}}]{\footnotesize modified $\pi_\xi(\cdot)\propto \exp\left(-\frac{1}{p}|\cdot|^p\right)$} (m-1-2)
(m-2-2) edge[transform canvas={shift={(0pt,+2pt)}}] node[above]{$q=2$} (m-2-1)
(m-2-1) edge[transform canvas={shift={(0pt,-2pt)}}] node[below]{\scriptsize K-L expansion} (m-2-2)
(m-2-2) edge node[above]{\scriptsize $\xi_\ell\to\xi_\ell(\cdot)$} (m-2-3) 
(m-2-3) edge[transform canvas={shift={(0pt,+2pt)}}] node[above]{$q=2$} (m-2-4)
(m-2-3) edge node[right]{$\mC_t=\mI_t$} (m-3-3)
(m-2-4) edge[transform canvas={shift={(0pt,-2pt)}}] node[below]{\scriptsize K-L expansion} (m-2-3)
(m-3-2) edge node[below,rotate=-20]{{$q=2$}} (m-2-1)
(m-3-2) edge node[right,transform canvas={shift={(-20pt,+2pt)}}]{\scriptsize K-L expansion} (m-2-2)
(m-3-2) edge[transform canvas={shift={(-15pt,2pt)}}] node[below,rotate=16,transform canvas={shift={(+2pt,+2pt)}}]{\scriptsize $\xi_\ell(\cdot)\sim \qEP(0, \mC_t)$} (m-2-3);
\end{tikzpicture}
\caption{Relationship between GP, BP, Q-EP and their spatiotemporal variants.}
\label{fig:relationship}
\end{figure}
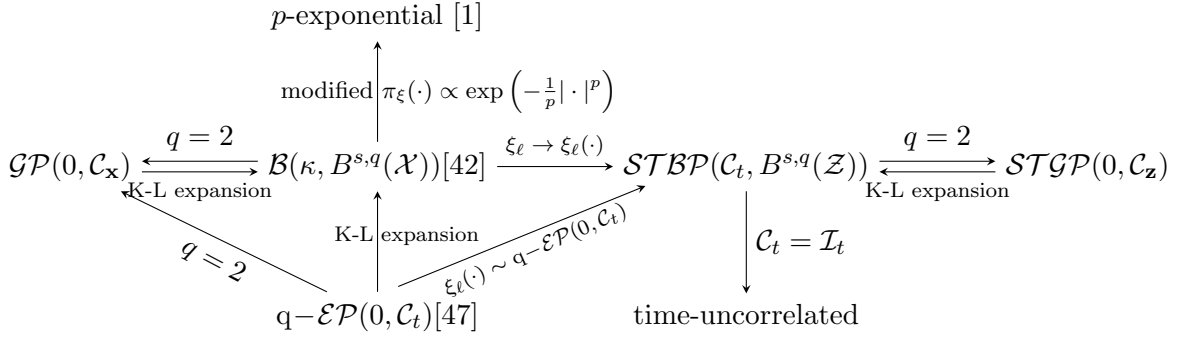

\subsection{STBP as A Prior}
In this section, we study multiple properties of STBP when used as a prior.
Similarly as in Section \ref{sec:Besov},
given a random draw $u\sim \STBP(\mC, B^{s,q}(\mZ))$ as in \eqref{eq:STBP_randf}, we have 
$\mbE[\Vert u\Vert_{s,q}^q] = \mbE[\Vert \xi_\mT\Vert_{q,q}^q]=\infty$ due to the iid assumption on $\xi_\mT$ in \eqref{eq:q_EP}. 
However, we have the following integrability of STBP function in the ambient space $B^{s',q}(\mZ)$ similar to Theorem \ref{thm:Lqbound}.
\begin{thm}\label{thm:Lqbound_ST}
    If $u \sim \STBP(\mC, B^{s,q}(\mZ))$ as in \eqref{eq:STBP_randf} satisfies both Assumptions \ref{asmp:eigen_summable}-(ii) and \ref{asmp:decaying_constant}, then $u\in L^q_\mbP(\Omega; B^{s',q}(\mZ))$ for all $s'<s-\frac{d}{q}$.
\end{thm}
The following (Fernique type) theorem enhances such well-definedness in the context of almost sure convergence for STBP random draws \citep{Lassas_2009,Dashti_2012}.
\begin{restatable}{thm}{Fernique}
\label{thm:Fernique}
Let $u\sim \STBP(\mC, B^{s,q}(\mZ))$ as in \eqref{eq:STBP_randf} satisfy both Assumptions \ref{asmp:eigen_summable}-(ii) and \ref{asmp:decaying_constant}. The following statements are equivalent:
\begin{enumerate}[label=(\roman*),itemsep=0pt]
\item $u\in B^{s',q}(\mZ)$ $\Pi$-a.s.
\item $\mathbb{E}[\exp(\alpha \Vert u \Vert_{s',q}^q)] <\infty \; for\; any\; \alpha\in \left(0, (\sup_\ell \lambda_\ell)^{-\frac{q}{2}}/2\right)$.
\item $s'<s-\frac{d}{q}$.
\end{enumerate}
\end{restatable}
\begin{proof}
See Supplement \ref{apx:Fernique}.
\end{proof}
\begin{rk}
    Fernique type result (ii) is important to make rigorous the formal Lebesgue density \eqref{eq:lebesgue} and the conditions of similar format in Assumption \ref{asmp:welldefpos} as well.
    The results $(i)$ and $(iii)$ immediately imply that the $\STBP(\mC, B^{s,q}(\mZ))$ measure $\Pi$ is supported on the ambient space $B^{s',q}(\mZ)$ for $s'<s-\frac{d}{q}$, as stated in the following corollary.
\end{rk}
\begin{cor}\label{cor:STBP_supp}
Let $\Pi$ be an $\STBP(\mC, B^{s,q}(\mZ))$ measure satisfying both Assumptions \ref{asmp:eigen_summable}-(ii) and \ref{asmp:decaying_constant}. Then $\Pi(B^{s',q}(\mZ))=1$ for any $s'<s-\frac{d}{q}$ and $\Pi(B^{s',q}(\mZ))=0$ for any $s'\geq s-\frac{d}{q}$.
\end{cor}

Note $B^{s,q}(\mZ)\subset B^{s',q}(\mZ)$ for $s'<s$.
The following general embedding highlights the relationship among various Besov spaces needed in the contraction theory (Section \ref{sec:post_contr}).
\begin{restatable}{prop}{embeddings}
\label{prop:embeddings}
For $q,q^\dagger\in [1,2]$ and
$s'<s^\dagger-\left(\frac{d}{q^\dagger}-\frac{d}{q}\right)_+$ with $x_+:=\max\{x,0\}$, we have $B^{s^\dagger,q^\dagger\wedge q, q}(\mZ)\hookrightarrow B^{s',q}$, where $q^\dagger\wedge q:=\min\{q^\dagger,q\}$.
\end{restatable}
\begin{proof}
See Supplement \ref{apx:embeddings}.
\end{proof}


Similarly to Theorem \ref{thm:KL}, we have the Karhunen-Lo\'eve theorem for an STBP $u(\cdot)$ as in \eqref{eq:STBP_randf} represented completely in the spatial ($\{\phi_\ell\}_{\ell=1}^\infty$) and temporal ($\{\psi_{\ell'}\}_{\ell'=1}^\infty$) bases.
\begin{restatable}{thm}{KLSTBP}[Karhunen-Lo\'eve]
\label{thm:KLSTBP}
If $u\sim \STBP(\mC, B^{s,q}(\mZ))$ as in \eqref{eq:STBP_randf} with a trace-class HS operator $T_\mC$ having eigen-pairs $\{\lambda_\ell, \psi_\ell(\cdot)\}_{\ell=1}^\infty$, then we have 
\begin{equation}\label{eq:STBP_serexp}
u(\bz) = \sum_{\ell=1}^\infty \sum_{\ell'=1}^\infty u_{\ell\ell'} \phi_\ell(\bx)\psi_{\ell'}(t), \quad u_{\ell\ell'}:=\int_\mT u_\ell(t) \psi_{\ell'}(t) dt \sim \qED(0, \gamma_\ell^2\lambda_{\ell'}) \,.
\end{equation}
Moreover, the spatiotemporal covariance of STBP bears a separable structure, i.e.
\begin{equation}\label{eq:STBP_cov}
    \mCv(u(\bz), u(\bz')) = \sum_{\ell=1}^{\infty}\gamma_\ell^2 \phi_\ell(\bx)\phi_\ell(\bx') \mC(t,t') \,.
\end{equation}
\end{restatable}
\begin{proof}
See Supplement \ref{apx:KLSTBP}.
\end{proof}
\begin{rk}\label{rk:STBP_cov}
In practice (Section \ref{sec:numerics}), we do not directly use the series-based representation \eqref{eq:STBP_serexp} because it is not straightforward to specify correlations among sample functions which requires tweaking basis functions according to \eqref{eq:STBP_cov}. 
Instead, we include a kernel $\mC$ in Definition \ref{dfn:STBP} of STBP to directly model the temporal correlations through covariance functions, e.g., squared exponential and Mat\'ern (See Section \ref{sec:hyperpar}). In Section \ref{sec:numerics}, we will also investigate the importance of temporal kernel $\mC$ by comparing with a time-uncorrelated method where $\mC=\mI$ (Refer to Figure \ref{fig:relationship}).
\end{rk}



The regularity of an STBP random draw $u(\bx ,t)$ as in \eqref{eq:STBP_serexp} also depends on the properties of spatial  ($\{\phi_\ell\}_{\ell=1}^\infty$) and temporal ($\{\psi_{\ell'}\}_{\ell'=1}^\infty$) bases. See Theorem \ref{thm:continuity} in Supplement \ref{apx:priorthm} for the H\"older continuity of $u$ proved by the Kolmogorov continuity test \citep[Theorem 30 in Section A.2.5 of][]{dashti2017}.
\subsection{Posterior Theorems of Bayesian Inverse Problems}\label{sec:post_contr}
In this section, we study the posterior properties of the Bayesian inverse model \eqref{eq:bip} with STBP priors. 
In particular, we consider the separable Banach space $\mbX = B^{s',q}(\mZ)$ for some $s'<s-\frac{d}{q}$, and let $\mbY$ be another separable Banach space, e.g., $\mbY=H^s(\mX)$ or $\mbY=\mbR^m$, depending on the applications.
For the potential (negative log-likelihood) function $\Phi:B^{s',q}(\mZ)\times\mbY\to\mbR$ as in \eqref{eq:potential}, we impose Lipschitz continuity in $u$ as stated in the following assumption which will be needed to bound the model complexity (refer to Lemma \ref{lem:KL_bounds}) in the proof of posterior contraction Theorem \ref{thm:postcontr}.
\begin{asump}\label{asmp:Lipschitz_u}
For every $r>0$, there exists $L=L(r)>0$ such that for every $y\in \mbY$ and for all $u_1,u_2\in B^{s',q}(\mZ)$ with $\max\{\Vert u_1\Vert_{s',q}, \Vert u_2\Vert_{s',q}\}<r$, 
\begin{equation*}
    |\Phi(u_1,y)-\Phi(u_2,y)|\leq L\Vert u_1-u_2\Vert_{s',q} \,. 
\end{equation*}
\end{asump}

Recall that $\Pi$ is the STBP prior defined by \eqref{eq:STBP_randf} and $\Pi(\cdot|y)$ is the resulting posterior measure of model \eqref{eq:bip}. The well-definedness (Theorem \ref{thm:bayespost}) and well-posedness (Theorem \ref{thm:wellpose}) of the posterior measure are re-examined in Supplement \ref{apx:posthm}.

Suppose that $u\in\mbX$ is evaluated at $I$ spatial locations and $J$ time points. Denote $n=I\wedge J=\min\{I,J\}$. Let the observations $Y^{(n)}:=\{Y_j\}_{j=1}^n$ be independent but not identically distributed.
Now we consider the concentration property of the posterior $\Pi_n(\cdot|Y^{(n)})$ in the limit $n\to\infty$. Unlike the Gaussian measure with reproducible kernel Hilbert space (RKHS), the lack of inner product structure on $B^{s,q}(\mZ)\subset B^{s',q}(\mZ)$ makes the posterior contraction theories more challenging \citep{Agapiou_2021}. 
We consider the separable Banach space $(B^{s',q}(\mZ), \Vert\cdot\Vert_{s',q})$. 
Define the concentration function of STBP measure $\Pi_n$ at $u=u^\dagger$ as
\begin{equation}\label{eq:contr_fun}
\varphi_{u^\dagger}(\eps) = \inf_{h\in B^{s,q}(\mZ): \Vert h-u^\dagger\Vert_{s',q}\leq \eps} \half \Vert h\Vert_{s,q}^q - \log \Pi_n(\Vert u\Vert_{s',q}\leq\eps).
\end{equation}

Denote $P_u^{(n)}:=\bigotimes_{j=1}^n P_{u,j}$ as the product measure on $\bigotimes_{j=1}^n(\mbY_j,\mB_j,\mu_j)$. Each $P_{u,j}$ has a density $p_{u_j}$ with respect to the $\sigma$-finite reference measure $\mu_j$, i.e. $\frac{d P_{u,j}}{d\mu_j} = p_{u_j}$.
Define the average Hellinger distance as $d^2_{n,H}(u,u')=\frac{1}{n}\sum_{j=1}^n\int(\sqrt{p_{u_j}}-\sqrt{p_{u'_j}})^2 d\mu_j$.
The following posterior contraction theorem states that the posterior, $u|Y^{(n)}$, converges to the true value $u^\dagger$ at a rate $\eps_n$ on sets $\Theta_n$ with dominant probability in $B^{s',q}(\mZ)$, justifying STBP as a valid learning tool in the infinite data limit.
\begin{restatable}{thm}{postcontr}[Posterior Contraction]
\label{thm:postcontr}
Let $u$ be an $\STBP(\mC, B^{s,q}(\mZ))$ random element as in \eqref{eq:STBP_randf} satisfying both Assumptions \ref{asmp:eigen_summable}-(ii) and \ref{asmp:decaying_constant} in $\Theta:=B^{s',q}(\mZ)$ with $s'<s-\frac{d}{q}$ and $P_u^{(n)}:=\bigotimes_{j=1}^n P_{u,j}$ is the product measure of $Y^{(n)}$ parameterized by $u$ with the potential function $\Phi$ \eqref{eq:potential} satisfying Assumption \ref{asmp:Lipschitz_u}. 
If the true value $u^\dagger\in \Theta$ is in the support of $u$, and $\eps_n$ satisfies the rate equation $\varphi_{u^\dagger}(\eps_n)\leq n\eps_n^2$ with $\eps_n\geq n^{-\half}$, then there exists a measurable set $\Theta_n\subset \Theta$ such that $P_{u^\dagger}^{(n)} \Pi_n(u\in \Theta_n: d_{n,H}(u,u^\dagger)\geq M_n\eps_n|Y^{(n)})\to 0$ for every $M_n\to\infty$.
Moreover, $P_{u^\dagger}^{(n)}\Pi_n(\Theta\backslash \Theta_n|Y^{(n)})\to 0$ as $n\to\infty$.
\end{restatable}

\begin{proof}
See Supplement \ref{apx:postcontr}.
\end{proof}

Denote $a\wedge b=\min\{a,b\}$, $a\vee b=\max\{a,b\}$, and $x_+=x\vee 0$.
By solving the inequality $\varphi_{u^\dagger}(\eps_n)\leq n\eps_n^2$ for the minimal $\eps_n$, we obtain the posterior contraction rate as follows. 
\begin{restatable}{thm}{postcontrate}[Posterior Contraction Rate]
\label{thm:postcontrate}
Let $u$ be an $\STBP(\mC, B^{s,q}(\mZ))$ random element in $\Theta:=B^{s',q}(\mZ)$ with $s'<s-\frac{d}{q}$. The rest of the settings are the same as in Theorem \ref{thm:postcontr}. If the true value $u^\dagger\in B^{s^\dagger,q^\dagger}(\mZ)$ with $s^\dagger>s'+\left(\frac{d}{q^\dagger}-\frac{d}{q}\right)_+$ and $q^\dagger, q\in[1,2]$, then we have the rate of the posterior contraction as 
$\eps_n= n^{-\frac{\sigma(s,q,s^\dagger,q^\dagger) - s'}{2(\sigma(s,q,s^\dagger,q^\dagger) - s')+q (s-\sigma(s,q,s^\dagger,q^\dagger))}}$, where $\sigma(s,q,s^\dagger,q^\dagger)=\left(s-\frac{d}{q}\right)\bigwedge \left(s^\dagger-\left(\frac{d}{q^\dagger}-\frac{d}{q}\right)_+\right)$.
\end{restatable}

\begin{proof}
See Supplement \ref{apx:postcontrate}.
\end{proof}

\begin{rk}\label{rk:whyq1}
The contraction rate $\eps_n$ becomes optimal, $\eps_n^* = n^{-\frac{1}{2+\frac{d}{s^\dagger-s'-\left(\frac{d}{q^\dagger}-\frac{d}{q}\right)_+}}}$, if $s=s^\dagger +\frac{d}{q}-\left(\frac{d}{q^\dagger}-\frac{d}{q}\right)_+$, which is further maximized as $\eps_n^\dagger = n^{-\frac{1}{2+\frac{d}{s^\dagger-s'}}}$ when $q\leq q^\dagger$.
Note that such optimal rate is achieved regardless of the value of modeling regularization parameter $q$ as long as $q\leq q^\dagger$. This implies that when modeling inhomogeneous data, under-smoothing (with smaller regularization parameter $q$) is preferred to over-smoothing (with larger $q$, refer to Figure \ref{fig:simulation_MAP_whiten_q}). This is also the reason why $q=1$ is often adopted -- the posterior converges the fastest if the true integrability $q^\dagger$ is at least $L_1$.
\end{rk}
\begin{rk}
Another observation is that the ambient space $B^{s',q}(\mZ)$ can be chosen for the smoothness parameter $s'<\sigma(s,q,s^\dagger,q^\dagger)$. 
In particular, we consider two cases:
\begin{enumerate}[label=(\roman*),itemsep=0pt]
\item If we set $\tau_q(s')=0$, i.e. $s'=\frac{d}{q}-\frac{d}{2}\geq 0$, then $B^{s',q}(\mZ)\cong \ell^q(L^q(\mT))$.
For the Gaussian case ($q^\dagger=2$), if we adopt $q=2$ and hence $s'=0$ and $B^{0,2}(\mZ)\cong \ell^2(L^2(\mT))\cong L^2(\mZ)$, then the optimal rate $\eps_n^\dagger =n^{-\frac{s^\dagger}{2s^\dagger+d}}$ is minimax \citep{vanderVaart08}.
For other sub-Gaussian cases ($q^\dagger<2$), such optimal rate $\eps_n^\dagger = n^{-\frac{s^\dagger-s'}{2(s^\dagger-s')+d}}$ is not minimax \citep{Agapiou_2021} regardless of the choice of $q\in[1,2]$ because either $s'>0$ ($q\leq q^\dagger<2$) or such optimal rate is not attained ($q>q^\dagger$).
\item On the other hand, if we allow $s'=0$ and consider a larger ambient space $B^{0,q}(\mZ)\cong \ell^{q,\half-\frac{1}{q}}(L^q(\mT)) \supset \ell^q(L^q(\mT))$, then the minimax rate $\eps_n^\dagger =n^{-\frac{s^\dagger}{2s^\dagger+d}}$ can be obtained for $q\leq q^\dagger$.
\end{enumerate}
\end{rk}

In general, the minimax posterior contraction rate cannot be achieved when $q>q^\dagger$. Therefore we typically rescale the prior to infuse additional regularity \citep{vandervaart09,Agapiou_2021}.
That is, we vary the scaling factor $\kappa>0$ as in \eqref{eq:q_EP} and rescale the Banach space $(\kappa B^{s,q}(\mZ),\Vert\cdot\Vert_{s,q}) \cong (B^{s,q}(\mZ), \kappa^{-1}\Vert\cdot\Vert_{s,q})$. Denote $(B^{s,q}_\kappa(\mZ),\Vert\cdot\Vert):=(B^{s,q}(\mZ), \kappa^{-1}\Vert\cdot\Vert_{s,q})$ and the corresponding (rescaled) STBP measure as $\Pi_\kappa$ in Definition \ref{dfn:STBP}. Now we redefine the posterior concentration function \eqref{eq:contr_fun} to be
\begin{equation}\label{eq:contr_fun_rescaled}
\varphi_{u^\dagger,\kappa}(\eps) = \inf_{h\in B^{s,q}_\kappa(\mZ): \Vert h-u^\dagger\Vert_{s',q}\leq \eps} \frac{\kappa^{-q}}{2} \Vert h\Vert_{s,q}^q - \log \Pi_\kappa(\Vert u\Vert_{s',q}\leq\eps).
\end{equation}
The following theorem regards the posterior contraction rate with rescaled STBP prior.
\begin{restatable}{thm}{adapostcontrate}[Adaptive Posterior Contraction Rate]
\label{thm:adapostcontrate}
Let $u$ be an $\STBP(\mC, B^{s,q}_\kappa(\mZ))$ random element in $\Theta:=B^{s',q}_\kappa(\mZ)$ with $s'<s-\frac{d}{q}$. Suppose $\eps_n$ satisfies the rate equation $\varphi_{u^\dagger,\kappa}(\eps_n)\leq n\eps_n^2$ with $\eps_n\geq n^{-\half}$. The rest of the settings are the same as in Theorem \ref{thm:postcontr}. If the true value $u^\dagger\in B^{s^\dagger,q^\dagger}_\kappa(\mZ)$ with $s^\dagger>s'+\left(\frac{d}{q^\dagger}-\frac{d}{q}\right)_+$ and $1\leq q^\dagger<q\leq 2$, then the minimax posterior contraction rate $\eps_n^\dagger =n^{-\frac{s^\dagger}{2s^\dagger+d}}$ can be attained at $s=\frac{s^\dagger(s^\dagger-s')}{s'+\left(\frac{d}{q^\dagger} -\frac{d}{q}\right)}+s'$ with the scaling factor $\kappa_n\asymp n^{-\frac{1}{2s^\dagger+d}\left[-\frac{s^\dagger(s^\dagger-s')}{s'+\left(\frac{d}{q^\dagger} -\frac{d}{q}\right)}+\frac{d}{q}+s^\dagger\right]}$.
\end{restatable}
\begin{proof}
See Supplement \ref{apx:adapostcontrate}.
\end{proof}

\section{Bayesian Inference}\label{sec:inference}
In this section, we describe the inference of the Bayesian inverse problem \eqref{eq:bip} with spatiotemporal observations using an STBP prior.
Assume the unknown function $u$ is evaluated at $I$ locations $\bX:=\{\bx_i\}_{i=1}^I$ and $J$ time points $\bt:=\{t_j\}_{j=1}^J$, that is,  $u(\bX,\bt):=\{u(\bx_i, t_j)\}_{i,j=1}^{I,J}$. 
The data $\bY=\{\by_j\}_{j=1}^J$ with $\by_{j}\in \mbR^{m}$ is observed through the forward operator $\mG$, which could be a linear mapping or a nonlinear one governed by a PDE.
Here we consider Gaussian noise and rewrite the model \eqref{eq:bip} as follows:
\begin{equation}\label{eq:STBP_model}
\begin{aligned}
\by_j &= \mG(u)(\bX, t_j) +  \vect\eps_j, \quad \vect\eps_j \overset{iid}{\sim} \mN_I(0, \Gamma_\mathrm{noise}), \quad j = 1,2,\dots, J, \\
u &\sim \STBP(\mC, B^{s,q}(\mZ)).
\end{aligned}
\end{equation}

In our applications of inverse problems, the spatial dimension $I$ is usually much larger than the temporal dimension ($I\gg J$). 
Therefore we truncate $u$ in \eqref{eq:STBP_randf} for the first $L>0$ terms:
$u(\bx, t) \approx u^L(\bx, t) = \sum_{\ell=1}^L \gamma_\ell \xi_\ell(t)\phi_\ell(\bx)$, 
and choose $L=2000$ in the numerical experiments (Section \ref{sec:numerics}).
Denote $\bu_j=u(\bX, t_j)\in\mbR^I$, and $\bU=[\bu_1,\cdots,\bu_J]_{I\times J}=u^L(\bX, \bt)=\bPhi \diag(\bgamma) \tp{\bXi}$ where $\bPhi=[\phi_1(\bX), \cdots, \phi_L(\bX)]_{I\times L}$, $\bgamma=(\gamma_1,\cdots, \gamma_L)$, and $\bXi=[\bxi_1, \cdots, \bxi_L]_{J\times L}$ with $\bxi_\ell=\xi_\ell(\bt)$. 
Instead of the large dimensional matrix $\bU$, we work with $\bXi$ of much smaller size.
Let $r_\ell = \tp{\bxi_\ell}\bC_J^{-1} \bxi_\ell$.
The log posterior for $\bXi$ is computed directly as
\begin{equation}\label{eq:logpost}
\begin{aligned}
\log p(\bXi, \theta|\bY) =& -\frac{J}{2} \log|\Gamma_\mathrm{noise}| - \half \sum_{j=1}^J \Vert \by_j - \mG(\bu_j)\Vert^2_{\Gamma_\mathrm{noise}} \\
&-\frac{L}{2}\log|\bC_J| + \frac{J}{2}(\frac{q}{2}-1) \sum_{\ell=1}^L \log r_\ell - \half \sum_{\ell=1}^L r_\ell^{\frac{q}{2}}.
\end{aligned}
\end{equation}



We optimize \eqref{eq:logpost} to obtain the maximum a posterior (MAP) estimate. To quantify the uncertainty efficiently, we need effective inference algorithms for high-dimensional models with non-Gaussian priors. We refer to the work of dimension-robust MCMC proposed by \cite{chen2018dimension} based on the pushforward of Gaussian white noise which in turn takes advantage of the dimension-independent sampling algorithms for Gaussian priors \citep{
beskos2017}. For the convenience of applications, in the following, we introduce a new white noise representation for STBP which is different from the one used in \cite{chen2018dimension} for series based priors.

\subsection{White Noise Representation}

Recall we have the stochastic representation \eqref{eq:stoch_QED} of $\vect\xi\sim \qED_J(\bzero, \bC)$: $\bxi=R\bL S$ with $R^q\sim \chi^2(J)$ and $S\sim \mathrm{Unif}(\mS^{J+1})$.
We can write
\begin{equation*}
    S= \frac{\bzeta}{\Vert\bzeta\Vert_2}, \quad R^q = \Vert\bzeta\Vert_2^2, \quad \textrm{for}\; \bzeta\sim \mN_J(\bzero, \bI_J).
\end{equation*}
Therefore, $\bxi$ can be represented in terms of the white noise $\bzeta$ by a pushforward mapping $\Lambda:\mbR^J\to\mbR^J$ and vice versa with its inverse $\Lambda^{-1}$:
\begin{equation}\label{eq:Lambda}
    \bxi = \Lambda(\bzeta) = \bL \bzeta \Vert\bzeta\Vert_2^{\frac{2}{q}-1}, \quad \bzeta=\Lambda^{-1}(\bxi)=\bL^{-1}\bxi \Vert \bL^{-1}\bxi\Vert_2^{\frac{q}{2}-1}.
\end{equation}
For $\zeta(\cdot) =\sum_{\ell'=1}^\infty \zeta_{\ell'}\psi_{\ell'}(\cdot),\; \xi(\cdot) =\sum_{\ell'=1}^\infty \xi_{\ell'}\psi_{\ell'}(\cdot) \in L^2(\mT)$, we can extend $\Lambda$ and its inverse $\Lambda^{-1}$ to $L^2(\mT)$ and have
\begin{equation*}
\xi(\cdot) = \Lambda(\zeta(\cdot)) = \sum_{\ell'=1}^\infty \lambda_{\ell'}^\half \zeta_{\ell'} \psi_{\ell'}(\cdot) \Vert \zeta(\cdot)\Vert_2^{\frac{2}{q}-1}, \quad \zeta(\cdot) = \Lambda^{-1}(\xi(\cdot)) = \sum_{\ell'=1}^\infty \lambda_{\ell'}^{-\half} \xi_{\ell'} \psi_{\ell'}(\cdot) \Vert \xi(\cdot)\Vert_{2,\mC}^{\frac{2}{q}-1},
\end{equation*}
where $\Vert \zeta(\cdot)\Vert_2^2 = \sum_{\ell'=1}^\infty \zeta_{\ell'}^2$ and $\Vert \xi(\cdot)\Vert_{2,\mC}^2 = \sum_{\ell'=1}^\infty \lambda_{\ell'}^{-1}\xi_{\ell'}^2$.
We propose the following representation of $u(\bz)$ in terms of an infinite sequence of white noises, i.e., $\zeta:=\{\zeta_\ell(\cdot)\}_{\ell=1}^\infty$:
\begin{equation*}\label{eq:T}
    u(\bz) = T(\zeta) = \sum_{\ell=1}^\infty \gamma_\ell \Lambda(\zeta_\ell(t))\phi_\ell(\bx), \quad \zeta_\ell(\cdot) \overset{i.i.d.}{\sim} \GP(0, \mI).
\end{equation*}

Denote $\bZeta=[\bzeta_1, \cdots, \bzeta_L]_{J\times L}$ with $\bzeta_\ell=\zeta_\ell(\bt)$. 
From the above equation, 
we have $\bU=T(\bZeta)=\bPhi \diag(\bgamma) \tp{\Lambda(\bZeta)}$. 
Then the log-posterior in \eqref{eq:logpost} can be rewritten in terms of $\bZeta$:
\begin{equation}\label{eq:logpost_whiten}
\begin{aligned}
\log p(\bZeta, \theta|\bY) =& -\frac{J}{2} \log|\Gamma_\mathrm{noise}| - \half \sum_{j=1}^J \tr (\tp{(\bY - \mG(T(\bZeta)))} \Gamma_\mathrm{noise}^{-1} (\bY - \mG(T(\bZeta))) )\\
&-\frac{L}{2}\log|\bC_J| + \frac{J}{2}(\frac{q}{2}-1) \sum_{\ell=1}^L \log r_\ell - \half \sum_{\ell=1}^L r_\ell^{\frac{q}{2}},
\end{aligned}
\end{equation}
where $r_\ell = \tp{\Lambda(\bzeta_\ell)}\bC_J^{-1} \Lambda(\bzeta_\ell)$.
Once the MAP $\bZeta_\text{\tiny MAP}$ is obtained by maximizing the above log-posterior, we can obtain $\bU_\text{\tiny MAP}=T(\bZeta_\text{\tiny MAP})$.
We refer to this process as ``optimization in the whitened space".
The objective function can be explored more efficiently in the whitened space with variables de-correlated (See Figure \ref{fig:STEMPO_err}).


\subsection{White Noise MCMC}
Denote the measure formed by infinite product of $\GP(0, \mI)$ as $\Pi_0$. Then our STBP prior measure $\Pi$ can be regained by the pushforward using $T$, i.e. $\Pi=T^\sharp \Pi_0$.
A class of dimension-independent MCMC algorithms \citep{beskos2017} for models with Gaussian prior $\Pi_0$ 
can be reintroduced to posterior sampling with STBP prior $\Pi$.

Let $u=T(\zeta)$ with $\zeta\sim \Pi_0$.
Consider the continuous-time Hamiltonian dynamics:
\begin{equation}\label{eq:HD}
\frac{d^2\zeta}{dt^2} + \mK(\zeta)\,
[ \zeta + \nabla\Phi(\zeta) ] = 0, \quad \left. \left(\eta:= \frac{d \zeta}{dt}\right)\right|_{t=0} \sim\mN(0,\mK(\zeta)),
\end{equation}
where $\Phi(\zeta):=\Phi(T(\zeta)) -\log |dT(\zeta)|$.
Generally, we set $\mK(\zeta)^{-1}=\mI+\beta \mH(\zeta)$,  where $\mH(\zeta)$ can be chosen as 
Gauss-Newton Hessian computed as $\mH(\zeta)= dT^* \mH(u) dT$ with $dT$ being the Jacobian.
Let $g(\zeta):=-\mK(\zeta)\{\alpha\nabla\Phi(\zeta)-\beta\mH(\zeta)\zeta\}$, where $\nabla_\zeta\Phi(\zeta)=dT^* \nabla_u\Phi(u)-\nabla_\zeta \log |dT(\zeta)|$.
The Hamiltonian Monte Carlo (HMC) algorithm \citep{neal10} solves the dynamics \eqref{eq:HD}
using the St\"ormer-Verlet symplectic (leapfrog) integrator with step size $\eps$: 
\begin{equation}\label{eq:mHDdiscret}
\begin{aligned}
\eta^- &= \eta_0 + \tfrac{\eps}{2}\,g(\zeta_0)\ ; \\
\begin{bmatrix} \zeta_\eps\\ \eta^{+}\end{bmatrix} &= \begin{bmatrix} \cos\eps & \sin\eps\\ -\sin\eps & \cos\eps
\end{bmatrix}  \begin{bmatrix} \zeta_0\\ \eta^{-}\end{bmatrix}\  ;\\
\eta_\eps &= \eta^{+} + \tfrac{\eps}{2}\,g(\zeta_\eps)\  .
\end{aligned}
\end{equation}
%
Equation \eqref{eq:mHDdiscret} gives rise to the
leapfrog map $\Psi_\eps: (\zeta_{0},\eta_{0})\mapsto (\zeta_\eps, \eta_\eps)$.
Given a time horizon $\tau$ and current position 
$\zeta$, the MCMC mechanism proceeds 
by concatenating $I=\lfloor \tau/\epsilon \rfloor$ steps of leapfrog map consecutively,
$\zeta' =\mathcal{P}_\zeta\big\{\Psi_\eps^I(\zeta,\eta)\big\}\ , \; \eta\sim\mN(0,\mK(\zeta))$,  
where $\mathcal{P}_\zeta$ denotes the projection onto the $\zeta$-argument.
Then, the proposal $\zeta'$ is accepted with probability $a(\zeta,\zeta')=1\wedge \exp(-\Delta E(\zeta, \eta))$ \citep{beskos2017}. 
At last we convert the sample $\zeta$ back to $u=T(\zeta)$.
This yields a white-noise infinite-dimensional manifold HMC (wn-$\infty$-mHMC) 
which reduces to white-noise infinite-dimensional manifold Metropolis adjusted Langevin algorithm (wn-$\infty$-mMALA) 
when $I=1$, and white-noise infinite-dimensional HMC (wn-$\infty$-HMC) 
when $\beta=0$ \citep{beskos2017}. We set $\alpha=1$ for both scenarios and summarize all these methods in Algorithm \ref{alg:wn-infMC} of Supplement \ref{apx:inference} named as \emph{white-noise dimension-independent MCMC (wn-$\infty$-MCMC)}.


\subsection{Hyper-parameter Tuning}\label{sec:hyperpar}
In Definitions \ref{dfn:qEP} and \ref{dfn:STBP}, there is a temporal kernel $\mC$ that has not been specified. This is the key component to capture the temporal correlation which is absent in a pure series based approach (See Remark \ref{rk:STBP_cov} and Section \ref{sec:numerics}).
There are hyper-parameters, denoted as $\theta$, in the covariance kernel $\mC$, e.g., variance magnitude ($\kappa$) and correlation length ($\rho$), i.e., $\theta=(\kappa, \rho)$, that require careful adjustment and fine tuning as in, e.g., Mat\'ern kernel:
\begin{equation}\label{eq:matern}
\mC(t,t')=\kappa \frac{2^{1-\nu}}{\Gamma(\nu)} w^\nu K_\nu(w), \quad w=\sqrt{2\nu} (\Vert t-t'\Vert/\rho)^s.
\end{equation}
Unless we assume the likelihood in the Bayesian inverse model \eqref{eq:bip} is another $\qED$ and the forward mapping is linear \citep[c.f. Theorem 3.5 of][]{Li_2023}, we do not have a tractable marginal likelihood to optimize for these hyper-parameters \citep{Rasmussen_2005}.
In general settings, e.g., in the model \eqref{eq:STBP_model} with a Gaussian likelihood, we need to jointly update $(\bXi, \theta)$ based on \eqref{eq:logpost} or $(\bZeta, \theta)$ according to \eqref{eq:logpost_whiten}.
Denote $\bC_0=\kappa^{-1}\bC$ and $r_{0,\ell}=\tp{\bxi_\ell} \bC_0^{-1} \bxi_\ell$.
Proposition \ref{prop:condconj} in Supplement \ref{apx:inference} states that 
$\kappa^{\frac{q}{2}}|\bu \sim \Gamma^{-1}(\alpha',\beta'), \quad 
\alpha'=\alpha+\frac{JL}{2}, \quad \beta'=\beta+\half\sum_{\ell=1}^L r_{0,\ell}^\frac{q}{2}$.
Therefore, we could either update $\kappa \leftarrow \left(\frac{\beta'}{\alpha'+1}\right)^{\frac{2}{q}}$ or sample $\kappa$ according to \eqref{eq:cond_kappa}.
In general, there is no such conditional conjugacy for the correlation length ($\rho$).
We impose a hyper-prior for $\rho$ and optimize with or sample from $p(\rho|\bXi)$.

\section{Numerical Experiments}\label{sec:numerics}

In this section, we compare the proposed STBP ($\STBP(\mC, B^{s,q}(\mZ))$) with STGP (equivalent to $\STBP(\mC, B^{s,2}(\mZ))$) and a time-uncorrelated prior ($\STBP(\mI, B^{s,q}(\mZ))$) using a simulated regression, two dynamic tomography imaging examples, an inverse problem of recovering a spatiotemporal function, and a spatiotemporal imputation of temperature anomalies. 
Since the main focus is to model inhomogeneous data such as images with edges, we tend to adopt sharper regularization and set $q=1$ for STBP throughout this section (See also Remark \ref{rk:whyq1}).
Our numerical results demonstrate the advantage of Besov ($L_1$) type priors over Gaussian ($L_2$) type priors in modeling inhomogeneity. Moreover, these examples highlight the importance of appropriately modeling temporal correlations in spatiotemporal inverse problems.
All computer codes are publicly available at \url{https://github.com/lanzithinking/Spatiotemporal-Besov-prior}.

In all these applications, 
$u(\bx_i, t_j)$ refers to the image pixel value of point $\bx_i$ at time $t_j$ with resolution $I = n_x\times n_y$.
To assess the quality of reconstructed images, 
we refer to several quantitative measures including the relative error,
$\text{RLE} = \frac{\Vert u^* - u^\dagger\Vert}{\Vert u^\dagger\Vert}$, 
where $u^\dagger$ denotes the reference/true image and $u^*$ is its reconstruction. 
Additionally, we adopt the peak signal-to-noise ratio,
$\text{PSNR} = 10*\log_{10}(\frac{\Vert u^\dagger\Vert_\infty^2}{\Vert u^* - u^\dagger\Vert_2^2})$,
by using the maximum possible pixel value as a reference point to normalize the MSE.
We also consider the structured similarity index \citep{wang2004image},
$\text{SSIM($u^*,u^{\dagger}$)} = \frac{(2\bar{u^*}\bar{u^{\dagger}}+c_1)(2s_{u^*u^{\dagger}}+c_2)}{(\bar{u^*}^2+\bar{u^{\dagger}}^2+c_1)(s^2_{u^*}+s^2_{u^{\dagger}}+c_2)}$,
where $\bar{u}$, $s^2_u$ and $s_{u_1u_2}$ denote the sample mean, sample variance, and sample covariance, respectively, $c_i = (k_iL)^2$ for $i=1,2$, $k_1=0.01$, $k_2=0.03$ and $L$ is the dynamic range of the pixel values of the reference images.
\begin{figure}[t]
\begin{tabular}{cccccc}
\quad Truth & \qquad Observations & \qquad STBP & \qquad STGP & \quad time-uncorrelated &  pure-Besov  \\
\end{tabular}
\vspace{-5pt}
\begin{tabular}{cccccc}
\includegraphics[width=0.16\textwidth,height=.16\textwidth]{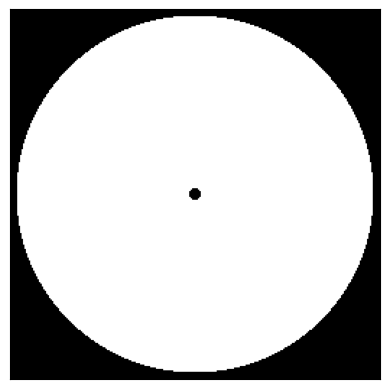}
\includegraphics[width=0.16\textwidth,height=.16\textwidth]{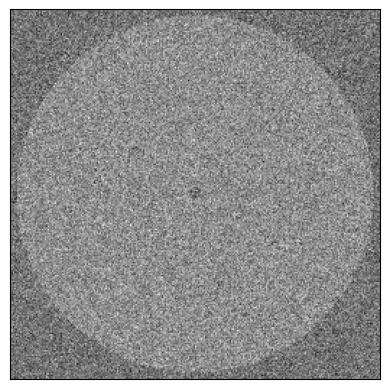}
\includegraphics[width=0.16\textwidth,height=.16\textwidth]{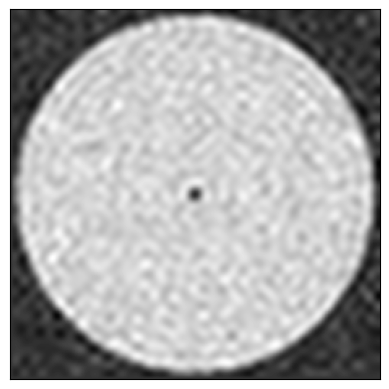}
\includegraphics[width=0.16\textwidth,height=.16\textwidth]{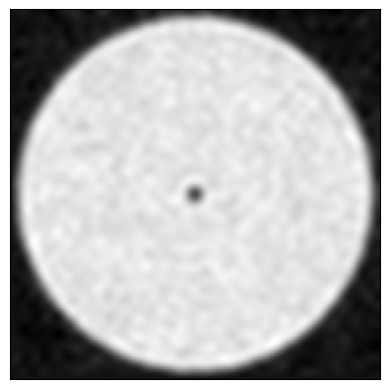}
\includegraphics[width=0.16\textwidth,height=.16\textwidth]{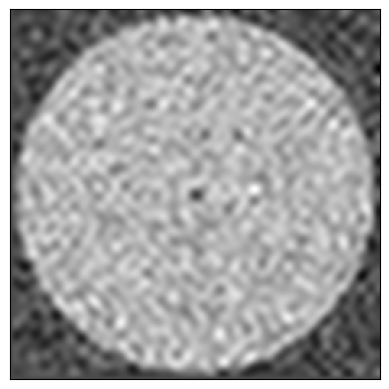}
\includegraphics[width=0.16\textwidth,height=.16\textwidth]{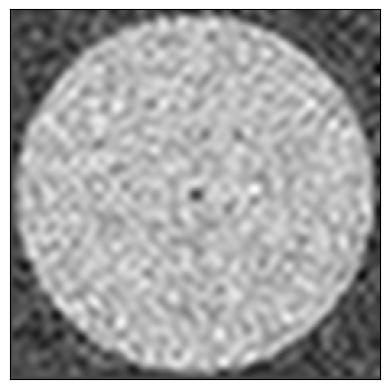}
\end{tabular}
\vspace{-5pt}
\begin{tabular}{cccccc}
\includegraphics[width=0.16\textwidth,height=.16\textwidth]{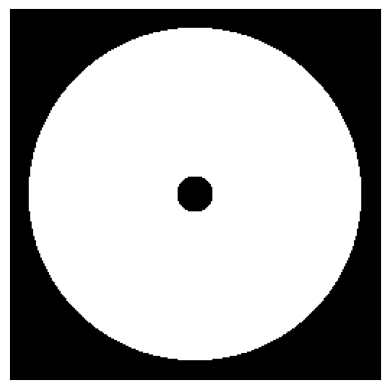}
\includegraphics[width=0.16\textwidth,height=.16\textwidth]{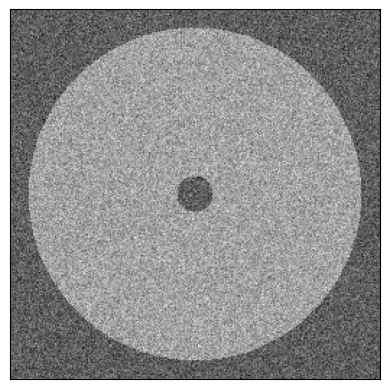}
\includegraphics[width=0.16\textwidth,height=.16\textwidth]{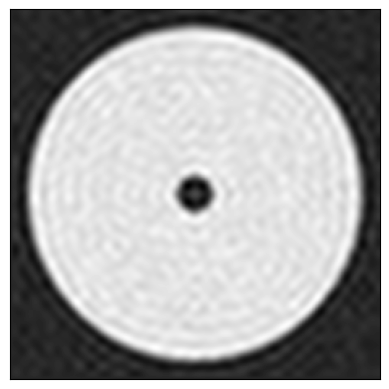}
\includegraphics[width=0.16\textwidth,height=.16\textwidth]{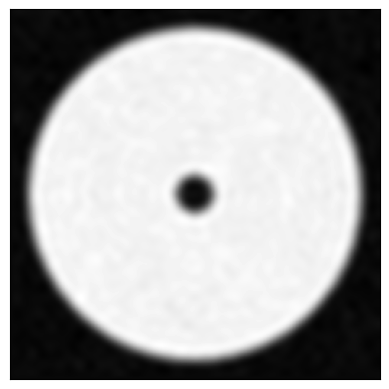}
\includegraphics[width=0.16\textwidth,height=.16\textwidth]{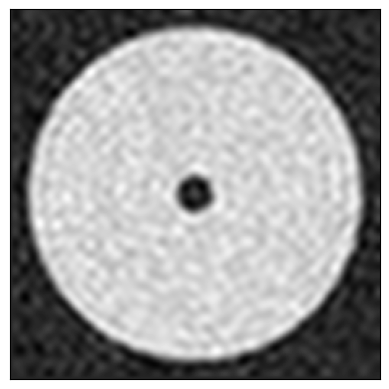}
\includegraphics[width=0.16\textwidth,height=.16\textwidth]{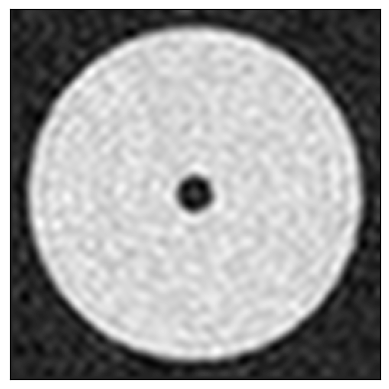}
\end{tabular}
\vspace{-5pt}
\begin{tabular}{cccccc}
\includegraphics[width=0.16\textwidth,height=.16\textwidth]{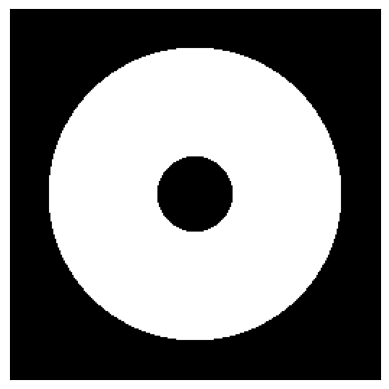}
\includegraphics[width=0.16\textwidth,height=.16\textwidth]{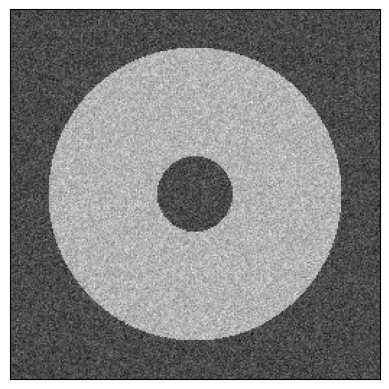}
\includegraphics[width=0.16\textwidth,height=.16\textwidth]{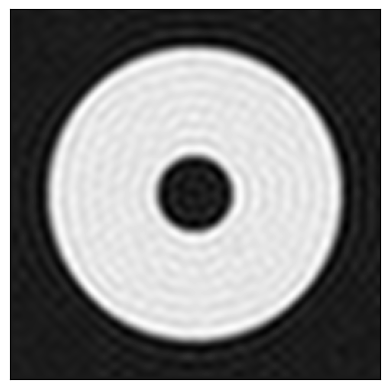}
\includegraphics[width=0.16\textwidth,height=.16\textwidth]{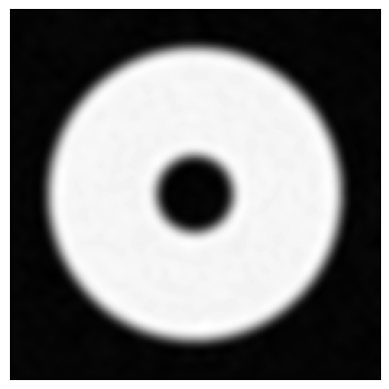}
\includegraphics[width=0.16\textwidth,height=.16\textwidth]{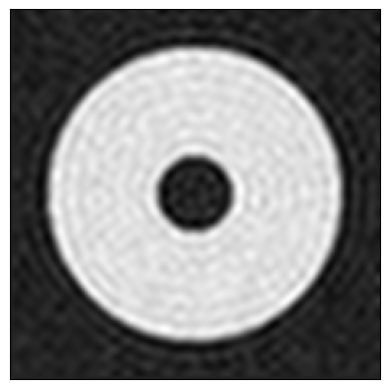}
\includegraphics[width=0.16\textwidth,height=.16\textwidth]{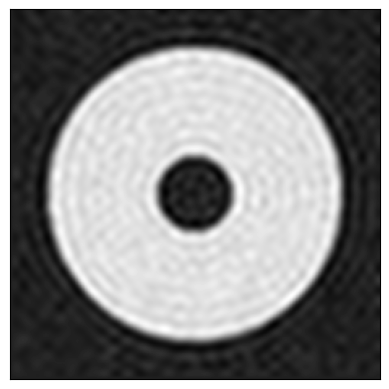}
\end{tabular}
\vspace{-5pt}
\begin{tabular}{cccccc}
\includegraphics[width=0.16\textwidth,height=.16\textwidth]{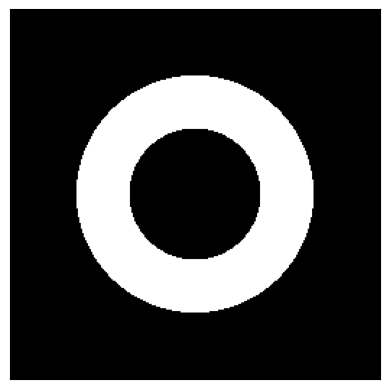}
\includegraphics[width=0.16\textwidth,height=.16\textwidth]{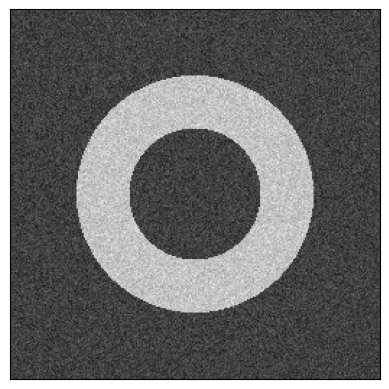}
\includegraphics[width=0.16\textwidth,height=.16\textwidth]{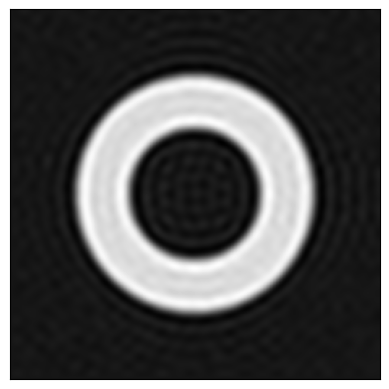}
\includegraphics[width=0.16\textwidth,height=.16\textwidth]{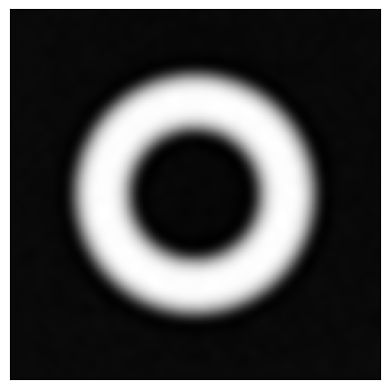}
\includegraphics[width=0.16\textwidth,height=.16\textwidth]{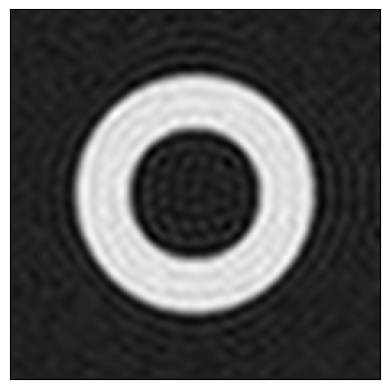}
\includegraphics[width=0.16\textwidth,height=.16\textwidth]{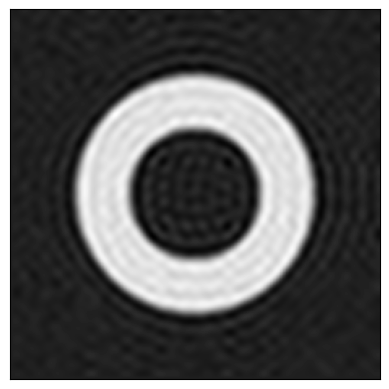}
\end{tabular}
\vspace{-5pt}
\caption{MAP reconstruction of simulated annulus with $I=256\times 256$, $J=100$. Columns: true images, observations, MAP estimates by STBP, STGP, time-uncorrelated and pure-Besov models, respectively. Rows from top to bottom: time step $t_j = 0.1, 0.3, 0.6$, and $0.9$.}
\label{fig:simulation_MAP_whiten}
\end{figure}
\subsection{Simulation}
First, we consider a simulated regression problem of a rising and shrinking 2d annulus: 
\begin{equation}\label{eq:sim_truth}
u(\bx, t) = t \delta(\sin(\pi\Vert \bx\Vert_2)\geq t), \quad \bx\in\mbR^2 \; such \; that \; \Vert \bx \Vert_\infty\leq 1, \quad t\in(0,1],
\end{equation}
where $\delta(\cdot)$ is the Dirac function. 
The first column of Figure \ref{fig:simulation_MAP_whiten} plots this function at a few time points illustrating a 2d annulus forming and shrinking as time goes by.
We simulate the data by discretizing the function $u(\bx, t)$ in \eqref{eq:sim_truth} on an $I=n_x\times n_y$ mesh (denoted as $\bX$) in the boxed spatial domain $\mX=[-1,1]^2$ over a grid of $J$ time points (denoted as $\bt$) in $\mT=(0,1]$, and adding some Gaussian noise with $\sigma_\eps=0.1$, i.e.,
$\by_j = u(\bX, t_j) + \vect\eps_j, \quad \vect\eps_j\overset{iid}{\sim} \mN_I(0, \sigma^2_\eps \bI)$.
The noisy spatiotemporal data are demonstrated in the second column of Figure \ref{fig:simulation_MAP_whiten}.
Based on the observed data, the goal of this Bayesian inverse problem is to recover the ground truth \eqref{eq:sim_truth} using STBP, STGP and time-uncorrelated priors.
Here, to contrast the effect of series based priors, we also include a pure-Besov prior whose random function is represented in \eqref{eq:STBP_serexp} with Fourier basis.
We use this example to numerically investigate the posterior contraction studied in Section \ref{sec:post_contr} as $n=I\wedge J\to \infty$. In particular, we will consider the problem with data observed at various spatiotemporal resolutions by considering combinations of $I=16\times 16, 32\times 32, 128\times 128, 256\times 256$ and $J=10, 20, 50, 100$, respectively.

\renewcommand{\arraystretch}{0.6}
\renewcommand*{\minval}{0.05}
\renewcommand*{\maxval}{0.40}
\begin{table}[t]\footnotesize
\centering
\caption{Comparison of MAP estimates for simulated annulus generated by STBP, STGP, time-uncorrelated and pure-Besov prior models
in terms of RLE with increasing data. Standard deviations (in bracket) are obtained by repeating the experiments for 10 times with different random seeds for initialization.}
\begin{tabular}{|*{6}{c|}}
    \hline
    $I\!=\!n_x\!\times\! n_y$ & $J$ & pure-Besov & time-uncorrelated & STGP & STBP \\ \hline
    & 10 & \gradient{0.2754} (3.47e-6) & \gradient{0.2768} (5.32e-6) & \gradient{0.2632} (7.91e-6) & \gradient{0.2770} (4.14e-6)   \\ \cline{2-5}
    $16 \times 16$ & 20 & \gradient{0.2406} (3.31e-6) & \gradient{0.2399} (3.52e-6) & \gradient{0.2063} (5.97e-6) & \gradient{0.2421} (6.10e-6)   \\ \cline{2-5}
    & 50 & \gradient{0.1910} (4.38e-6) & \gradient{0.2067} (1.02e-6) & \gradient{0.1594} (1.08e-5) & \gradient{0.1553} (1.07e-5)    \\ \cline{2-5}
    & 100 & \gradient{0.1605} (4.72e-6) & \gradient{0.1850} (8.12e-7) & \gradient{0.1217} (8.60e-6) & \gradient{0.1211} (7.79e-6)    \\ \hline
    & 10 & \gradient{0.3650} (4.61e-6) & \gradient{0.3677} (3.00e-6) & \gradient{0.2465} (7.08e-6) & \gradient{0.3672} (6.81e-6)    \\ \cline{2-5}
    $32 \times 32$ & 20 & \gradient{0.2407} (5.99e-6) & \gradient{0.2460} (1.49e-6) & \gradient{0.1911} (4.88e-6) & \gradient{0.2359} (5.02e-6)   \\ \cline{2-5}
    & 50 & \gradient{0.2103} (2.78e-6) & \gradient{0.2236} (1.30e-6) & \gradient{0.1464} (7.27e-6) & \gradient{0.1662} (7.36e-6)   \\ \cline{2-5}
    & 100 & \gradient{0.1626} (5.87e-6) & \gradient{0.1918} (9.90e-7) & \gradient{0.1203} (7.73e-6) & \gradient{0.1241} (9.31e-6)    \\ \hline
    & 10 & \gradient{0.1926} (1.49e-5) & \gradient{0.1943} (1.57e-5) & \gradient{0.2052} (4.51e-6) & \gradient{0.1937} (1.36e-5)   \\ \cline{2-5}
    $128 \times 128$ & 20 & \gradient{0.1440} (7.30e-6) & \gradient{0.1474} (7.23e-6) & \gradient{0.1497} (4.16e-6) & \gradient{0.1399} (1.77e-5)  \\ \cline{2-5}
    & 50 & \gradient{0.1182} (1.17e-5) & \gradient{0.1227} (4.95e-6) & \gradient{0.1083} (1.03e-5) & \gradient{0.1030} (1.51e-5)   \\ \cline{2-5}
    & 100 & \gradient{0.1073} (7.02e-6) & \gradient{0.1146} (2.31e-6) & \gradient{0.0934} (1.51e-5) & \gradient{0.0909} (1.24e-5)   \\ \hline
    & 10 & \gradient{0.1630} (1.50e-5) & \gradient{0.1635} (1.47e-5) & \gradient{0.1970} (2.70e-6) & \gradient{0.1633} (1.03e-5)   \\ \cline{2-5}
    $256 \times 256$ & 20 & \gradient{0.1159} (8.73e-6) & \gradient{0.1190} (7.24e-6) & \gradient{0.1442} (7.18e-6) & \gradient{0.1088} (1.24e-5)  \\ \cline{2-5}
    & 50 & \gradient{0.0949} (7.89e-6) & \gradient{0.0966} (5.77e-6) & \gradient{0.1012} (1.46e-5) & \gradient{0.0858} (1.31e-5)   \\ \cline{2-5}
    & 100 & \gradient{0.0892} (8.48e-6) & \gradient{0.0910} (3.89e-6) & \gradient{0.0864} (6.68e-6) & \gradient{0.0808} (1.36e-5)   \\ \hline
\end{tabular}
\label{tab:simulation}
\end{table}

Note, the spatial image of the function at each time point, when viewed as a picture, has clear edges. This imposes challenges for GP as it tends to oversmooth when modeling inhomogeneous objects while BP is more amenable. On the other hand, these sequential images are not isolated from each other in time, and the temporal kernel $\mC$ in STBP (STGP) can be well-used to capture such dependence. More specifically, we adopt the Mat\'ern kernel \eqref{eq:matern}
with $\nu=\half$, $\sigma^2=1$, $\rho=0.1$ and $s=1$.
The MAP estimate for $\bU=u(\bX, \bt)$ is obtained by minimizing the negative log-posterior \eqref{eq:logpost_whiten} in the whitened space of $\bZeta$ and converting $\bZeta_\text{\tiny MAP}$ back to $\bU_\text{\tiny MAP}=T(\bZeta_\text{\tiny MAP})$.
The last four columns of Figure \ref{fig:simulation_MAP_whiten} compare the MAP estimates by STBP, STGP, time-uncorrelated and pure-Besov models at $I=256\times 256$, $J=100$. The STGP model indeed returns an over-smoothed result; while the time-uncorrelated model yields a more noisy estimate due to the negligence of temporal correlation.
The results by pure-Besov model are comparably noisy to those obtained by the time-uncorrelated model (See also Table \ref{tab:simulation}).
Figure \ref{fig:simulation_MAP_whiten_q} also demonstrates different degrees of regularization interpolating with parameter $q$ in the range of $(0,2]$ with $q=2$ (STGP) yielding the most blurry solution.

Next, we vary the spatiotemporal resolution by changing the mesh and the time interval for observations. Figure \ref{fig:simulation_postcontr_whiten} investigates the MAP estimates by the STBP model with increasing data. They gradually approximate the true function \eqref{eq:sim_truth} as the spatiotemporal resolution is refined. This verifies the posterior consistency described in Theorem \ref{thm:postcontr} in terms of point estimation.
Table \ref{tab:simulation} also shows an error reducing phenomenon with increasing data for all three models. To make a fair comparison across different resolutions, we adopt $\Vert \bU\Vert_{\infty,1}=\max_{1\leq i\leq I}\sum_{j=1}^J|u(\bx_i,t_j)|$ in the RLE to focus on the pixel differences while averaging over the time domain. The STBP model outperforms the other two in most cases. Though not a direct verification of the posterior contraction rate in Theorem \ref{thm:postcontrate}, it shows that STBP reduces error with increasing data at rates not slower than STGP.

Though having similar performance as the time-uncorrelated prior, the pure-Besov prior has non-zero temporal correlations, as expressed in $\mC_t(t, t')=\sum_{\ell'=1}^\infty \lambda_{\ell'}\psi_{\ell'}(t)\psi_{\ell'}(t')$. Since our motivation is to contrast different strategies on temporal dependence in spatiotemporal modeling, we omit the pure-Besov prior from the following comparison.

\subsection{Dynamic Tomography Reconstruction}
In this section, we investigate 
the dynamic reconstruction of a simulated (STEMPO) and a real (emoji in Supplement \ref{apx:emoji}) tomography problem. Computed tomography (CT) is a medical imaging technique used to non-intrusively obtain detailed internal images of a subject such as human body \citep{Shepp_1974}. 
CT scanners project (Radon transformation) X-ray over the subject at different angles and measure the attenuated signals by an array of sensors recorded as sinograms. 

\begin{figure}[t]
\begin{tabular}{ccccc}
\quad Truth & \qquad \qquad Observations & \quad \qquad STBP & \qquad \qquad STGP & \quad \quad time-uncorrelated  \\
\end{tabular}
\vspace{-5pt}
\begin{tabular}{ccccc}
\includegraphics[width=0.19\textwidth,height=.18\textwidth]{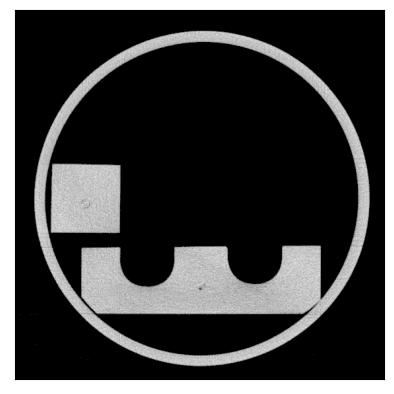}
\includegraphics[width=0.19\textwidth,height=.18\textwidth]{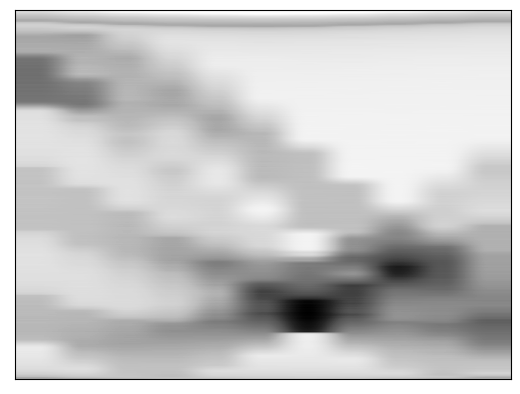}
\includegraphics[width=0.19\textwidth,height=.18\textwidth]{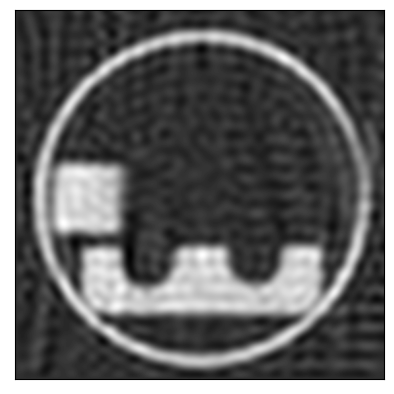}
\includegraphics[width=0.19\textwidth,height=.18\textwidth]{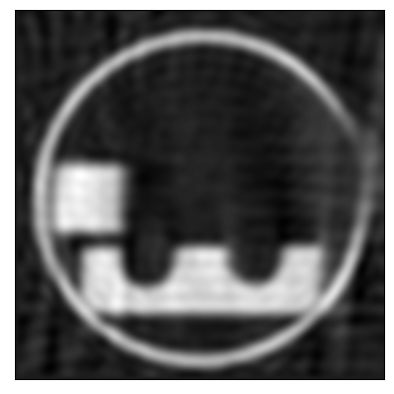}
\includegraphics[width=0.19\textwidth,height=.18\textwidth]{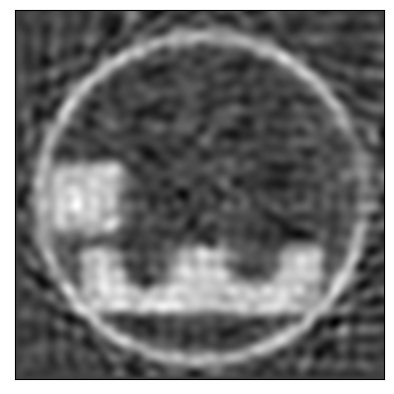}
\end{tabular}
\vspace{-5pt}
\begin{tabular}{ccccc}
\includegraphics[width=0.19\textwidth,height=.18\textwidth]{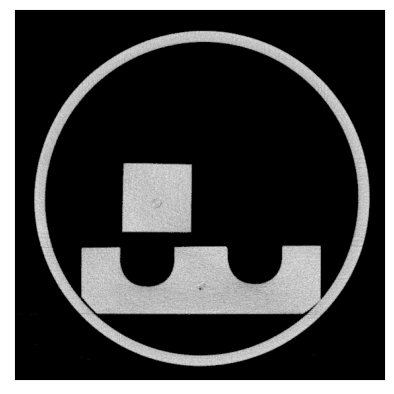}
\includegraphics[width=0.19\textwidth,height=.18\textwidth]{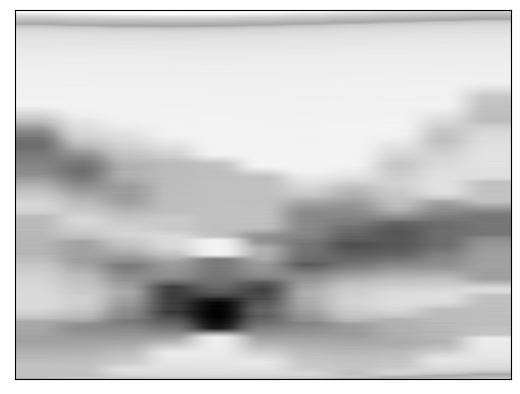} 
\includegraphics[width=0.19\textwidth,height=.18\textwidth]{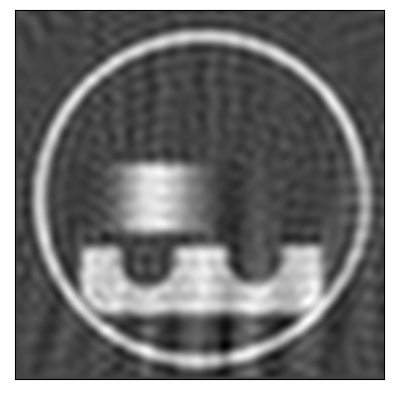} 
\includegraphics[width=0.19\textwidth,height=.18\textwidth]{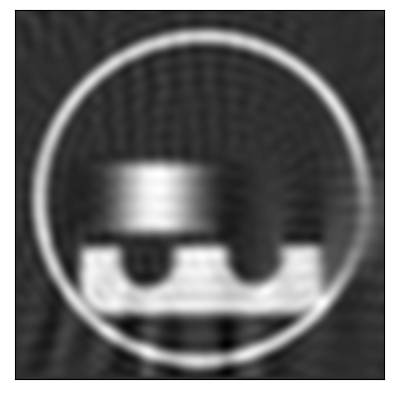} 
\includegraphics[width=0.19\textwidth,height=.18\textwidth]{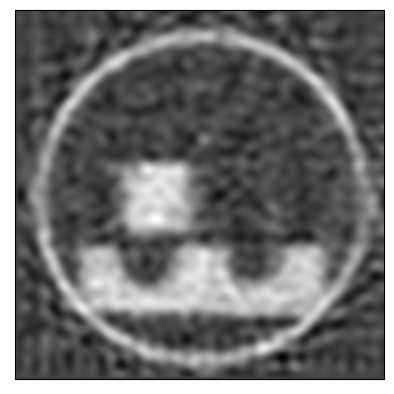} 
\end{tabular}
\vspace{-5pt}
\begin{tabular}{ccccc}
\includegraphics[width=0.19\textwidth,height=.18\textwidth]{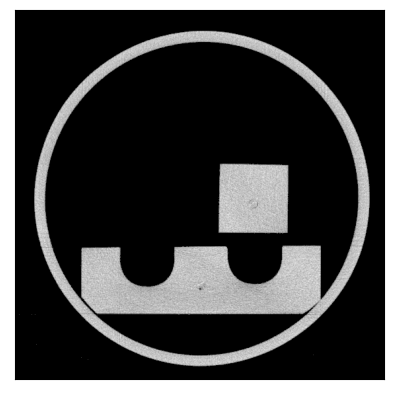}
\includegraphics[width=0.19\textwidth,height=.18\textwidth]{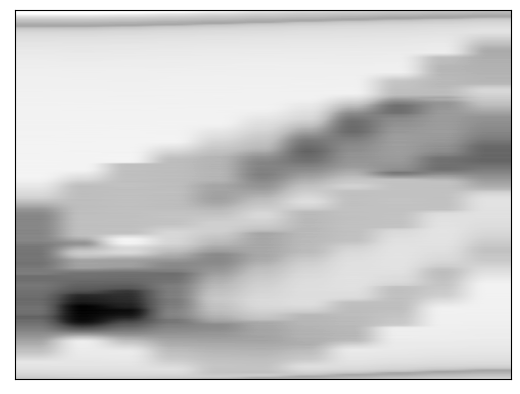}
\includegraphics[width=0.19\textwidth,height=.18\textwidth]{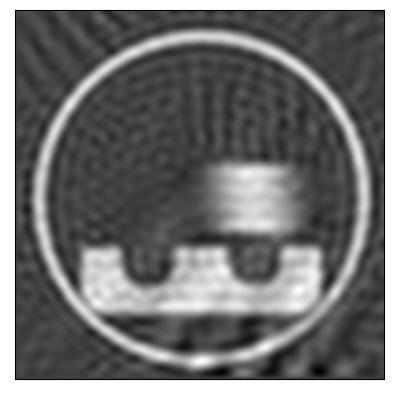}
\includegraphics[width=0.19\textwidth,height=.18\textwidth]{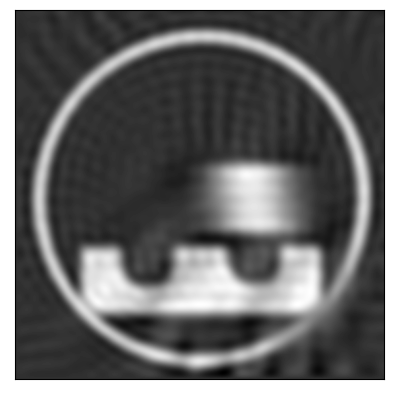}
\includegraphics[width=0.19\textwidth,height=.18\textwidth]{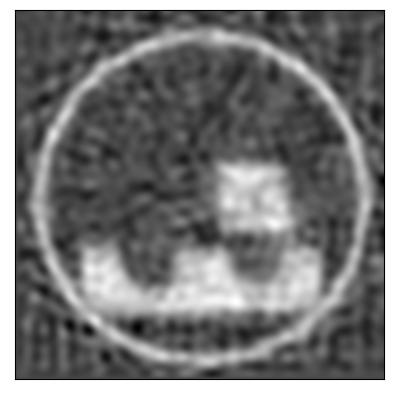}
\end{tabular}
\vspace{-5pt}
\begin{tabular}{ccccc}
\includegraphics[width=0.19\textwidth,height=.18\textwidth]{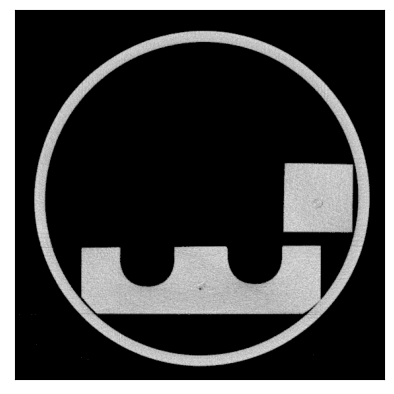}
\includegraphics[width=0.19\textwidth,height=.18\textwidth]{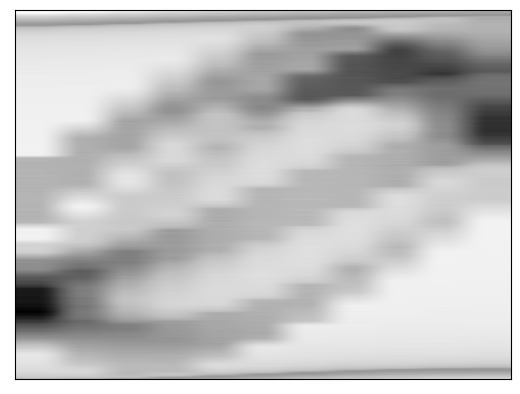}
\includegraphics[width=0.19\textwidth,height=.18\textwidth]{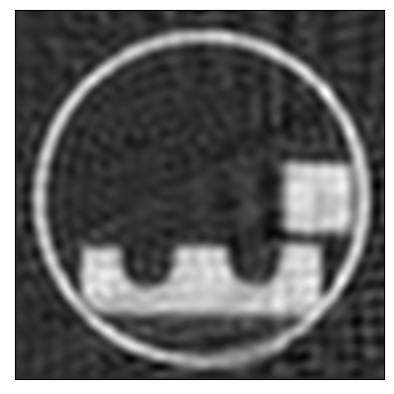}
\includegraphics[width=0.19\textwidth,height=.18\textwidth]{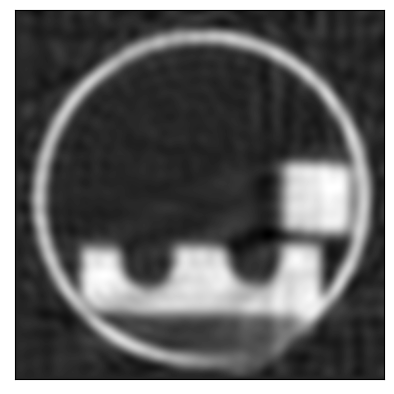}
\includegraphics[width=0.19\textwidth,height=.18\textwidth]{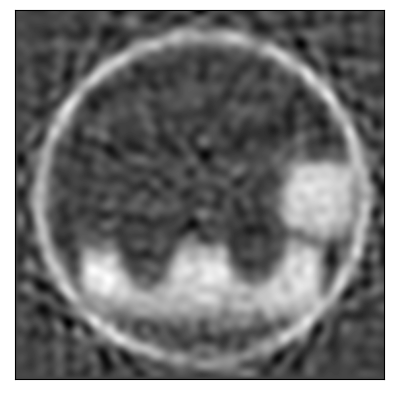}
\end{tabular}
\vspace{-5pt}
\caption{MAP reconstruction of dynamic STEMPO tomography in the whitened space. Columns from left to right: true images, sinograms, MAP estimates by STBP, STGP and time-uncorrelated models respectively. Rows from top to bottom: time step $j = 0, 6, 13, 19$.}
\label{fig:STEMPO_MAP_whiten}
\end{figure}

Firstly, we consider a simulated dynamic tomography named Spatio-TEmporal Motor-POwered (STEMPO) ground truth phantom from \cite{heikkila2022stempo}. 
The dataset contains 360 snapshots of CT images each of size $I=560\times 560$, and we choose $J=20$ at equal time intervals.
Each linear operator, $\mG_j\in\mbR^{8701 \times 313600}$, 
then projects the true image $u^\dagger(\bX,t_j)$ to a sinogram $\mG_j(u^\dagger(\bX, t_j)) \in \mbR^{791 \times 11}$ at $n_a=11$ angles ($x$-axis) with $n_s=791$ equally spaced X-ray detectors ($y$-axis) as shown in the second column of Figure \ref{fig:STEMPO_MAP_whiten}. Finally, following Model \eqref{eq:STBP_model}, Gaussian white noise $\vect\eps_j\sim \mN_{n_a n_s}(0, \sigma_j^2 \bI_{n_a n_s})$ is added to the signogram to obtain the observation $\bY_{j}=\mG_j(u^\dagger(\bX, t_j))+\vect\eps_j$ such that the 
noise level 
$\Vert\sigma_j\Vert_2/\Vert\mG_j\left(u^\dagger(\bX, t_j))\right)\Vert_2=0.01$.
The true images $\{u^\dagger(\bX,t_j)\}_{j=1}^J$ and the noisy observations $\{\bY_{j}\}_{j=1}^J$ at time $j=0,6,13,19$ are shown in the first two columns of Figure \ref{fig:STEMPO_MAP_whiten}, respectively.

We minimize the negative log-posterior densities \eqref{eq:logpost_whiten} in terms of the whitened coordinates $\bZeta$ for the three models 
to obtain the MAP estimates. The rightmost three columns of Figure \ref{fig:STEMPO_MAP_whiten} compare these MAP estimates obtained in the whitened space and mapped to the original space. STBP has the sharpest reconstruction that is the closest to the truth. However, the results by the other two models are either blurry (by STGP on the forth column) or noisy (by the time-uncorrelated model on the last column).
Table \ref{tab:STEMPO} confirms that the STBP model yields the best reconstruction with the lowest $\textrm{RLE}=32.17\%$ on average in 10 experiments repeated for different random seeds. Though their log-likelihood values are not comparable in the regularized optimization, the same advantage is supported by the high values in other quality measures such as PSNR and SSIM. 

On the other hand, the MAP estimates generated by minimizing the negative log-posterior \eqref{eq:logpost} in terms of the original parameters $\bXi$ are compared in Figure \ref{fig:STEMPO_MAP}. They have more than $40\%$ RLE's and are generally more blurry than those obtained in the whitened space (See Figure \ref{fig:STEMPO_MAP_whiten}). Such difference is also observed in Figure \ref{fig:STEMPO_err} where the objective functions and RLE's are compared between optimization in the original space (w.r.t. $\bXi$, left two panels) and optimization in the whitened space (w.r.t. $\bZ$, right two panels) for these three models: the latter yields better results within fewer iterations bearing lower errors, 
possibly due to faster exploration in the whitened space with variables de-correlated.
In general, STBP converges fastest to the lowest error state among the three models.

Lastly, we apply 
wn-$\infty$-mMALA (Algorithm \ref{alg:wn-infMC}) to sample from the posterior distributions of the two models with STBP and STGP priors, respectively (the result for time-uncorrelated prior is far worse and hence omitted) and compare their posterior estimates in Figure \ref{fig:STEMPO_MCMC}.
We generate 3000 samples and discard the first 1000 samples. The remaining 2000 samples are used to estimate the posterior means (the second and the third columns) and posterior standard deviations (the last two columns). Due to the large dimensionality ($560\times 560\times 20$) and limited number of samples, these posterior estimates tend to be noisy. The posterior mean estimates are not as good as their MAP estimates. Yet the posterior standard deviations by STBP (the forth column) provides uncertainty information with more clear spatial features than those by STGP (the last column).

\subsection{Navier-Stokes Inverse Problem}
Let us consider a complex non-linear inverse problem involving the following 2-d Navier-Stokes equation (NSE) for a viscous, incompressible fluid in vorticity form on 
$\mbT^2=(0,1)^2$:
\begin{align*}
\begin{split}
\partial_t w(x,t) + u(x,t) \cdot \nabla w(x,t) &= \nu \Delta w(x,t) + f(x), \qquad x \in (0,1)^2, t \in (0,T],  \\
\nabla \cdot u(x,t) &= 0, \qquad \qquad \qquad \qquad \quad x \in (0,1)^2, t \in [0,T],  \\
w(x,0) &= w_0(x), \qquad \qquad \qquad \quad x \in (0,1)^2 .
\end{split}
\end{align*}
where $u \in C([0,T]; H^r(\mbT^2; \mbR^2))$ for any $r>0$ 
is the velocity field, $w = \nabla \times u$ is the vorticity, $w_0 \in L^2(\mbT^2;\mbR)$ is the initial vorticity,  $\nu \in \mbR_+$ is the viscosity coefficient, and $f \in L^2(\mbT^2;\mbR)$ is the forcing function. 

Because NSE is computationally intensive to solve,
we 
build an emulator based on the Fourier operator neural network (FNO) \citep{Li_2021} that maps the vorticity up to time $T_0=10$ to the vorticity up to some later time $T>10$:
\begin{equation*}
\mG : C([0,T_0]; H^r(\mbT^2; \mbR^2)) \to C((T_0,T]; H^r(\mbT^2; \mbR^2)), \quad
 w|_{(0,1)^2\times[0,10]} \mapsto w|_{(0,1)^2\times(10,T]}.
\end{equation*}
One of the attractive features of FNO is that the neural network is built to learn operators defined on function spaces. Compared with traditional neural networks for simulating PDE solutions including CNN 
and PINNs \citep{RAISSI2019}, 
FNO is mesh-independent and very efficient for the inference of Bayesian inverse problem constrained by NSE.

In this example, we choose the viscosity $\nu=1 \mathrm{e}-3$ and set $T-T_0=30$. Since the target operator, $\mG^\dagger$, is time-dependent, we train a 3-d FNO (FNO-3d) based on 5000 pairs of input vorticity (for the first 10 unit time) and output vorticity (for the following 30 unit time) solved on $I=64\times 64$ spatial mesh (denoted as $\bX$) using the same network configuration as in \cite{Li_2021}. We initialize the vorticity $w_0$ with a (star-convex) polygon shown as in the top left of Figure \ref{fig:NSE_whiten_MAP} which also demonstrates a few snapshots of true vorticity trajectory, $w^\dagger|_{(0,1)^2\times[0,10]}(\bX, t_j)$, at $j=0,3,6,9$ in the first column.
We then observe data of vorticity $w|_{(0,1)^2\times(10,40]}(\bX, t_j)$ with $t_j\in(T_0,T]$ for $j=0,\cdots,29$ based on the true initial inputs $w^\dagger|_{(0,1)^2\times[0,10]}$, with Gaussian noise contamination, i.e., $\by_j = \mG(w^\dagger|_{(0,1)^2\times[0,10]}(\bX,t_j)) + \vect\eta_j$ with $\vect\eta_j\sim N(0,\Gamma_\mathrm{noise})$, and $\Gamma_\mathrm{noise}$ empirically estimated as in the previous example. A few time snapshots of the observed vorticity are illustrated in the second column of Figure \ref{fig:NSE_whiten_MAP}.
Figure \ref{fig:NSE_FNO} compares the trajectory emulated by the FNO network (lower row) against that solved by the classical PDE solver (upper row) in the observation time window $(T_0,T]$. Their visual difference is hardly discernible.


\renewcommand{\arraystretch}{0.6}
\begin{table}[t]
\caption{Comparison of MAP estimates of NSE trajectory generated by STBP, STGP and time-uncorrelated prior models in terms of RLE, 
log-likelihood, PSNR, and SSIM 
measures. Standard deviations (in bracket) are obtained by repeating the experiments for 10 times with different random seeds for initialization.}
\centering
\begin{tabular}{|c|c|c|c|}
\toprule
 & time-uncorrelated  & STGP  & STBP \\
 \midrule
RLE  & 0.7656 (7.60e-5) &   0.7457 (3.041e-5) & \cellcolor{lightgray} 0.6618 (1.07e-4) \\
log-lik  &   -229.18 (0.21) & -1586.11 (0.47) &  -173.33 (0.10)\\
PSNR &  15.7267 (8.62e-4) &  15.9555 (3.54e-4) & \cellcolor{gray} 16.9921 (1.40e-3) \\
SSIM &  0.1842 (7.84e-5) & 0.2213 (5.68e-5) & \cellcolor{gray} 0.3416 (1.67e-4) \\
\bottomrule
\end{tabular}
\label{tab:NSE}
\end{table}

Unlike the traditional time-dependent inverse problems seeking the solution of the initial condition $w_0$ alone, we are interested in the inverse solution of vorticity for an initial period, i.e., $w|_{(0,1)^2\times[0,10]}$. What is more, we want to obtain UQ for such spatiotemporal object in addition to its point estimate (MAP) using STBP, STGP and time-uncorrelated priors.
Their MAP estimates are compared in the last three columns of Figure \ref{fig:NSE_whiten_MAP}.
Note, this inverse problem for a spatiotemporal solution is much more challenging than the traditional inverse problem for just the initial condition based on the same amount of downstream observations. STBP still yields the inverse solution (the third column) closest to the true trajectory (the first column) among the three models, especially the initial condition at $t=0$ (the first row) which is the most difficult because it is the farthest from the observation window $(T_0,T]$. Note, due to the lack of temporal correlation, the solution trajectory from the time-uncorrelated prior model appears excessively erratic.
Table \ref{tab:NSE} further confirms that STBP prior model yields the best inverse solution with the lowest RLE, $66.18\%$, compared with the true trajectory, almost $10\%$ lower than the other two methods.
The high values of image reconstruction metrics also support the superiority of STBP model compared with the other two.
Figure \ref{fig:NSE_err} compares the optimization objective (the negative log-posterior) and the relative error as functions of iterations. STBP converges to lower RLE value, while STGP terminates earlier at a higher RLE value.

Lastly, because of the computational cost, we apply 
wn-$\infty$-HMC (Algorithm \ref{alg:wn-infMC}) instead of wn-$\infty$-mMALA for the UQ. We run 20,000 iterations, discard the first 5,000, and sub-sample one of every three. The remaining 5,000 samples are used to obtain posterior estimates illustrated in Figure \ref{fig:NSE_MCMC} comparing STBP model (the second and the forth columns) against STGP model (the third and the last columns). The posterior mean by STBP (the second column) is more noisy compared with that by STGP (the third row). However the posterior standard deviation by STBP (the forth column) is more informative than that of STGP (the last column).

\begin{figure}[t]
\begin{tabular}{cccc}
\qquad Observations & \qquad \qquad STBP & \quad \qquad \qquad STGP & \quad \qquad \qquad time-uncorrelated  \\
\end{tabular}
\vspace{-5pt}
\begin{tabular}{ccccc}
\includegraphics[width=0.24\textwidth,height=.18\textwidth]{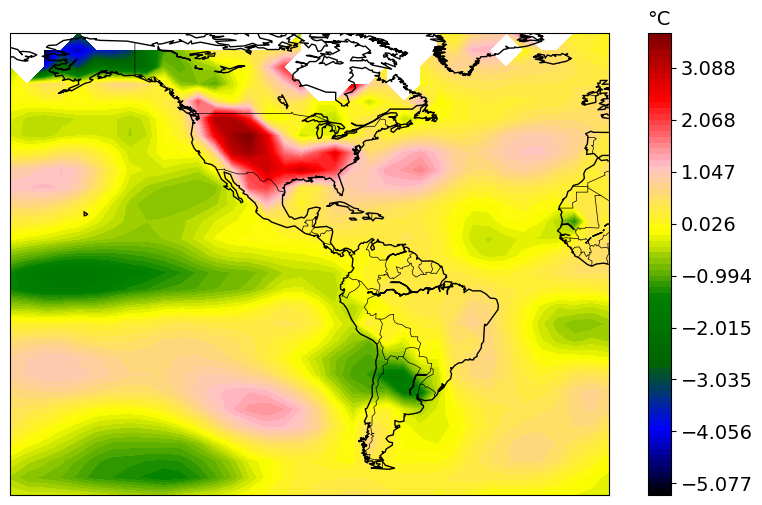}
\includegraphics[width=0.24\textwidth,height=.18\textwidth]{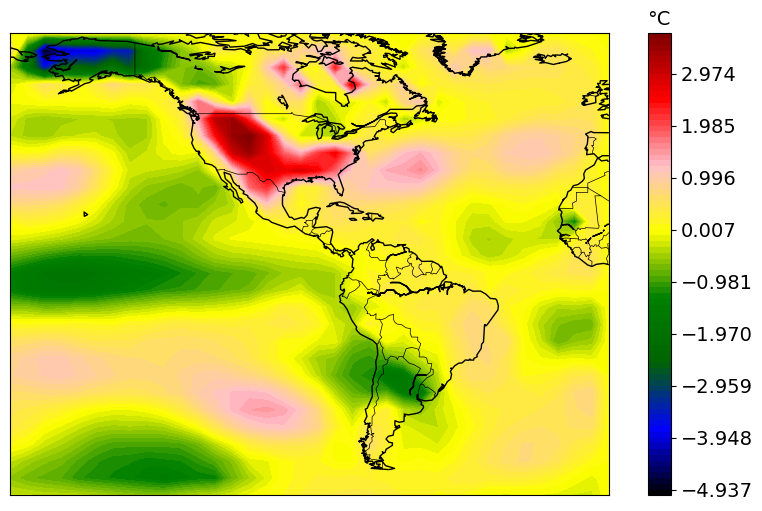}
\includegraphics[width=0.24\textwidth,height=.18\textwidth]{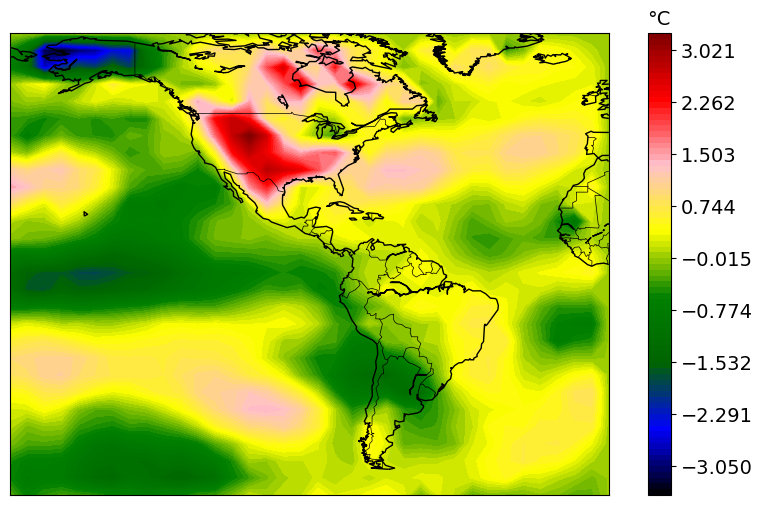}
\includegraphics[width=0.24\textwidth,height=.18\textwidth]{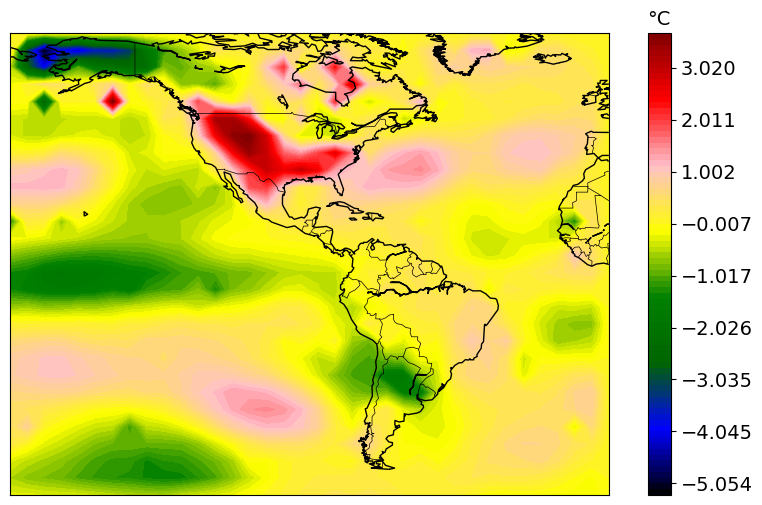}
\end{tabular}
\vspace{-5pt}
\begin{tabular}{ccccc}
\includegraphics[width=0.24\textwidth,height=.18\textwidth]{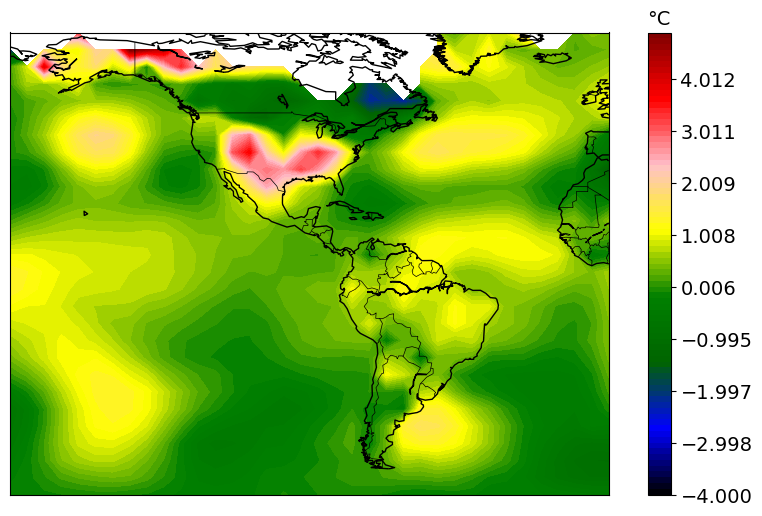}
\includegraphics[width=0.24\textwidth,height=.18\textwidth]{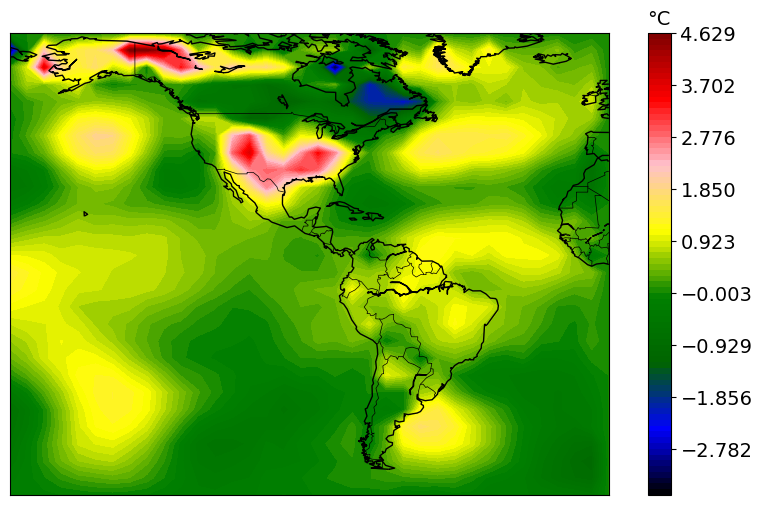}
\includegraphics[width=0.24\textwidth,height=.18\textwidth]{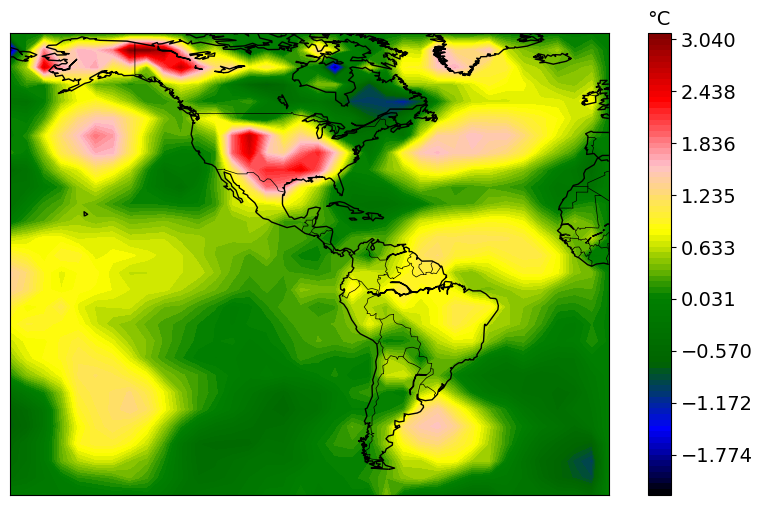}
\includegraphics[width=0.24\textwidth,height=.18\textwidth]{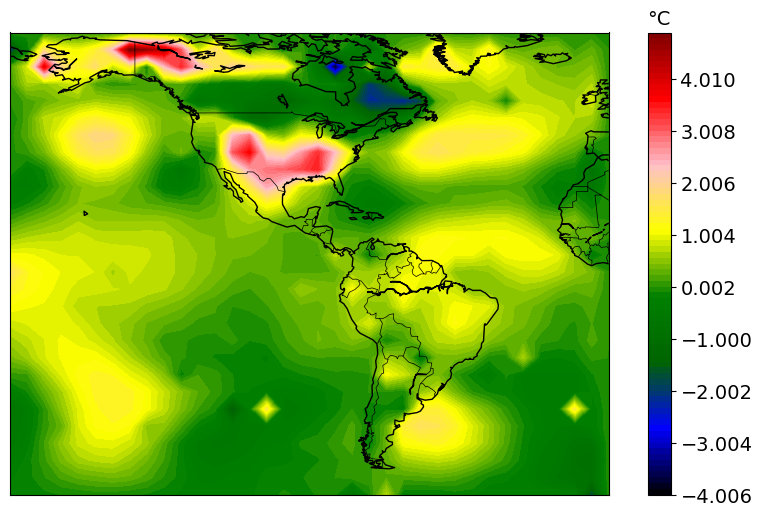}
\end{tabular}
\vspace{-5pt}
\begin{tabular}{ccccc}
\includegraphics[width=0.24\textwidth,height=.18\textwidth]{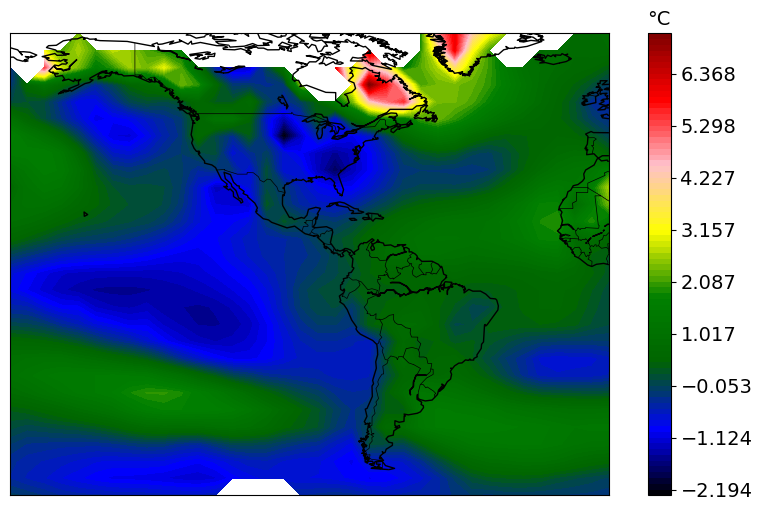}
\includegraphics[width=0.24\textwidth,height=.18\textwidth]{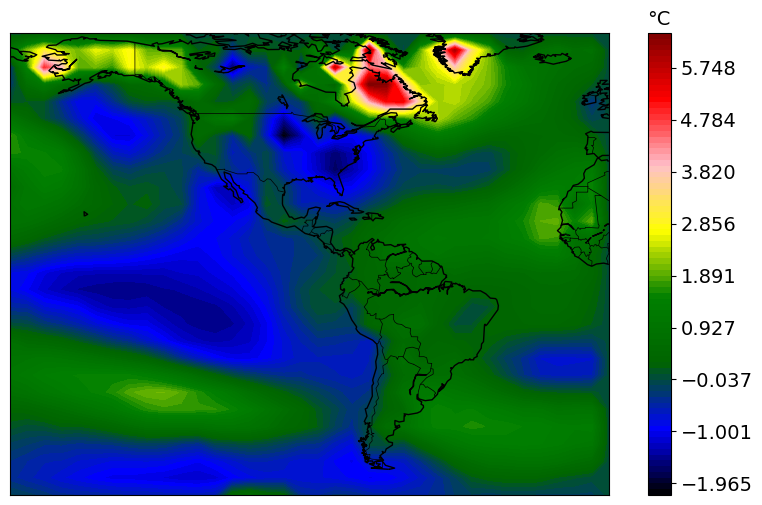}
\includegraphics[width=0.24\textwidth,height=.18\textwidth]{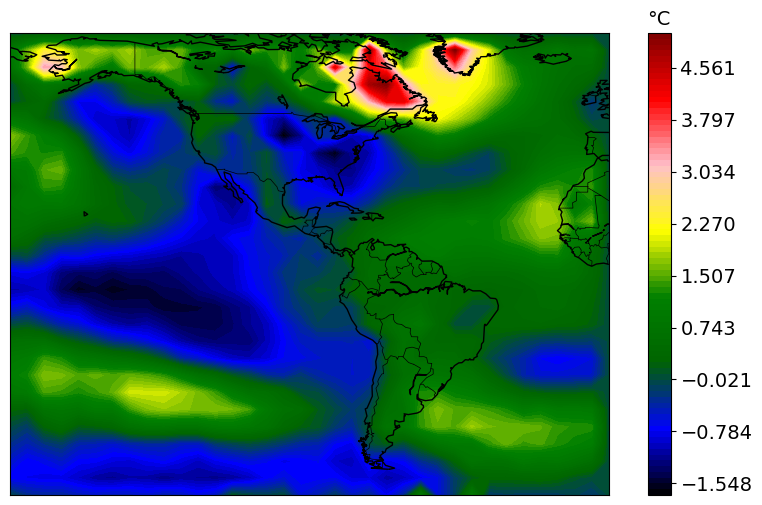}
\includegraphics[width=0.24\textwidth,height=.18\textwidth]{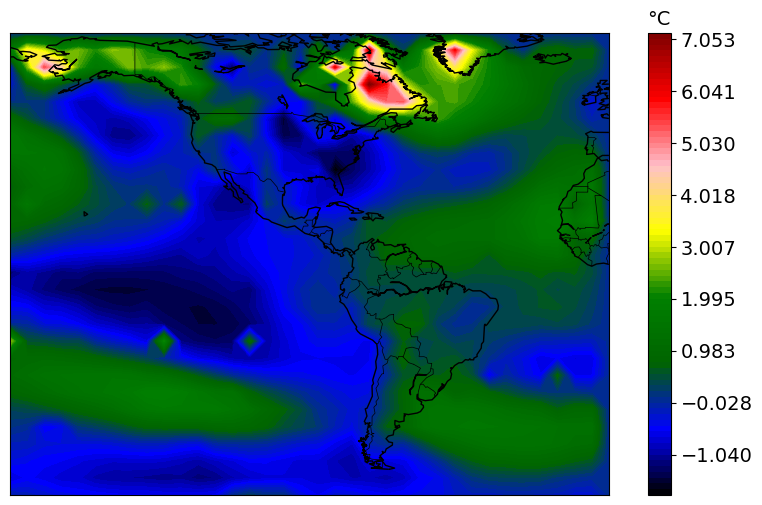}
\end{tabular}
\vspace{-5pt}
\begin{tabular}{ccccc}
\includegraphics[width=0.24\textwidth,height=.18\textwidth]{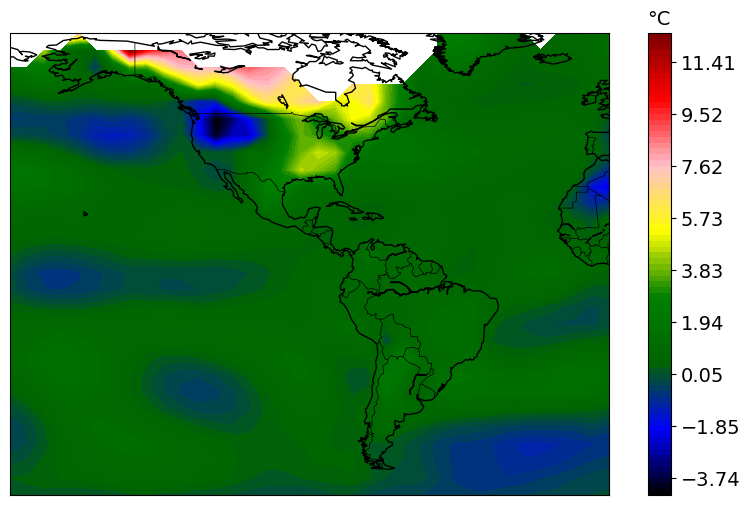}
\includegraphics[width=0.24\textwidth,height=.18\textwidth]{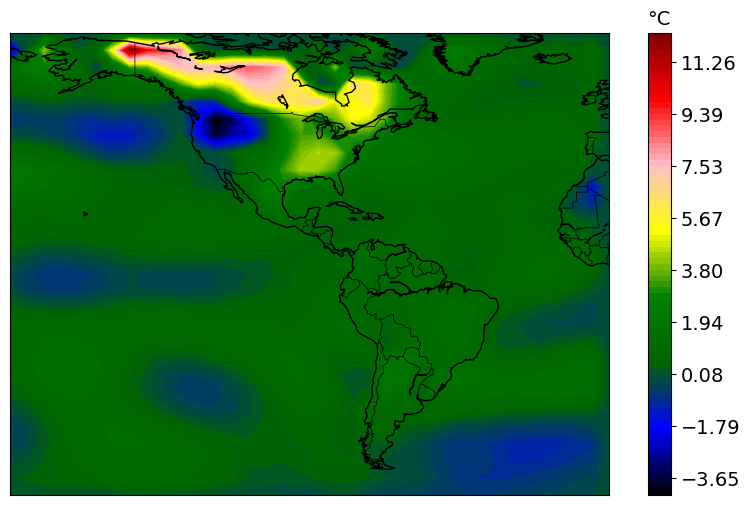}
\includegraphics[width=0.24\textwidth,height=.18\textwidth]{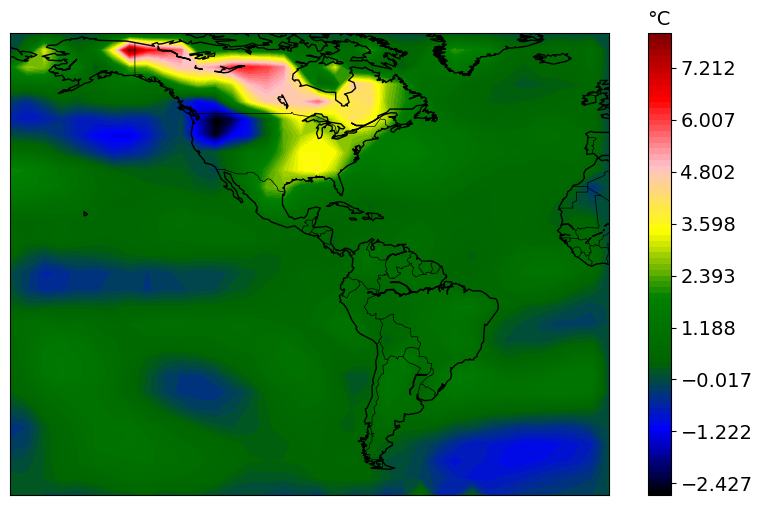}
\includegraphics[width=0.24\textwidth,height=.18\textwidth]{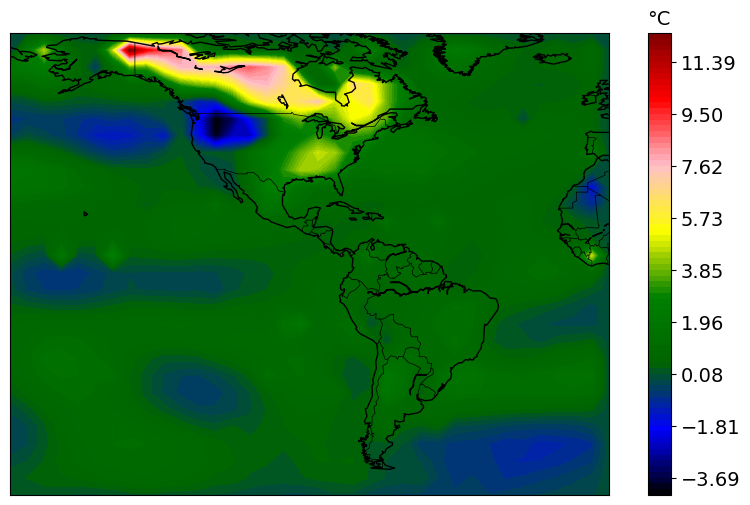}
\end{tabular}
\vspace{-5pt}
\caption{MAP reconstruction of NOAA temperature anomalies. Columns: observations with missing values, MAP estimates by STBP, STGP and time-uncorrelated models, respectively. Rows from top to bottom: time step $t_j = 1999, 2005, 2011$, and $2017$ (Januaries).}
\label{fig:noaatmp_MAP_whiten}
\end{figure}
\subsection{NOAA Temperature Anomalies}
In this section, we conduct a spatiotemporal analysis on a real dataset of monthly gridded temperature anomalies from U.S. National Oceanic and Atmospheric Administration (NOAA) \citep{Gu_2020, Gu_2022}. This dataset consists of the average air and marine temperate anomalies at 5 degrees longitude-latitude grids ranging from 182.5 to 357.5 in longitude and from -62.5 to 72.5 in latitude ($I=36\times 28$), with time spanning from Jan 1999 to Dec 2018 ($J=240$). 
This results in a dataset with $241,920$ items of which $11, 122$ are missing. The leftmost column of Figure \ref{fig:noaatmp_MAP_whiten} illustrates a few timestamps of the data with blank area corresponding to the missing values.
In addition to the missing data, we hold out $10\%$ random samples of $230,798$ valid entries and train STBP, STGP and time-uncorrelated models respectively on the remaining $207,719$ observations.
We then test their prediction performance on the held-out data and compare their MAP estimates in the right three columns of Figure \ref{fig:noaatmp_MAP_whiten}.
In all the snapshots, those generated by STBP prior model match the observations the closest. However, besides the unmatched ranges, the STGP model misses more spatial details in the recovered temperature data. Due to the negligence of temporal association, the time-uncorrelated model yields many noisy star-shaped irregular estimates deviating from the actual observations. The superior performance of STGP is further confirmed by the quantitative comparison summarized in Table \ref{tab:noaatmp}.
We also compare the MCMC estimates by applying wn-$\infty$-MALA (Algorithm \ref{alg:wn-infMC}) to the three models to generate 5000 samples after discarding the first 2000 samples. Their posterior median estimates are plotted in Figure \ref{fig:noaatmp_MCMC} with RLEs $32.4\%$ for STBP compared with $32.43\%$ for STGP and $40.83\%$ for time-uncorrelated model tested respectively on the $10\%$ held-out data.


\section{Conclusion}\label{sec:conclusion}

In this paper, we propose a nonparametric Bayesian framework to solve spatiotemporal inverse problems with inhomogeneous data, such as sequential images with edges.
Our proposed STBPs are generalizations of BP from the spatial to spatiotemporal domain. The key idea is to replace random coefficients (following a $q$-exponential distribution) in the series definition of BP with the recently proposed $Q$-EP \citep{Li_2023}
to account for the temporal correlations among spatial function images through a covariance kernel, similarly as in GP.
Moreover, STBP controls the regularization of posterior solutions through a parameter $q\in[1,2]$ and includes STGP as a special case ($q=2$). 

We conduct a thorough theoretical investigation regarding well-definedness, 
series representation of STBP priors and their 
posterior 
properties 
to justify the suitability and superiority of the proposed methodology. 
To address the challenges of posterior inference, we propose dimension-independent MCMC algorithms based on a new white noise representation for series-based priors \citep{chen2018dimension}. 
Through extensive numerical experiments from various 
spatiotemporal inverse problems we demonstrate that 
STBP ($q=1$) has the advantage in handling spatial inhomogeneity over STGP (which tends to be oversmooth) and in capturing temporal correlations over a time-uncorrelated approach.

Several directions remain open for future research.
The inference can be sped up by 
more efficient variational Bayes approach.
The equal constraint of the STBP prior on the spatial ($q$) and temporal ($p$) regularization parameters 
(Section \ref{sec:STBP}) can be relaxed to allow for independent control on the spatial and temporal regularities of the posterior solution. 
We also aim to extend the current work to the non-convex regime for $q\in (0,1)$, 
which imposes more regularization (refer to Figure \ref{fig:simulation_MAP_whiten_q}) with considerably more complexity.

\section*{Acknowledgments}
The authors gratefully acknowledge NSF grant DMS-2134256 (SL), NSF awards 2202846 and DMS-2410699 (MP).
The authors report there are no competing interests to declare.

{\footnotesize
\setlength{\bibsep}{7pt}
\setlength{\baselineskip}{7pt}
\bibliographystyle{chicago}
\bibliography{ref}
}

\clearpage
\appendix
\begin{center}
{\large\bf SUPPLEMENTARY MATERIAL of\\ ``Spatiotemporal Besov Priors for Bayesian Inverse Problems"}
\end{center}

\numberwithin{equation}{section}
\numberwithin{asump}{section}
\numberwithin{lem}{section}
\numberwithin{prop}{section}
\numberwithin{thm}{section}
\numberwithin{figure}{section}
\numberwithin{table}{section}
\numberwithin{algorithm}{section}

\section{Proofs}\label{apx:proof}
\subsection{Theorems of BP, STBP Priors}\label{apx:priorthm}

\Lqbound*
\begin{proof}[Proof of Theorem \ref{thm:Lqbound}]
\label{apx:Lqbound}
Based on \eqref{eq:q_exp}, it is straightforward to verify
\begin{equation*}
    \mbE[\Vert u\Vert_{s',q}^q] = \sum_{\ell=1}^\infty \ell^{\tau_q(s')q}\mbE|u_\ell|^q =\kappa^{-1} \mbE[|\xi_1|^q] \sum_{\ell=1}^\infty \ell^{(\tau_q(s')-\tau_q(s))q} <\infty
\end{equation*}
if $(\tau_q(s')-\tau_q(s))q = \frac{s'-s}{d}q<-1$, i.e. $s'<s-\frac{d}{q}$.
\end{proof}

\qepint*
\begin{proof}[Proof of Theorem \ref{thm:qepint}]
\label{apx:qepint}
Note $r(\xi_\ell)^{\frac{q}{2}}=\lambda_\ell^{-\frac{q}{2}} |\xi_\ell|^q\sim \chi^2(1)$ for all $\ell\in\mbN$ by Proposition \ref{prop:QED_stochrep}.
Denote $\chi^2_\ell := \lambda_\ell^{-\frac{q}{2}} |\xi_\ell|^q \overset{iid}{\sim} \chi^2(1)$.
Hence $\Vert \xi\Vert_{s',q}^q = \sum_{\ell=1}^\infty \ell^{\tau_q(s')q}\lambda_\ell^{\frac{q}{2}} \chi^2_\ell$ becomes an infinite mixture of chi-squared random variables whose density is analytically intractable.
Yet we have
\begin{equation*}
    \mbE[\Vert \xi(\cdot)\Vert_{s',q}^q] = \sum_{\ell=1}^\infty \ell^{\tau_q(s')q}\lambda_\ell^{\frac{q}{2}} \mbE[\chi^2_\ell] = \sum_{\ell=1}^\infty \ell^{\tau_q(s')q}\lambda_\ell^{\frac{q}{2}} <\infty
\end{equation*}
if Assumption \ref{asmp:eigen_summable}-(i) holds.
Hence we have proved the first conclusion.
From above we have $\mbE[\Vert \xi(\cdot)\Vert_{q}^q] = \Vert\lambda\Vert_{\frac{q}{2}}^{\frac{q}{2}} <\infty$ if Assumption \ref{asmp:eigen_summable}-(ii) holds.
Thus it completes the proof.
\end{proof}

\Fernique*
\begin{proof}[Proof of Theorem \ref{thm:Fernique}]
\label{apx:Fernique}
We complete the proof by showing $(iii)\implies (ii)\implies (i) \implies (iii)$.
First, by Assumption-\ref{asmp:decaying_constant} and the proof of Theorem \ref{thm:qepint} we have
\begin{equation*}
    \Vert u\Vert_{s',q}^q = \sum_{\ell=1}^\infty \ell^{\tau_q(s')q} \Vert u_\ell(\cdot)\Vert_q^q = \sum_{\ell=1}^\infty \ell^{(\tau_q(s')-\tau_q(s))q} \Vert \xi_\ell(\cdot)\Vert_q^q = \sum_{\ell=1}^\infty \ell^{\frac{s'-s}{d}q} \sum_{\ell'=1}^\infty \lambda_{\ell'}^{\frac{q}{2}} \chi^2_{\ell\ell'} ,
\end{equation*}
where $\chi^2_{\ell\ell'}\overset{iid}{\sim}\chi^2(1)$.
Denote $\alpha_{\ell\ell'} = \alpha \ell^{\frac{s'-s}{d}q} \lambda_{\ell'}^{\frac{q}{2}}$. Then we have
\begin{equation*}
    \mbE[\exp(\alpha_{\ell\ell'} \chi^2_{\ell\ell'})] = M_{\chi^2(1)}(\alpha_{\ell\ell'}) = \left[1 -2\alpha_{\ell\ell'}\right]^{-\half}, \quad for \; \alpha_{\ell\ell'}<\half .
\end{equation*}

\paragraph{$(iii)\implies (ii)$.}
Now that
\begin{equation*}
    \mbE[\exp(\alpha \Vert u \Vert_{s',q}^q)] = \mbE\left[\exp\left(\sum_{\ell=1}^\infty \sum_{\ell'=1}^\infty \alpha_{\ell\ell'} \chi^2_{\ell\ell'})\right)\right] 
    = \prod_{\ell=1}^\infty \prod_{\ell'=1}^\infty \left[1 -2\alpha_{\ell\ell'}\right]^{-\half} .
\end{equation*}
Assume $\alpha>0$, we have each item in the product bigger than 1. To bound such infinite product, it suffices to bound the following infinite sum
\begin{equation*}
    \sum_{\ell=1}^\infty \sum_{\ell'=1}^\infty \left(\left[1 -2\alpha_{\ell\ell'}\right]^{-\half} -1\right) = \sum_{\ell=1}^\infty \sum_{\ell'=1}^\infty \frac{2\alpha_{\ell\ell'}}{\sqrt{1-2\alpha_{\ell\ell'}} + 1-2\alpha_{\ell\ell'}} \leq \frac{2\alpha \sum_{\ell=1}^\infty \ell^{\frac{s'-s}{d}q} \Vert\lambda\Vert_{\frac{q}{2}}^{\frac{q}{2}}}{\sqrt{1-2\alpha \Vert \lambda\Vert_\infty^{\frac{q}{2}}} + 1-2\alpha \Vert \lambda\Vert_\infty^{\frac{q}{2}}} .
\end{equation*}
By Assumptions \ref{asmp:eigen_summable}-(ii),
it is finite provided that $\frac{s'-s}{d}q<-1$ and $1-2\alpha \Vert \lambda\Vert_\infty^{\frac{q}{2}}>0$, that is, $s'<s-\frac{d}{q}$ and $\alpha<\half \Vert \lambda\Vert_\infty^{-\frac{q}{2}}$ where $\Vert \lambda\Vert_\infty = \sup_\ell \lambda_\ell$.

\paragraph{$(ii)\implies (i) \implies (iii)$.} One can follow \cite[Lemma 10 of][]{Lassas_2009} or \cite[Theorem 5 of][]{dashti2017} for the same argument.
\end{proof}

\embeddings*
\begin{proof}[Proof of Proposition \ref{prop:embeddings}]
\label{apx:embeddings}
If $q^\dagger=q$, we have $B^{s^\dagger,q^\dagger}(\mZ)\subset B^{s',q}(\mZ)$ when $s'<s^\dagger$; 
if $q^\dagger>q$, for $u\in B^{s^\dagger,q^\dagger}(\mZ)$ by H\"older inequality with $(\frac{q^\dagger}{q}, \frac{q^\dagger}{q^\dagger-q})$,
\begin{equation}\label{eq:qs_bigg_q}
\begin{aligned}
    \Vert u\Vert_{s',q}^q &= \sum_{\ell=1}^\infty \ell^{\tau_q(s')q}\Vert u_\ell(\cdot)\Vert_q^q \leq \sum_{\ell=1}^\infty \ell^{\tau_{q^\dagger}(s^\dagger)q}\Vert u_\ell(\cdot)\Vert_{q^\dagger}^q \ell^{(\tau_q(s')-\tau_{q^\dagger}(s^\dagger))q}\\
    &\leq \Vert u\Vert_{s^\dagger,q^\dagger}^q \left[\sum_{\ell=1}^\infty \ell^{(\tau_q(s')-\tau_{q^\dagger}(s^\dagger))q q^\dagger/(q^\dagger -q)} \right]^{1-\frac{q}{q^\dagger}} <\infty
\end{aligned}
\end{equation}
holds when $(\tau_q(s')-\tau_{q^\dagger}(s^\dagger))q q^\dagger/(q^\dagger -q) = \frac{(s'-s^\dagger)/d}{\frac{1}{q}-\frac{1}{q^\dagger}}-1<-1$, i.e., $s'<s^\dagger$;
if $q^\dagger<q$, for $u\in B^{s^\dagger,q^\dagger,q}(\mZ)$, we have
\begin{equation}\label{eq:qs_less_q}
\begin{aligned}
\Vert u \Vert_{s',q}^q &= \sum_\ell \ell^{\tau_q(s')q}\Vert u_\ell(\cdot)\Vert_q^q = \sum_\ell \ell^{\tau_{q^\dagger}(s^\dagger)q^\dagger}\Vert u_\ell(\cdot)\Vert_q^{q^\dagger} \ell^{(\tau_q(s')q-\tau_{q^\dagger}(s^\dagger)q^\dagger)} \Vert u_\ell(\cdot)\Vert_q^{(q-q^\dagger)} \\
&\leq \sum_\ell \ell^{\tau_{q^\dagger}(s^\dagger)q^\dagger}\Vert u_\ell(\cdot)\Vert_q^{q^\dagger} \ell^{(\tau_q(s')q-\tau_{q^\dagger}(s^\dagger)q^\dagger)} \Vert u\Vert_{s^\dagger, q^\dagger, q}^{(q-q^\dagger)} \ell^{-\tau_{q^\dagger}(s^\dagger)(q-q^\dagger)} \\
&= \Vert u\Vert_{s^\dagger, q^\dagger, q}^{(q-q^\dagger)} \sum_\ell \ell^{\tau_{q^\dagger}(s^\dagger)q^\dagger}\Vert u_\ell(\cdot)\Vert_q^{q^\dagger} \ell^{(\tau_q(s')-\tau_{q^\dagger}(s^\dagger))q} 
\lesssim \Vert u\Vert_{s^\dagger, q^\dagger, q}^q
\end{aligned}
\end{equation}
hold when $(\tau_q(s')-\tau_{q^\dagger}(s^\dagger))q=\left(\frac{s'-s^\dagger}{d}-\frac{1}{q}+\frac{1}{q^\dagger}\right)q<0$, i.e., $s'<s^\dagger +\frac{d}{q}-\frac{d}{q^\dagger}$.
\end{proof}

\KLSTBP*
\begin{proof}[Proof of Theorem \ref{thm:KLSTBP}]
\label{apx:KLSTBP}
The series expansion \eqref{eq:STBP_serexp} is the result of applying Theorem \ref{thm:KL} to each $\xi_\ell(\cdot)\overset{iid}{\sim} \qEP(0,\mC)$ and the convergence is in $L^2_\mbP(\Omega, B^{s,2}(\mZ))$.
We then directly compute the spatiotemporal covariance
\begin{equation*}
\begin{aligned}
     & ~~~~\mCv(u(\bz), u(\bz')) \\
    &=\E(u(\bz) u(\bz')) = \E\left[\sum_{\ell=1}^{\infty} \gamma_\ell \xi_\ell(t)\phi_\ell(\bx) \sum_{\ell'=1}^{\infty} \gamma_{\ell'} \xi_{\ell'}(t')\phi_{\ell'}(\bx')\right]\\
    &= \sum_{\ell,\ell'=1}^{\infty}\gamma_\ell \gamma_{\ell'} \phi_\ell(\bx)\phi_{\ell'}(\bx') \E[\xi_\ell(t)\xi_{\ell'}(t')] 
    = \sum_{\ell=1}^{\infty}\gamma_\ell^2 \phi_\ell(\bx)\phi_\ell(\bx') \E[\xi_\ell(t)\xi_\ell(t')] \\
    &= \sum_{\ell=1}^{\infty}\gamma_\ell^2 \phi_\ell(\bx)\phi_\ell(\bx') \mC(t,t') ,
\end{aligned}
\end{equation*}
where we use the iid assumption of $\xi_\ell(\cdot)$ so that $\E[\xi_\ell(t)\xi_{\ell'}(t')]=\E[\xi_\ell(t)\xi_\ell(t')]\delta_{\ell\ell'}$.
\end{proof}

The regularity of an STBP random draw $u(\bx ,t)$ as in \eqref{eq:STBP_serexp} also depends on the properties of spatial  ($\{\phi_\ell\}_{\ell=1}^\infty$) and temporal ($\{\psi_{\ell'}\}_{\ell'=1}^\infty$) bases.
To study its H\"older continuity, we make the following assumption that mainly states that the bases are H\"older continuous with summable Lipschitz constants.
\begin{asump}\label{asmp:bases}
In the series expansion \eqref{eq:STBP_serexp}, we assume the spatial  ($\{\phi_\ell\}_{\ell=1}^\infty$) and temporal ($\{\psi_{\ell'}\}_{\ell'=1}^\infty$) bases satisfy for $\alpha>0$ and $s'<s-\frac{d}{q}$:
\begin{enumerate}[label=(\roman*),nosep]
\item for $\forall\bx,\bx'\in\mX$, the following holds
\begin{equation*}
|\phi_\ell(\bx)-\phi_\ell(\bx')| \leq L(\phi_\ell)|\bx-\bx'|^{\alpha+\frac{d}{q}}, \quad \sum_{\ell=1}^\infty \ell^{\frac{s'-s}{d}q} L(\phi_\ell)<\infty, \quad \sum_{\ell=1}^\infty \ell^{\frac{s'-s}{d}q} \Vert \phi_\ell(\cdot)\Vert_\infty^q<\infty;
\end{equation*}
\item for $\forall t,t'\in\mT$, the following holds
\begin{equation*}
|\psi_{\ell'}(t)-\psi_{\ell'}(t')| \leq L(\psi_\ell)|t-t'|^{\alpha+\frac{d}{q}}, \quad \sum_{\ell'=1}^\infty \lambda_{\ell'}^{\frac{q}{2}} L(\psi_{\ell'})<\infty.
\end{equation*}
\end{enumerate}
\end{asump}

The following theorem regarding the H\"older continuity of STBP random functions can be proved by the Kolmogorov continuity test \citep[Theorem 30 in Section A.2.5 of][]{dashti2017}.
\begin{thm}[H\"older Continuity]
\label{thm:continuity}
Let $u\sim \STBP(\mC, B^{s,q}(\mZ))$ as in \eqref{eq:STBP_randf} satisfying both Assumptions \ref{asmp:eigen_summable}-(ii) and \ref{asmp:decaying_constant}. Suppose the spatial  ($\{\phi_\ell\}_{\ell=1}^\infty$) and temporal ($\{\psi_{\ell'}\}_{\ell'=1}^\infty$) bases in the series representation \eqref{eq:STBP_serexp} satisfy Assumption \ref{asmp:bases}.
Then for any $\beta<\alpha$, there exists a version \footnote{A version of stochastic process $u(\bz)$ is $\tilde u(\bz)$ such that $\Pi[\tilde u(\bz)=u(\bz)]=1$ for $\forall \bz\in \mZ$.}  $\tilde u(\bz)$ of $u(\bz)$ in $C^{0,\beta}(\mZ, B^{s',q}(\mZ))$.
\end{thm}
\begin{proof}
Based on the series representation of $u(\cdot)$ in \eqref{eq:STBP_randf}, we have by Jensen's inequality
\begin{equation*}
\begin{aligned}
    \mbE[\Vert u(\bz)-u(\bz')\Vert_{\tau_q(s'),q}^q] &\leq \sum_{\ell=1}^\infty \mbE[ \ell^{(\tau_q(s')-\tau_q(s))q}| \phi_\ell(\bx)\xi_\ell(t)-\phi_\ell(\bx')\xi_\ell(t') |^q] \\
    &\lesssim \sum_{\ell=1}^\infty \ell^{\frac{s'-s}{d}q} \mbE[|\Delta\phi_\ell|^q|\xi_\ell(t)|^q + |\phi_\ell(\bx')|^q|\Delta\xi_\ell|^q ] \quad (convexity \; of \; |\cdot|^q) \\
    &\lesssim \sum_{\ell=1}^\infty \ell^{\frac{s'-s}{d}q} [ L(\phi_\ell)|\bx-\bx'|^{q\alpha+d} + \Vert \phi_\ell(\cdot)\Vert_\infty^q \mbE|\Delta\xi_\ell|^q ] ,
\end{aligned}
\end{equation*}
where $\Delta \phi=\phi(\bx)-\phi(\bx')$ and $\Delta \xi=\xi(t)-\xi(t')$.
Now based on the series representation $\xi(\cdot)$ in \eqref{eq:KL}, by Jensen's inequality and the proof of Theorem \ref{thm:qepint} we have further for $\forall \ell\in\mbN$ (due to the iid assumption of $\xi_\ell(\cdot)$ in Assumption \ref{asmp:decaying_constant})
\begin{equation*}
    \mbE[|\Delta\xi_\ell|^q] \leq \sum_{\ell'=1}^\infty \mbE [|\xi_{\ell'}|^q |\psi_{\ell'}(t)-\psi_{\ell'}(t')|^q ] \leq \sum_{\ell'=1}^\infty \lambda_{\ell'}^{\frac{q}{2}} L(\psi_{\ell'})|t-t'|^{q\alpha+d} .
\end{equation*}
Therefore by Assumption \ref{asmp:bases} we have
\begin{equation*}
\begin{aligned}
    &\mbE[\Vert u(\bz)-u(\bz')\Vert_{\tau_q(s'),q}^q] \\
    &\lesssim |\bx-\bx'|^{q\alpha+d} \sum_{\ell=1}^\infty \ell^{\frac{s'-s}{d}q} L(\phi_\ell) + |t-t'|^{q\alpha+d} \sum_{\ell=1}^\infty \ell^{\frac{s'-s}{d}q} \Vert \phi_\ell(\cdot)\Vert_\infty^q \sum_{\ell'=1}^\infty \lambda_{\ell'}^{\frac{q}{2}} L(\psi_{\ell'}) 
    \lesssim |\bz -\bz'|^{q\alpha+d} .
\end{aligned}
\end{equation*}
The conclusion follows by the Kolmogorov continuity theorem \citep{{dashti2017}}.
\end{proof}

\subsection{Posterior Theorems of Inverse Problems with STBP Priors}\label{apx:posthm}

Here we re-examine the well-definedness and well-posedness of the posterior measure. 
Following \cite{Dashti_2012}, we impose some additional conditions on the potential (negative log-likelihood) function $\Phi:B^{s',q}(\mZ)\times\mbY\to\mbR$ as in \eqref{eq:potential} regarding its lower (i) and upper (ii) bounds, and the Lipschitz continuity in 
$y$ (iii) in the following assumption.
\begin{asump}\label{asmp:welldefpos}
The potential function $\Phi: B^{s',q}(\mZ)\times \mbY\to \mbR$ satisfies:
\begin{enumerate}[label=(\roman*),nosep]
\item there is an $\alpha_1>0$ and for every $r>0$, an $M\in \mbR$, such that for all $u\in B^{s',q}(\mZ)$, and for all $y\in \mbY$ such that $\Vert y\Vert_\mbY<r$, 
\begin{equation*}
    \Phi(u, y)\geq M-\alpha_1\Vert u\Vert_{s',q} \;;
\end{equation*}
\item for every $r>0$ there exists $K=K(r)>0$ such that for all $u\in B^{s',q}(\mZ), y\in\mbY$ with $\max\{\Vert u\Vert_{s',q}, \Vert y\Vert_\mbY\}<r$, 
\begin{equation*}
    \Phi(u,y)\leq K \;;
\end{equation*}
\item there is an $\alpha_2>0$ and for every $r>0$ a $C\in\mbR$ such that for all $y_1,y_2\in \mbY$ with $\max\{\Vert y_1\Vert_\mbY, \Vert y_2\Vert_\mbY\}<r$ and for every $u\in B^{s',q}(\mZ)$, 
\begin{equation*}
    |\Phi(u,y_1)-\Phi(u,y_2)|\leq \exp(\alpha_2\Vert u\Vert_{s',q}+C)\Vert y_1-y_2\Vert_\mbY \;.
\end{equation*}
\end{enumerate}
\end{asump}


\begin{thm}[Well-definedness of Posterior]
\label{thm:bayespost}
Let the potential function $\Phi$ in \eqref{eq:potential} satisfy Assumption \ref{asmp:welldefpos} (i)-(ii) and Assumption \ref{asmp:Lipschitz_u}. 
If $\Pi$ is an $\STBP(\mC, B^{s,q}(\mZ))$ measure with $s>s'+\frac{d}{q}$, 
then $\Pi(\cdot|y)\ll \Pi$ and satisfies
\begin{equation}\label{eq:Bayes_post}
\frac{d\Pi(\cdot|y)}{d\Pi}(u) = \frac{1}{Z(y)}\exp(-\Phi(u;y)),
\end{equation}
with the normalizing factor $0<Z(y)=\int_{B^{s',q}(\mZ)} \exp(-\Phi(u;y))\Pi(du)<\infty$. 
\end{thm}
\setstcolor{red}
\begin{proof}
The proof is based on \citep[Theorem 3.2 of][]{Dashti_2012} and \citep[Theorem 14 of][]{dashti2017}.
Define $\pi_0(du, dy)=\Pi(du)\otimes \mbQ_0(dy)$ and $\pi(du, dy)=\Pi(du)\otimes \mbQ_u(dy)$. 
We assume $\mbQ_0\ll \mbQ$ and the Radon-Nikodym derivative \eqref{eq:potential} holds for $\Pi$-a.s.. Thus, for fixed $u$, $\Phi(u;\cdot): \mbY\to \mbR$ is $\mbQ_0$-measurable and $\int_\mbY \exp(-\Phi(u; y))\mbQ_0(dy)=1$.
On other hand, by Assumption \ref{asmp:Lipschitz_u}, 
we have $\Phi(\cdot; y): B^{s',q}(\mZ) \to \mbR$ is continuous on $B^{s',q}(\mZ)$. By Corollary \ref{cor:STBP_supp} $\Pi(B^{s',q}(\mZ))=1$. Hence $\Phi(\cdot; y)$ is $\Pi$-measurable. Therefore, $\Phi$ is $\pi_0$-measurable and $\pi\ll \pi_0$ with
\begin{equation*}
    \frac{d\pi}{d\pi_0}(u; y) = \exp(-\Phi(u;y)), \quad \int_{B^{s',q}(\mZ)\times \mbY} \exp(-\Phi(u;y)) \pi_0(du, dy) =1 .
\end{equation*}

Then by 
\citep[Theorem 13 of][]{dashti2017}, the conditional distribution $\Pi(du|y):=\pi(du, dy'|y'=y)\ll \Pi(du)=\pi_0(du, dy'|y'=y)$ due to the definition of $\pi_0$. The same Lemma/Theorem implies \eqref{eq:Bayes_post} if the normalizing constant $Z(y)>0$.

First, by Assumption \ref{asmp:welldefpos}-(i), we have
\begin{equation*}
    Z(y) = \int_{B^{s',q}(\mZ)} \exp(-\Phi(u;y))\Pi(du) \leq \int_{B^{s',q}(\mZ)} \exp(-M+\alpha_1\Vert u\Vert_{s',q}) \Pi(du) <\infty ,
\end{equation*}
where the boundedness is the result of Theorem \ref{thm:Fernique}-(ii).
Now we show $Z(y)>0$. By Theorem \ref{thm:Lqbound_ST}, 
we have $R = \mbE\Vert u\Vert_{s',q} <\infty$. 
Since $\Vert u\Vert_{s',q}$ is nonnegative, we have $\Pi(\Vert u\Vert_{s',q}<R)>0$. Let $r=\max\{\Vert y\Vert_\mbY, R\}$. By Assumption \ref{asmp:welldefpos}-(ii), we have
\begin{equation*}
    \int_{B^{s',q}(\mZ)} \exp(-\Phi(u;y))\Pi(du) \geq \int_{\Vert u\Vert_{s',q}<R} \exp(-K)\Pi(du) = \exp(-K) \Pi(\Vert u\Vert_{s',q}<R) > 0 .
\end{equation*}
\end{proof}

Now we show the well-posedness of the posterior measure $\Pi(\cdot|y)$ with respect to the data $y$. Define the Hellinger metric as $d_H(\mu,\mu') = \sqrt{\half\int\left(\sqrt{\frac{d\mu}{d\nu}}-\sqrt{\frac{d\mu'}{d\nu}}\right)^2 d\nu}$. Note we require $\mu, \mu'\ll \nu$, but this definition is independent of the choice of the measure $\nu$.
The following theorem states that the posterior measure is Lipschitz with respect to data $y$, in the Hellinger metric.
\begin{thm}[Well-posedness of Posterior]
\label{thm:wellpose}
Let the potential function $\Phi$ in \eqref{eq:potential} satisfy Assumptions \ref{asmp:welldefpos} and \ref{asmp:Lipschitz_u}. 
If $\Pi$ is an $\STBP(\mC, B^{s,q}(\mZ))$ measure with $s>s'+\frac{d}{q}$, 
then
\begin{equation*}
d_H(\Pi(\cdot|y), \Pi(\cdot|y')) \leq C\Vert y-y'\Vert_\mbY ,
\end{equation*}
where $C=C(r)$ is a constant depending on $r$ such that $\max\{\Vert y\Vert_\mbY, \Vert y'\Vert_\mbY\}\leq r$. 
\end{thm}
\begin{proof}
The proof is based on \citep[Theorem 3.3 of ][]{Dashti_2012} and \citep[Theorem 16 of][]{dashti2017}.
As in Theorem \ref{thm:bayespost}, $Z(y), Z(y')\in(0,\infty)$. 
We directly compute
\begin{equation*}
\begin{aligned}
    2d_H^2(\Pi(\cdot|y), \Pi(\cdot|y')) =& \int_{B^{s',q}(\mZ)} \left[Z(y)^{-\half}\exp\left(-\half\Phi(u;y)\right) - Z(y')^{-\half}\exp\left(-\half\Phi(u;y')\right) \right]^2 \Pi(du) \\
    \leq & \frac{2}{Z(y)} \int_{B^{s',q}(\mZ)} \left[\exp\left(-\half\Phi(u;y)\right) - \exp\left(-\half\Phi(u;y')\right) \right]^2 \Pi(du) \\
    & + 2|Z(y)^{-\half} - Z(y')^{-\half}|^2 Z(y') .
\end{aligned}
\end{equation*}

By the mean value theorem and Assumptions \ref{asmp:welldefpos} (i) and (iii), we have
\begin{equation*}
\begin{aligned}
&\int_{B^{s',q}(\mZ)} \left[\exp\left(-\half\Phi(u;y)\right) - \exp\left(-\half\Phi(u;y')\right) \right]^2 \Pi(du) \\
&\leq \int_{B^{s',q}(\mZ)} \frac{1}{4} \exp(\alpha_1\Vert u\Vert_{s',q}-M) \exp(2\alpha_2\Vert u\Vert_{s',q}+2C)\Vert y_1-y_2\Vert_\mbY^2 \Pi(du) \leq C \Vert y_1-y_2\Vert_\mbY^2 ,
\end{aligned}
\end{equation*}
where the boundedness is the result of Theorem \ref{thm:Fernique}-(ii).
By the mean value theorem, Jesen's inequality and Assumptions \ref{asmp:welldefpos} (i) and (iii) we have
\begin{equation*}
\begin{aligned}
& 2|Z(y)^{-\half} - Z(y')^{-\half}|^2 Z(y')\leq C|Z(y)-Z(y')|^2 \\
&\leq C \left[\int_{B^{s',q}(\mZ)} |\exp(-\Phi(u;y)) - \exp(-\Phi(u;y'))| \Pi(du)\right]^2\\
&\leq C \int_{B^{s',q}(\mZ)} |\exp(-\Phi(u;y)) - \exp(-\Phi(u;y'))|^2 \Pi(du) \leq C \Vert y_1-y_2\Vert_\mbY^2 ,
\end{aligned}
\end{equation*}
where we used the above result.
\end{proof}

The posterior concentration theory developed by \cite{Agapiou_2021} for the $p$-exponential process applies to $\qEP(0, \mC)$ process.
For $\xi(\cdot)\in B^{s',q}(\mT)$ with $s'<s-\frac{d}{q}$, we define the concentration function of the $q$-exponential measure $\mu$ (hence $\mB(1, B^{s,q}(\mT))$) at $\xi=\xi^\dagger$ as follows:
\begin{equation*}\label{eq:contr_fun_t}
    \varphi_{\xi^\dagger}(\eps) = \inf_{h\in B^{s,q}(\mT): \Vert h-\xi^\dagger\Vert_{s',q}\leq \eps} \half \Vert h\Vert_{s,q}^q - \log \mu(\Vert \xi\Vert_{s',q}\leq\eps) .
\end{equation*}
Quoted from \citep[Theorem 3.1 and Lemma 5.14 of][]{Agapiou_2021}, the following general contraction theorem for $p$-exponential process (also BP, and hence Q-EP) priors will be used in the proof of posterior contraction Theorem \ref{thm:postcontr} for STBP priors.
\begin{thm}\label{thm:p-exp_contr}
Let $\mu$ be a $\qEP(0, \mC)$ measure satisfying Assumption \ref{asmp:eigen_summable}-(i) in the separable Banach space $(B^{s',q}(\mT), \Vert\cdot\Vert_{s',q})$, where $q\in[1,2]$. Let $\xi\sim \mu$ and the true parameter $\xi^\dagger\in B^{s',q}(\mT)$. Assume $\eps_n>0$ such that $\varphi_{\xi^\dagger}(\eps_n)\leq n\eps_n^2$, where $n\eps_n^2\gtrsim 1$.
Then for any $C>1$, there exists a measurable set $B_n\subset B^{s',q}(\mT)$ and a constant $R>0$ depending on $C$ and $q$, such that
\begin{eqnarray}
    \log N(4\eps_n, B_n, \Vert\cdot\Vert_{s',q}) &\leq& R n\eps_n^2 ; \label{eq1:LeCam_dim}\\
    \mu(\xi\notin B_n) &\leq& \exp(-Cn\eps_n^2) ; \label{eq2:residual}\\
    \mu(\Vert \xi-\xi^\dagger\Vert_{s',q}<2\eps_n) &\geq& \exp(-n\eps_n^2) , \label{eq3:small_ball}
\end{eqnarray}
where $N(4\eps_n, B_n, \Vert\cdot\Vert_{s',q})$ is the minimal number of $\Vert\cdot\Vert_{s',q}$-balls of radius $4\eps_n$ to cover $B_n$.
\end{thm}

\postcontr*

Before we proceed with the proof, we need the following lemma to bound the Hellinger distance, Kullback-Leibler (K-L) divergence and K-L variation.
\begin{lem}\label{lem:KL_bounds}
Suppose the inverse model \eqref{eq:bip} has the potential function $\Phi$ \eqref{eq:potential} satisfy Assumption \ref{asmp:Lipschitz_u}. 
Then we have
\begin{itemize}
    \item $d_H(p_u, p_{u'})\lesssim \Vert u-u'\Vert_{s',q}$ ;
    \item $K(p_u, p_{u'})\lesssim \Vert u-u'\Vert_{s',q}$ ;
    \item $V(p_u, p_{u'})\lesssim \Vert u-u'\Vert^2_{s',q}$ .
\end{itemize}
\end{lem}
\begin{proof}
First, we consider K-L divergence:
\begin{equation*}
    K(p_u, p_{u'}) = \int \log \frac{p_u}{p_{u'}} p_u d\mu =\int (\Phi(u';y)-\Phi(u;y)) p_u d\mu(y) \leq L\Vert u-u'\Vert_{s',q}
\end{equation*}
by Assumption \ref{asmp:Lipschitz_u}. 
Similarly, we have for K-L variation:
\begin{equation*}
    V(p_u, p_{u'}) = \int \left(\log \frac{p_u}{p_{u'}}\right)^2 p_u d\mu =\int |\Phi(u';y)-\Phi(u;y)|^2 p_u d\mu(y) \leq L^2\Vert u-u'\Vert_{s',q}^2 .
\end{equation*}
Lastly, we bound the Hellinger distance:
\begin{equation*}
\begin{aligned}
     2d^2_H(p_u, p_{u'}) &= \int(\sqrt{p_u}-\sqrt{p_{u'}})^2 d\mu = \int \left[1 - \exp\left(\half\Phi(u;y)-\half\Phi(u';y)\right)\right]^2 p_u d\mu(y) \\
     & \leq \int \frac{C}{4} |\Phi(u';y)-\Phi(u;y)|^2 p_u d\mu(y) \leq \frac{CL^2}{4}\Vert u-u'\Vert_{s',q}^2 ,
\end{aligned}
\end{equation*}
where the inequality holds for $\Vert u-u'\Vert_{s',q}^2$ small enough.
\end{proof}

\begin{proof}[Proof of Theorem \ref{thm:postcontr}]
\label{apx:postcontr}
Based on \citep[Theorem 1 of][]{ghosal2007}, it suffices to verify the following two conditions (the entropy condition (2.4), and the prior mass condition (2.5)) for some universal constants $\eta, K>0$ and sufficiently large $k\in \mbN$,
\begin{eqnarray}
\sup_{\eps>\eps_n}\log N(\eta \eps/2, \{u\in\Theta_n:d_{n,H}(u, u^\dagger)<\eps\}, d_{n,H}) &\leq & n\eps_n^2 ; \label{eq:entropy_cond}\\
\frac{\Pi_n(u\in\Theta_n: k\eps_n<d_{n,H}(u,u^\dagger)\leq 2k\eps_n)}{\Pi_n(B_n(u^\dagger, \eps_n))} &\leq & e^{Kn\eps_n^2k^2/2} , \label{eq:prior_mass}
\end{eqnarray}
where the left side of \eqref{eq:entropy_cond} is 
logarithm of the minimal number of $d_{n,H}$-balls of radius $\xi \eps/2$ needed to cover a ball of radius $\eps$ around the true value $u^\dagger=\sum_{\ell=1}^\infty \gamma_\ell \xi_\ell^\dagger(t)\phi_\ell(\bx)$; $B_n(u^\dagger,\eps_n)=\{u\in\Theta: \frac{1}{n}\sum_{j=1}^n K_j(u^\dagger,u)\leq \eps_n^2, \frac{1}{n}\sum_{j=1}^n V_j(u^\dagger,u)\leq \eps_n^2\}$ with $K_j(u^\dagger,u)=K(P_{u^\dagger,j},P_{u,j})$ and $V_j(u^\dagger,u)=V(P_{u^\dagger,j},P_{u,j})$.

We adopt the argument for infinite sequence of functions in \cite[Theorem 4 of][]{Lan_2022}.
For each $1\leq\ell\leq n$, $\xi_\ell(\cdot) \in B^{s',q}(\mT) \subset L^q(\mT)$ satisfy conditions for Theorem \ref{thm:p-exp_contr}. Therefore, there exists $B_{n, \ell}\subset B^{s',q}(\mT)$ such that \eqref{eq1:LeCam_dim}-\eqref{eq3:small_ball} hold for each $\ell$ with $B_n$ replaced by $B_{n, \ell}$, and $\eps_n$ replaced by $\eps_{n,\ell}=c 2^{-\ell}\eps_n$ for some constant $c>0$.
Note, for given spatial basis $\{\phi_\ell(\bx)\}_{\ell=1}^\infty$ and $\gamma=\{\gamma_\ell\}_{\ell=1}^\infty$ in \eqref{eq:STBP_randf}, $u\in \Theta=B^{s',q}(\mZ)$ can be identified with $\xi_\mT=\{\xi_\ell(\cdot)\}_{\ell=1}^\infty\in \ell^{q,\tau_q(s')}(L^q(\mT))$ through $f_\gamma$ in Definition \ref{dfn:STBP}, 
i.e. $\Theta\cong \ell^{q,\tau_q(s')}(L^q(\mT))$.
Now we set 
\begin{equation*}
    \Theta_n=\{u = f_\gamma(\xi_\mT) \in\Theta | \xi_\ell(\cdot)\in B_{n,\ell}, \; for \; \ell=1,\cdots,n\} \subset \Theta .
\end{equation*}
For $\forall u, u'\in \Theta_n$ such that $\Vert \xi_\ell(\cdot)-\xi'_\ell(\cdot)\Vert_{s',q}\leq \eps_{n,\ell}$ for $\ell=1,\cdots,n$, we have by Assumption-\ref{asmp:decaying_constant}
\begin{equation*}
\Vert u-u'\Vert_{s',q}^q = \sum_{\ell=1}^\infty \ell^{(\tau_q(s')-\tau_q(s))q} \Vert \xi_\ell(\cdot)-\xi'_\ell(\cdot)\Vert_q^q \lesssim \sum_{\ell=1}^n \Vert \xi_\ell(\cdot)-\xi'_\ell(\cdot)\Vert_{s',q}^q \lesssim \eps_n^q \sum_{\ell=1}^n 2^{-q\ell} .
\end{equation*}
Therefore,
$d_{n,H}(u,u')\lesssim \Vert u-u'\Vert_{s',q}\lesssim \eps_n$ by Lemma \ref{lem:KL_bounds}.
With $N(\eps_n, \Theta_n, d_{n,H})=\max_{1\leq\ell\leq n} N(4\eps_{n,\ell}, B_{n,\ell}, \Vert\cdot\Vert_{s',q})$, we have the following global entropy bound by \eqref{eq1:LeCam_dim}
\begin{equation*}
    \log N(\eps_n, \Theta_n, d_{n,H}) \leq Rn(c 2^{-\ell}\eps_n)^2 \lesssim n \eps_n^2 .
\end{equation*}
for some $1\leq\ell\leq n$ and $c>0$, which is stronger than the local entropy condition \eqref{eq:entropy_cond}.

Now by Lemma \ref{lem:KL_bounds} and \eqref{eq3:small_ball}, we have
\begin{equation*}
\begin{aligned}
    \Pi_n(B_n(u^\dagger, \eps_n)) & \geq \Pi_n(\Vert u^\dagger-u\Vert_{s',q}\leq \eps_n^2, \Vert u^\dagger-u\Vert_{s',q}^2\leq \eps_n^2) \\
    &= \Pi_n(\Vert u^\dagger-u\Vert_{s',q}^q\leq \eps_n^{2q})
    \gtrsim \exp\left\{\sum_{\ell=1}^n\log\mu(\Vert \xi_\ell(\cdot)-\xi'_\ell(\cdot)\Vert_{s',q}< 2\eps_{n,\ell}^2)\right\} \\
    &\geq e^{-n\sum_{\ell=1}^n\eps_{n,\ell}^4} \geq e^{-Knk^2\eps_n^2/2}, \quad with\; K=2, \; k^2 = c^4\sum_{\ell=1}^n2^{-4\ell} .
\end{aligned}
\end{equation*}
Then the prior mass condition \eqref{eq:prior_mass} is satisfied because the numerator is bounded by 1. 

Now we prove the complementary assertion $P_{u^\dagger}^{(n)}\Pi_n(\Theta\backslash \Theta_n|Y^{(n)})\to 0$ by \cite[Lemma 1 of][]{ghosal2007}.
It suffices to show $\frac{\Pi_n(\Theta\backslash\Theta_n)}{\Pi_n(B_n(u^\dagger, \eps_n))} = o(e^{-2n\eps_n^2})$.
By Theorem \ref{thm:p-exp_contr}, for $C_\ell = Kk^2\ell 2^{2\ell}c^{-2}$, we have $B_{n,\ell}\subset B^{s',q}(\mT)$ such that \eqref{eq2:residual} holds. Then
\begin{equation*}
\begin{aligned}
\Pi_n(\Theta\backslash\Theta_n) &= \Pi_n(u\in\Theta | \exists \ell\in\{1,\cdots,n\} \,such\, that\,  \xi_\ell(\cdot)\notin B_{n,\ell}) \\
&\leq \sum_{\ell=1}^n \mu(\xi_\ell\notin B_{n,\ell}) \leq \sum_{\ell=1}^n \exp(-C_\ell n\eps_{n,\ell}^2) \leq \frac{e^{-Kk^2n\eps_n^2}}{1-e^{-Kk^2n\eps_n^2}} .
\end{aligned}
\end{equation*}
By the above argument, we have $\Pi_n(B_n(u^\dagger, \eps_n)) \geq e^{-Knk^2\eps_n^2/2}$. Therefore
\begin{equation*}
\frac{\Pi_n(\Theta\backslash\Theta_n)}{\Pi_n(B_n(u^\dagger, \eps_n))} \leq \frac{e^{-Kk^2n\eps_n^2/2}}{1-e^{-Kk^2n\eps_n^2}} = o(e^{-2n\eps_n^2}) 
\end{equation*}
for chosen $c>0$ such that $k^2>2$.
The proof is hence completed.
\end{proof}

\postcontrate*

The following lemma studies the small ball probability in the concentration function \eqref{eq:contr_fun}.
\begin{lem}[Small ball probability]\label{lem:small_ball}
Let $\Pi$ be an $\STBP(\mC, B^{s,q}(\mZ))$ prior on $B^{s',q}(\mZ)$ with $s'<s-\frac{d}{q}$. Then as $\eps\to 0$, we have
\begin{equation*}
    -\log \Pi(\Vert u\Vert_{s',q}\leq \eps) \asymp \eps^{-\frac{1}{\frac{s-s'}{d} - \frac{1}{q}}} .
\end{equation*}
\end{lem}
\begin{proof}
We can compute
\begin{equation*}
    \Pi(\Vert u\Vert_{s',q}\leq \eps) = \mbP\left[\sum_{\ell=1}^\infty (\ell^{\tau_q(s')-\tau_q(s)}\Vert\xi_\ell(\cdot)\Vert_q)^q \leq \eps^q\right] ,
\end{equation*}
where $\mbP$ is the probability measure on the infinite product space $\Omega=(L^q(\mT))^\infty$ as in Definition \ref{dfn:STBP}.
From the proof of Theorem \ref{thm:qepint} we know $\Vert \xi\Vert_q^q = \sum_{\ell=1}^\infty \lambda_\ell^{\frac{q}{2}} \chi^2_\ell$ is an infinite mixture of $\chi^2(1)$ random variables, so the condition of \citep[Theorem 4.2 of][]{Aurzada2007} is trivially met and we have
\begin{equation*}
    \log \mbP\left[\sum_{\ell=1}^\infty (\ell^{\tau_q(s')-\tau_q(s)}\Vert\xi_\ell(\cdot)\Vert_q)^q \leq \eps^q\right] \asymp \eps^{-\frac{1}{\tau_q(s)-\tau_q(s') - \frac{1}{q}}} .
\end{equation*}
\end{proof}

The second lemma gives a upper bound of the first term of the concentration function \eqref{eq:contr_fun}.
\begin{lem}[Decentering function]\label{lem:decenter}
Assume $u^\dagger\in B^{s^\dagger,q^\dagger}(\mZ)$ for some $s^\dagger>s'$ and $q^\dagger\in[1,2]$. Then as $\eps\to 0$, we have the following bounds
\begin{enumerate}[label=(\roman*)]
\item If $q^\dagger \geq q$, we require $s^\dagger>s'$:
\begin{equation*}
\inf_{h\in B^{s,q}(\mZ): \Vert h-u^\dagger \Vert_{s',q}\leq \eps} \Vert h\Vert_{s,q}^q \lesssim 
\begin{cases}
    1, & \! if \, s < s^\dagger \\
    (-\log \eps)^{1-\frac{q}{q^\dagger}}, & \! if \, s=s^\dagger \\
    \eps^{-\frac{s-s^\dagger}{s^\dagger-s'}(q\wedge q^\dagger)}, & \! if \, s>s^\dagger
\end{cases} ;
\end{equation*}
\item If $q^\dagger < q$, we restrict $u^\dagger\in B^{s^\dagger,q^\dagger, q}(\mZ)\subsetneq B^{s^\dagger,q^\dagger}(\mZ)$ and require $s^\dagger>s' -\frac{d}{q}+\frac{d}{q^\dagger}$:
\begin{equation*}
\inf_{h\in B^{s,q}(\mZ): \Vert h-u^\dagger \Vert_{s',q}\leq \eps} \Vert h\Vert_{s,q}^q \lesssim
\begin{cases}
    1, & \! if \, s \leq s^\dagger+\frac{d}{q} -\frac{d}{q^\dagger} \\
    \eps^{-\frac{\frac{s-s^\dagger}{d}-\frac{1}{q}+\frac{1}{q^\dagger}}{\frac{s^\dagger-s'}{d}+\frac{1}{q}-\frac{1}{q^\dagger}}q}, & \! if \, s>s^\dagger +\frac{d}{q} -\frac{d}{q^\dagger}
\end{cases} .
\end{equation*}
\end{enumerate}
\end{lem}
\begin{proof}
First of all, by Proposition \ref{prop:embeddings} we have $B^{s^\dagger,q^\dagger\wedge q, q}(\mZ)\subset B^{s',q}$ for $s'<s^\dagger-\left(\frac{d}{q^\dagger}-\frac{d}{q}\right)_+$.

Next, for given spatial basis $\{\phi_\ell\}_{\ell=1}^\infty$ in Definition \ref{dfn:STBP}, we identify $u^\dagger\in B^{s^\dagger,q^\dagger}$ with $u^\dagger_\mT=\{u^\dagger_\ell(\cdot)\}_{\ell=1}^\infty\in \ell^{q^\dagger,\tau_{q^\dagger}(s^\dagger)}(L^{q^\dagger}(\mT))$. 
Then we follow \cite{Agapiou_2021} to approximate $u^\dagger_\mT$ with $h_{1:L}=\{u^\dagger_\ell(\cdot)\}_{\ell=1}^\infty$ where $u^\dagger_\ell(\cdot) \equiv 0$ for all $\ell>L$.
Note $h_{1:L}\in \ell^{q,\tau_q(s)}(L^q(\mT))$ for any finite $L\in\mbN$.
Identifying $h_{1:L}$ with $h\in B^{s,q}(\mZ)$, we use the similar argument as above to get
\begin{equation*}
\Vert h-u^\dagger \Vert_{s',q}^q = \sum_{\ell=L+1}^\infty \ell^{\tau_q(s')q}\Vert u^\dagger_\ell(\cdot)\Vert_q^q \leq 
\begin{cases}
    \Vert u^\dagger\Vert_{s^\dagger,q^\dagger}^{q^\dagger} L^{\frac{s'-s^\dagger}{d}q^\dagger}, & \! if \, q^\dagger=q \\
    \Vert u^\dagger\Vert_{s^\dagger,q^\dagger}^q L^{\frac{s'-s^\dagger}{d}q}, & \! if \, q^\dagger>q \\
    \Vert u^\dagger\Vert_{s^\dagger, q^\dagger, q}^q L^{\left(\frac{s'-s^\dagger}{d}-\frac{1}{q}+\frac{1}{q^\dagger}\right)q}, & \! if \, q^\dagger<q
\end{cases} .
\end{equation*}
Therefore, to have $\Vert h-u^\dagger \Vert_{s',q}\leq \eps$ we let 
\begin{equation}\label{eq:L_lbd}
L \gtrsim
\begin{cases}
    \eps^{-\frac{d}{s^\dagger-s'}}, & \! if \, q^\dagger \geq q\\
    \eps^{-\frac{1}{\frac{s^\dagger-s'}{d}+\frac{1}{q}-\frac{1}{q^\dagger}}} , & \! if \, q^\dagger <q
\end{cases} .
\end{equation}

On the other hand, the infimum is less than $\Vert h \Vert_{s,q}^q$ with above $h$, which can be bounded as follows.
If $q^\dagger=q$,
\begin{equation*}
\Vert h \Vert_{s,q}^q = \sum_{\ell=1}^L \ell^{\tau_{q^\dagger}(s)q^\dagger}\Vert u^\dagger_\ell(\cdot)\Vert_{q^\dagger}^{q^\dagger} \leq
\begin{cases}
    \Vert u^\dagger \Vert_{s^\dagger,q^\dagger}^{q^\dagger}, & \! if \, s\leq s^\dagger\\
    \Vert u^\dagger \Vert_{s^\dagger,q^\dagger}^{q^\dagger} L^{\frac{s-s^\dagger}{d}q^\dagger}, & \! if \, s>s^\dagger
\end{cases} .
\end{equation*}
If $q^\dagger>q$, by similar argument using H\"older inequality as in \eqref{eq:qs_bigg_q},
\begin{equation*}
\Vert h \Vert_{s,q}^q = \sum_{\ell=1}^L \ell^{\tau_q(s)q}\Vert u^\dagger_\ell(\cdot)\Vert_q^q \leq
\begin{cases}
    C \Vert u^\dagger \Vert_{s^\dagger,q^\dagger}^q, & \! if \, s< s^\dagger\\
    \Vert u^\dagger \Vert_{s^\dagger,q^\dagger}^q (\log L)^{1-\frac{q}{q^\dagger}}, & \! if \, s=s^\dagger \\
    \Vert u^\dagger\Vert_{s^\dagger,q^\dagger}^q L^{\frac{s-s^\dagger}{d}q}, & \! if \, s>s^\dagger
\end{cases} .
\end{equation*}
If $q^\dagger<q$, by similar argument as in \eqref{eq:qs_less_q},
\begin{equation*}
\Vert h \Vert_{s,q}^q = \sum_{\ell=1}^L \ell^{\tau_q(s)q}\Vert u^\dagger_\ell(\cdot)\Vert_q^q \leq
\begin{cases}
    \Vert u^\dagger \Vert_{s^\dagger,q^\dagger,q}^q, & \! if \, s \leq s^\dagger+\frac{d}{q} -\frac{d}{q^\dagger}\\
    \Vert u^\dagger\Vert_{s^\dagger,q^\dagger,q}^q L^{\left(\frac{s-s^\dagger}{d}-\frac{1}{q}+\frac{1}{q^\dagger}\right)q}, & \! if \, s > s^\dagger+\frac{d}{q} -\frac{d}{q^\dagger}
\end{cases} .
\end{equation*}
Substituting $L$ in \eqref{eq:L_lbd} to the above equations yields the conclusion.
\end{proof}

\begin{proof}[Proof of Theorem \ref{thm:postcontrate}]
\label{apx:postcontrate}
By Lemmas \ref{lem:small_ball} and \ref{lem:decenter}, we have the following bounds for the concentration function \eqref{eq:contr_fun} as $\eps\to 0$, if $q^\dagger\geq q$,
\begin{equation*}
    \varphi_{u^\dagger}(\eps) \lesssim 
\begin{cases}
    1 + \eps^{-\frac{1}{\frac{s-s'}{d} - \frac{1}{q}}}, & \! if \, s < s^\dagger \\
    (-\log \eps)^{1-\frac{q}{q^\dagger}} + \eps^{-\frac{1}{\frac{s-s'}{d} - \frac{1}{q}}}, & \! if \, s=s^\dagger \\
    \eps^{-\frac{s-s^\dagger}{s^\dagger-s'}(q\wedge q^\dagger)} + \eps^{-\frac{1}{\frac{s-s'}{d} - \frac{1}{q}}}, & \! if \, s>s^\dagger
\end{cases} .
\end{equation*}
For $s\leq s^\dagger$, the bound is dominated by $\eps^{-\frac{1}{\frac{s-s'}{d} - \frac{1}{q}}}$.
For the last case, we need to determine a balancing point of $s$ for the two terms by setting their powers equal. The calculation shows that if $s\leq s^\dagger +\frac{d}{q}$, the bound is still dominated by $\eps^{-\frac{1}{\frac{s-s'}{d} - \frac{1}{q}}}$, but otherwise is dominated by $\eps^{-\frac{s-s^\dagger}{s^\dagger-s'}q}$.
Therefore, we have
\begin{equation*}
    \varphi_{u^\dagger}(\eps) \lesssim 
\begin{cases}
    \eps^{-\frac{1}{\frac{s-s'}{d} - \frac{1}{q}}}, & \! if \, s \leq s^\dagger +\frac{d}{q} \\
    \eps^{-\frac{s-s^\dagger}{s^\dagger-s'}q}, & \! if \, s>s^\dagger +\frac{d}{q}
\end{cases} .
\end{equation*}
We need to determine minimal $\eps_n$ such that $\varphi_{u^\dagger}(\eps_n)\leq n\eps_n^2$. Hence for $q^\dagger \geq q$,
\begin{equation*}
    \eps_n \asymp 
\begin{cases}
    n^{-\frac{q(s-s')-d}{2q(s-s')+(q-2)d}}, & \! if \, s \leq s^\dagger +\frac{d}{q} \\
    n^{-\frac{s^\dagger-s'}{2(s^\dagger-s')+q(s-s^\dagger)}}, & \! if \, s>s^\dagger +\frac{d}{q}
\end{cases} .
\end{equation*}

Now if $q^\dagger<q$, by similar argument we have the concentration function \eqref{eq:contr_fun} as $\eps\to 0$
\begin{equation*}
    \varphi_{u^\dagger}(\eps) \lesssim 
\begin{cases}
    1 + \eps^{-\frac{1}{\frac{s-s'}{d} - \frac{1}{q}}}, & \! if \, s \leq s^\dagger+\frac{d}{q} -\frac{d}{q^\dagger} \\
    \eps^{-\frac{\frac{s-s^\dagger}{d}-\frac{1}{q}+\frac{1}{q^\dagger}}{\frac{s^\dagger-s'}{d}+\frac{1}{q}-\frac{1}{q^\dagger}}q} + \eps^{-\frac{1}{\frac{s-s'}{d} - \frac{1}{q}}}, & \! if \, s>s^\dagger +\frac{d}{q} -\frac{d}{q^\dagger}
\end{cases} .
\end{equation*}
Thus the contraction rate for $q^\dagger<q$ becomes
\begin{equation*}
    \eps_n \asymp 
\begin{cases}
    n^{-\frac{q(s-s')-d}{2q(s-s')+(q-2)d}}, & \! if \, s \leq s^\dagger +\frac{2d}{q} - \frac{d}{q^\dagger} \\
    n^{-\frac{s^\dagger-s'+\frac{d}{q}-\frac{d}{q^\dagger}}{2(s^\dagger-s')+q(s-s^\dagger) -(q-2)(\frac{d}{q}-\frac{d}{q^\dagger})}}, & \! if \, s>s^\dagger +\frac{2d}{q} - \frac{d}{q^\dagger}
\end{cases} .
\end{equation*}
Rewriting the equations into one yields the conclusion.
\end{proof}

\adapostcontrate*
\begin{proof}[Proof of Theorem \ref{thm:adapostcontrate}]
\label{apx:adapostcontrate}
Denote the upper bound of the first term in the concentration in Lemma \ref{lem:decenter} as $d(\eps)$. By re-examining the proof of Theorem \ref{thm:postcontrate}, we have the concentration function \eqref{eq:contr_fun_rescaled} bounded as
\begin{equation*}
\varphi_{u^\dagger,\kappa}(\eps) \lesssim \kappa^{-q} d(\eps) + (\eps/\kappa)^{-\frac{1}{\frac{s-s'}{d} - \frac{1}{q}}} .
\end{equation*}
The optimal choice, $\kappa \asymp d(\eps)^{\frac{1}{q}-\frac{d}{q^2(s-s')}} \eps^{\frac{d}{q(s-s')}}$, is made by balancing the above two terms. Hence the concentration function bound becomes $\varphi_{u^\dagger,\kappa}(\eps) \lesssim d(\eps)^{\frac{d}{q(s-s')}} \eps^{-\frac{d}{s-s'}}$.
Note, most bounds in Lemma \ref{lem:decenter} appears in the format of $d(\eps)\asymp \eps^{-b}$ except when $q^\dagger\geq q$ and $s=s^\dagger$. We substitute in and force the derived rate to be minimax:
\begin{equation*}
n^{-\frac{1}{2+\frac{db}{q(s-s')}+\frac{d}{s-s'}}} \asymp n^{-\frac{1}{2+\frac{d}{s^\dagger}}} ,
\end{equation*}
which implies that $b(s)=q(\frac{s-s'}{s^\dagger}-1)$ and the corresponding scaling factor $\kappa\asymp \eps^{-\frac{s-s'-\frac{d}{q}}{s^\dagger}+1}$.

Next we examine whether the bound $\eps^{-b(s)}=\eps^{-q(\frac{s-s'}{s^\dagger}-1)}$ can be achieved as those $d(\eps)$ in Lemma \ref{lem:decenter} .
If $q^\dagger\geq q$, setting $b(s)=0$ leads to $s=s'+s^\dagger$ contradicting with $s<s^\dagger$; $b(s)=\frac{s-s^\dagger}{s^\dagger-s'}(q\wedge q^\dagger)$ yields $s=s^\dagger+s'-(s^\dagger)^2/s'$ contradicting with $s>s^\dagger$; Lastly, $s=s^\dagger$ does not solve $(-\log \eps)^{(1-\frac{q}{q^\dagger})\frac{d}{q(s-s')}} \eps^{-\frac{d}{s-s'}}=\eps^{-\frac{1}{2+\frac{d}{s^\dagger}}}$.
If $q^\dagger<q$, $s=s'+s^\dagger$ does not satisfy $s \leq s^\dagger-\left(\frac{d}{q^\dagger} -\frac{d}{q}\right)$; solving $b(s)=\frac{\frac{s-s^\dagger}{d}-\frac{1}{q}+\frac{1}{q^\dagger}}{\frac{s^\dagger-s'}{d}+\frac{1}{q}-\frac{1}{q^\dagger}}q$ gives $s=\frac{s^\dagger(s^\dagger-s')}{s'+\left(\frac{d}{q^\dagger} -\frac{d}{q}\right)}+s'$, which can be shown $s>s^\dagger -\left(\frac{d}{q^\dagger} -\frac{d}{q}\right)$.
Hence, substitute the only feasible $s$ and the minimax rate $\eps_n^\dagger$ into the above expression of $\kappa$ and the scaling factor becomes $\kappa_n\asymp (\eps_n^\dagger)^{-\frac{s^\dagger-s'}{s'+\left(\frac{d}{q^\dagger} -\frac{d}{q}\right)}+\frac{d}{qs^\dagger}+1} \asymp n^{-\frac{1}{2s^\dagger+d}\left[-\frac{s^\dagger(s^\dagger-s')}{s'+\left(\frac{d}{q^\dagger} -\frac{d}{q}\right)}+\frac{d}{q}+s^\dagger\right]}$.
\end{proof}

\section{Inference}\label{apx:inference}
\begin{algorithm}[!ht]
\caption{White-noise dimension-independent MCMC (wn-$\infty$-MCMC)}
\label{alg:wn-infMC}
\centering
\begin{algorithmic}[1]
\STATE Initialize current state $u^{(0)}$ and transform it into the whitened space $\zeta^{(0)}=T^{-1}(u^{(0)})$
\STATE Sample velocity $\eta^{(0)}\sim \mN(0,I)$
\STATE Calculate current energy $E_0=\Phi(\zeta^{(0)}) - \frac{\eps^2}{8} \Vert g(\zeta^{(0)})\Vert^2 + \half \log |\mK(\zeta^{(0)})|$
\FOR{ $i = 0$ to $I-1$}
\STATE Run $\Psi_\eps: (\zeta^{(i)},\eta^{(i)})\mapsto (\zeta^{(i+1)}, \eta^{(i+1)})$ according to \eqref{eq:mHDdiscret}.
\STATE Update the energy $E_0 \gets E_0 + \frac{\eps}{2} (\langle g(\zeta^{(i)}), \eta^{(i)} \rangle + \langle g(\zeta^{(i+1)}), \eta^{(i+1)} \rangle)$
\ENDFOR
\STATE Calculate new energy $E_1=\Phi(\zeta^{(I)}) - \frac{\eps^2}{8} \Vert g(\zeta^{(I)})\Vert^2 + \half \log |\mK(\zeta^{(I)})|$
\STATE Calculate acceptance probability $a=\exp(-E_1+E_0)$. 
\STATE Accept $\zeta^{(I)}$ with probability $a$ for the next state $\zeta'$ or set $\zeta'=\zeta^{(0)}$.
\STATE Record the next state $u'=T(\zeta')$ in the original space.
\end{algorithmic}
\end{algorithm}


The following proposition permits conditional conjugacy for the variance magnitude ($\kappa$) given an appropriate hyper-prior.
\begin{prop}
\label{prop:condconj}
If we assume a inverse-gamma hyper-prior for the variance magnitude $\kappa^{\frac{q}{2}}\sim \Gamma^{-1}(\alpha,\beta)$ such that $\bxi_\ell|\kappa \overset{iid}{\sim} \qED_J(\bzero,\bC)$ in \eqref{eq:STBP_model}, then we have
\begin{equation}\label{eq:cond_kappa}
\kappa^{\frac{q}{2}}|\bu \sim \Gamma^{-1}(\alpha',\beta'), \quad 
\alpha'=\alpha+\frac{JL}{2}, \quad \beta'=\beta+\half\sum_{\ell=1}^L r_{0,\ell}^\frac{q}{2} .
\end{equation}
\end{prop}
\begin{proof}
We can compute the joint density of $\bXi$ and $\kappa$
\begin{align*}
& p(\bXi, \kappa)= \prod_{\ell=1}^L p(\bxi_\ell|\kappa) p(\kappa^q) \\
&= \left(\frac{q}{2}\right)^L (2\pi)^{-\frac{JL}{2}} |\bC_0|^{-\frac{L}{2}} \prod_{\ell=1}^L r_{0,\ell}^{(\frac{q}{2}-1)\frac{J}{2}} \kappa^{-\frac{q}{2}\cdot\frac{JL}{2}}\exp\left\{-\kappa^{-\frac{q}{2}} \sum_{\ell=1}^L \frac{r_{0,\ell}^\frac{q}{2}}{2}\right\} 
\frac{\beta^\alpha}{\Gamma(\alpha)} \kappa^{-\frac{q}{2}(\alpha+1)} \exp(-\beta \kappa^{-\frac{q}{2}}) \\
&\propto \left(\kappa^{\frac{q}{2}}\right)^{-(\alpha+\frac{JL}{2}+1)} \exp\left\{-\kappa^{-\frac{q}{2}}\left(\beta+\half\sum_{\ell=1}^L r_{0,\ell}^\frac{q}{2}\right)\right\} .
\end{align*}
By identifying the parameters for $\kappa^{\frac{q}{2}}$ we recognize that $\kappa^{\frac{q}{2}}|\bXi$ is another inverse-gamma with parameters $\alpha'$ and $\beta'$ as given.
\end{proof}

\section{More Numerical Results}

\subsection{Simulation}

\begin{figure}[!ht]
\begin{tabular}{ccccc}
\qquad Truth & \qquad \qquad $q=0.5$ & \quad \qquad $q=1$ (STBP) & \qquad $q=1.5$ & \quad \quad $q=2$ (STGP)  \\
\end{tabular}
\vspace{-5pt}
\begin{tabular}{ccccc}
\includegraphics[width=0.19\textwidth,height=.18\textwidth]{simulation/truth/simulation_009.png}
\includegraphics[width=0.19\textwidth,height=.18\textwidth]{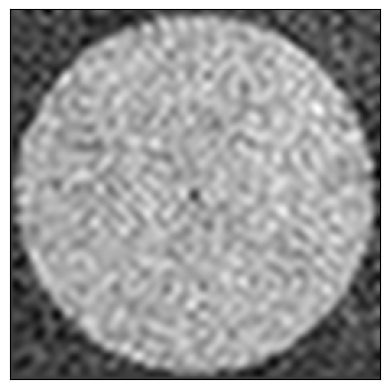}
\includegraphics[width=0.19\textwidth,height=.18\textwidth]{simulation/I256_J100/MAP_q1p1s1_whiten/simulation_009.png}
\includegraphics[width=0.19\textwidth,height=.18\textwidth]{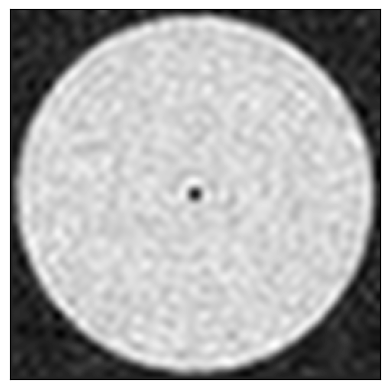}
\includegraphics[width=0.19\textwidth,height=.18\textwidth]{simulation/I256_J100/MAP_q2p2s1_whiten/simulation_009.png}
\end{tabular}
\vspace{-5pt}
\begin{tabular}{ccccc}
\includegraphics[width=0.19\textwidth,height=.18\textwidth]{simulation/truth/simulation_029.png}
\includegraphics[width=0.19\textwidth,height=.18\textwidth]{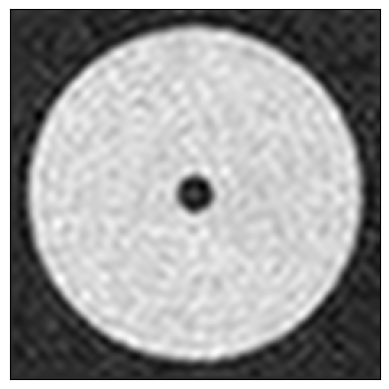}
\includegraphics[width=0.19\textwidth,height=.18\textwidth]{simulation/I256_J100/MAP_q1p1s1_whiten/simulation_029.png}
\includegraphics[width=0.19\textwidth,height=.18\textwidth]{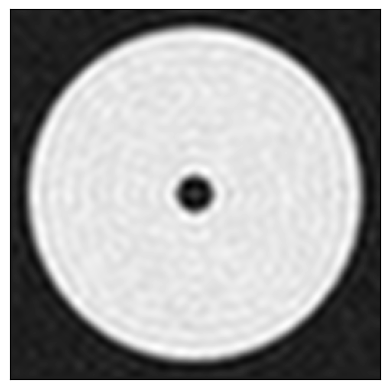}
\includegraphics[width=0.19\textwidth,height=.18\textwidth]{simulation/I256_J100/MAP_q2p2s1_whiten/simulation_029.png}
\end{tabular}
\vspace{-5pt}
\begin{tabular}{ccccc}
\includegraphics[width=0.19\textwidth,height=.18\textwidth]{simulation/truth/simulation_059.png}
\includegraphics[width=0.19\textwidth,height=.18\textwidth]{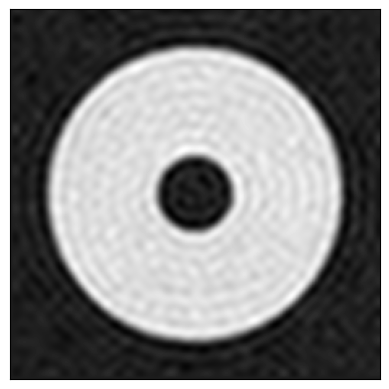}
\includegraphics[width=0.19\textwidth,height=.18\textwidth]{simulation/I256_J100/MAP_q1p1s1_whiten/simulation_059.png}
\includegraphics[width=0.19\textwidth,height=.18\textwidth]{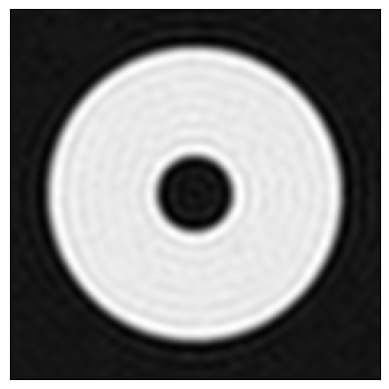}
\includegraphics[width=0.19\textwidth,height=.18\textwidth]{simulation/I256_J100/MAP_q2p2s1_whiten/simulation_059.png}
\end{tabular}
\vspace{-5pt}
\begin{tabular}{ccccc}
\includegraphics[width=0.19\textwidth,height=.18\textwidth]{simulation/truth/simulation_089.png}
\includegraphics[width=0.19\textwidth,height=.18\textwidth]{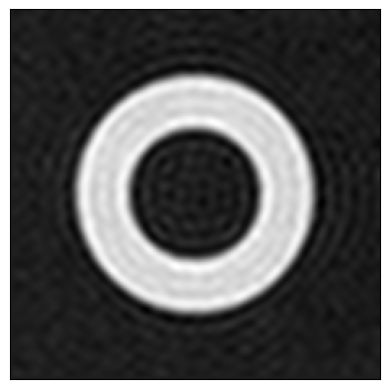}
\includegraphics[width=0.19\textwidth,height=.18\textwidth]{simulation/I256_J100/MAP_q1p1s1_whiten/simulation_089.png}
\includegraphics[width=0.19\textwidth,height=.18\textwidth]{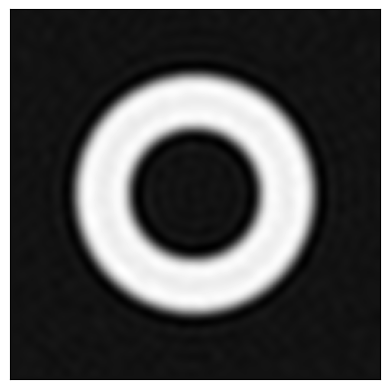}
\includegraphics[width=0.19\textwidth,height=.18\textwidth]{simulation/I256_J100/MAP_q2p2s1_whiten/simulation_089.png}
\end{tabular}
\vspace{-5pt}
\caption{MAP reconstruction of simulated annulus with $I=256\times 256$, $J=100$. Columns from left to right: true images, MAP estimates by STBP models with $q=0.5, 1, 1.5, 2$ respectively. Rows from top to bottom: time step $t_j = 0.1, 0.3, 0.6, 0.9$.}
\label{fig:simulation_MAP_whiten_q}
\end{figure}

Figure \ref{fig:simulation_MAP_whiten_q} illustrates the regularization effect of parameter $q>0$ of STBP priors in the simulated regression problem of a shrinking annulus. When the regularization parameter $q$ ranges in $(0,2]$, the smaller $q$ is, the more regularization it imposes hence the sharper MAP solution the corresponding model renders compared with the truth. When $q=2$, STBP reduces to STGP which returns the smoothest reconstruction with blurring boundaries. Even $q=0.5$ is not in the main range of interest $[1,2]$ where the associated priors have good properties, e.g. convexity, the resulting prior model still yields a solution (the second column) with the lowest error among the models for selective $q$'s.

\begin{figure}[t]
\begin{tabular}{ccccc}
\quad \quad Truth & \qquad \qquad $\overset{I=16\times 16}{J=10}$  & \qquad \qquad $\overset{I=32\times 32}{J=20}$ & \quad \qquad $\overset{I=128\times 128}{J=50}$ & \qquad \qquad $\overset{I=256\times 256}{J=100}$  \\
\end{tabular}
\vspace{-5pt}
\begin{tabular}{ccccc}
\includegraphics[width=0.19\textwidth,height=.18\textwidth]{simulation/truth/simulation_009.png}
\includegraphics[width=0.19\textwidth,height=.18\textwidth]{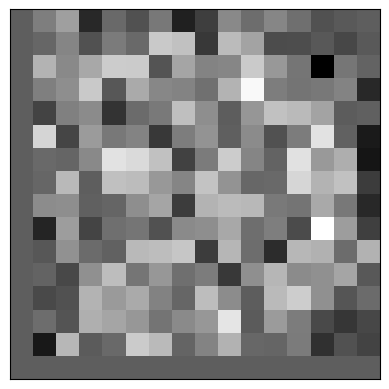}
\includegraphics[width=0.19\textwidth,height=.18\textwidth]{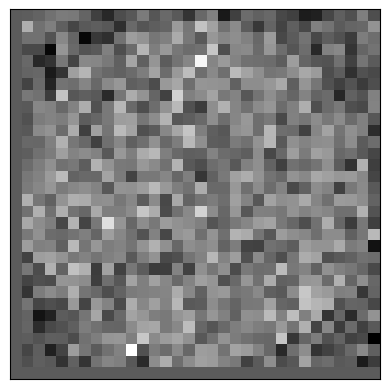}
\includegraphics[width=0.19\textwidth,height=.18\textwidth]{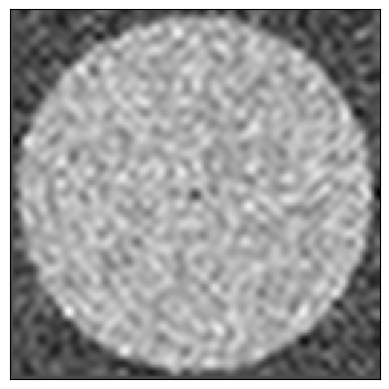}
\includegraphics[width=0.19\textwidth,height=.18\textwidth]{simulation/I256_J100/MAP_q1p1s1_whiten/simulation_009.png}
\end{tabular}
\vspace{-5pt}
\begin{tabular}{ccccc}
\includegraphics[width=0.19\textwidth,height=.18\textwidth]{simulation/truth/simulation_029.png}
\includegraphics[width=0.19\textwidth,height=.18\textwidth]{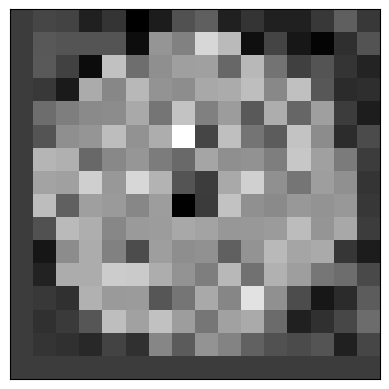}
\includegraphics[width=0.19\textwidth,height=.18\textwidth]{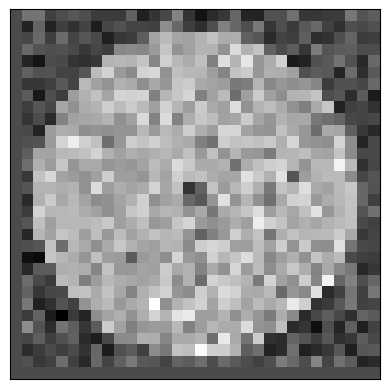}
\includegraphics[width=0.19\textwidth,height=.18\textwidth]{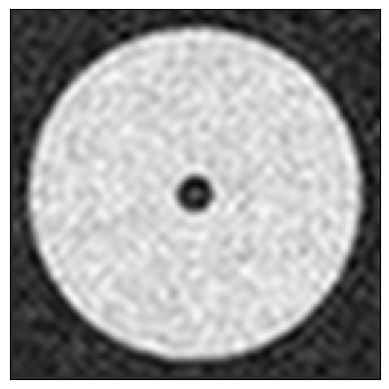}
\includegraphics[width=0.19\textwidth,height=.18\textwidth]{simulation/I256_J100/MAP_q1p1s1_whiten/simulation_029.png}
\end{tabular}
\vspace{-5pt}
\begin{tabular}{ccccc}
\includegraphics[width=0.19\textwidth,height=.18\textwidth]{simulation/truth/simulation_059.png}
\includegraphics[width=0.19\textwidth,height=.18\textwidth]{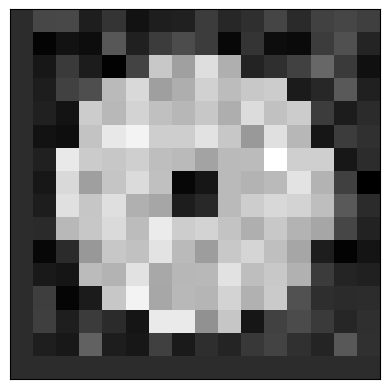}
\includegraphics[width=0.19\textwidth,height=.18\textwidth]{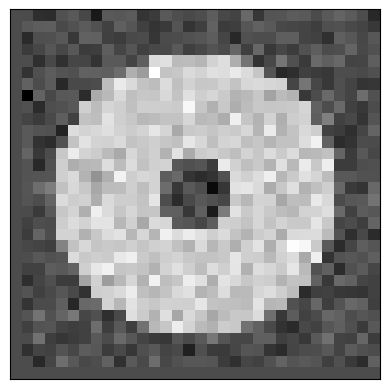}
\includegraphics[width=0.19\textwidth,height=.18\textwidth]{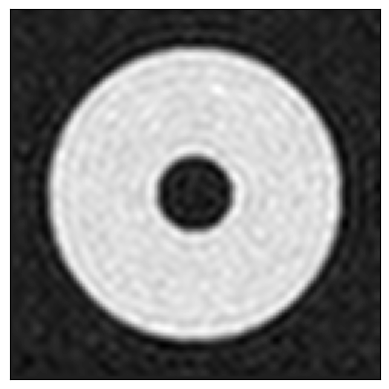}
\includegraphics[width=0.19\textwidth,height=.18\textwidth]{simulation/I256_J100/MAP_q1p1s1_whiten/simulation_059.png}
\end{tabular}
\vspace{-5pt}
\begin{tabular}{ccccc}
\includegraphics[width=0.19\textwidth,height=.18\textwidth]{simulation/truth/simulation_089.png}
\includegraphics[width=0.19\textwidth,height=.18\textwidth]{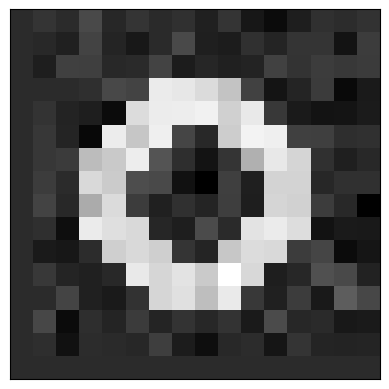}
\includegraphics[width=0.19\textwidth,height=.18\textwidth]{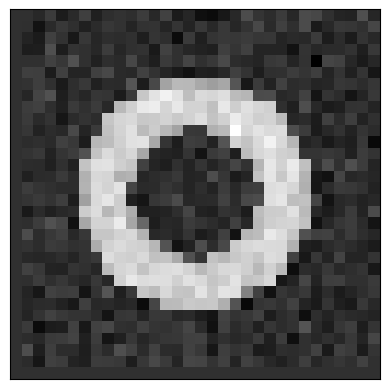}
\includegraphics[width=0.19\textwidth,height=.18\textwidth]{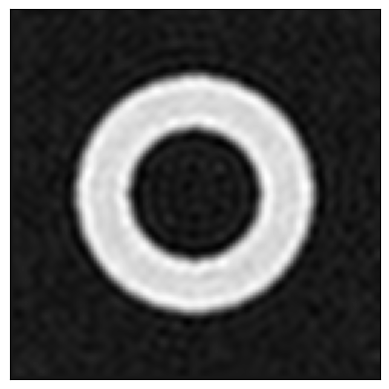}
\includegraphics[width=0.19\textwidth,height=.18\textwidth]{simulation/I256_J100/MAP_q1p1s1_whiten/simulation_089.png}
\end{tabular}
\vspace{-5pt}
\caption{MAP reconstruction of simulated annulus by STBP model with increasing data. Columns from left to right: true images, MAP estimates obtained at different spatiotemporal resolutions. Rows from top to bottom: time step $t_j = 0.1, 0.3, 0.6, 0.9$.}
\label{fig:simulation_postcontr_whiten}
\end{figure}

We increase the spatiotemporal resolution in Figure \ref{fig:simulation_postcontr_whiten} to illustrate the MAP estimates by the STBP model approximating the ground truth. This demonstrates the posterior contraction phenomenon in the infinitely informative data limit as described Theorem \ref{thm:postcontr}.

\subsection{Dynamic Tomography Reconstruction}

\begin{figure}[!ht]
\begin{tabular}{ccccc}
\quad Truth & \qquad \qquad Observations & \quad \qquad STBP & \qquad \qquad STGP & \quad \quad time-uncorrelated  \\
\end{tabular}
\vspace{-5pt}
\begin{tabular}{ccccc}
\includegraphics[width=0.19\textwidth,height=.18\textwidth]{STEMPO/truth/stempo_00.png}
\includegraphics[width=0.19\textwidth,height=.18\textwidth]{STEMPO/sino/stempo_00.png}
\includegraphics[width=0.19\textwidth,height=.18\textwidth]{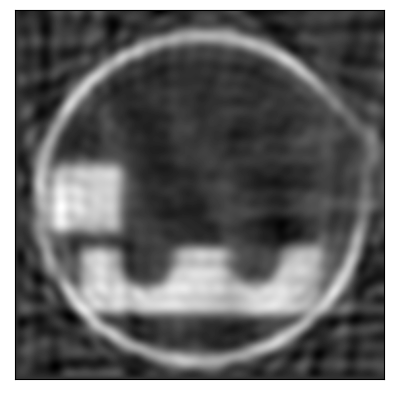}
\includegraphics[width=0.19\textwidth,height=.18\textwidth]{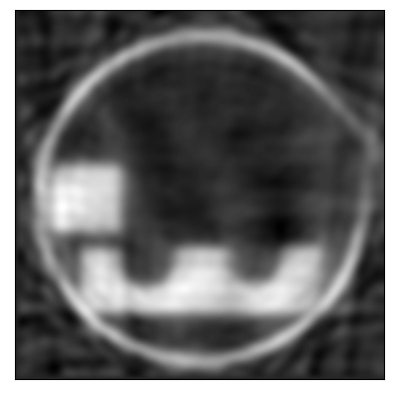}
\includegraphics[width=0.19\textwidth,height=.18\textwidth]{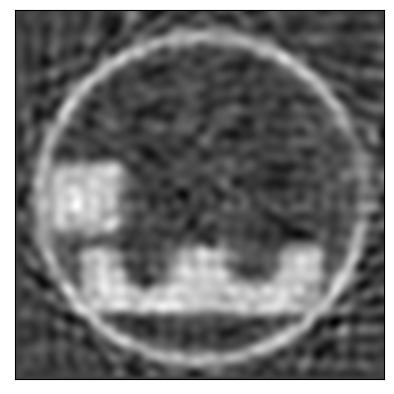}
\end{tabular}
\vspace{-5pt}
\begin{tabular}{ccccc}
\includegraphics[width=0.19\textwidth,height=.18\textwidth]{STEMPO/truth/stempo_06.png}
\includegraphics[width=0.19\textwidth,height=.18\textwidth]{STEMPO/sino/stempo_06.png} 
\includegraphics[width=0.19\textwidth,height=.18\textwidth]{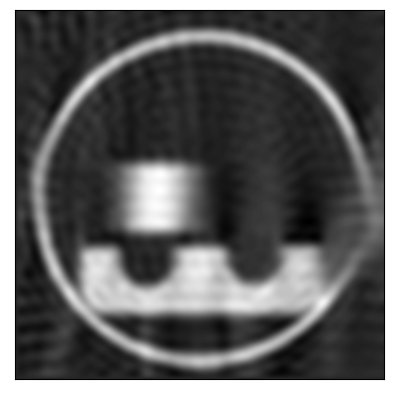} 
\includegraphics[width=0.19\textwidth,height=.18\textwidth]{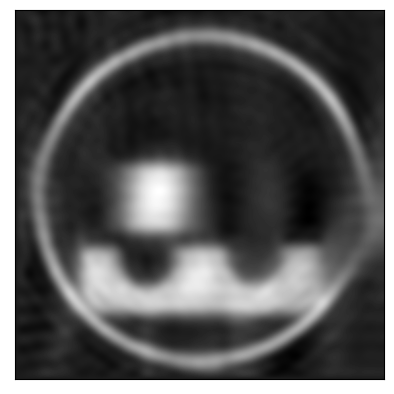} 
\includegraphics[width=0.19\textwidth,height=.18\textwidth]{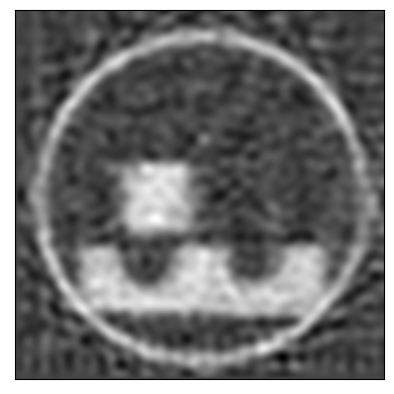} 
\end{tabular}
\vspace{-5pt}
\begin{tabular}{ccccc}
\includegraphics[width=0.19\textwidth,height=.18\textwidth]{STEMPO/truth/stempo_13.png}
\includegraphics[width=0.19\textwidth,height=.18\textwidth]{STEMPO/sino/stempo_13.png}
\includegraphics[width=0.19\textwidth,height=.18\textwidth]{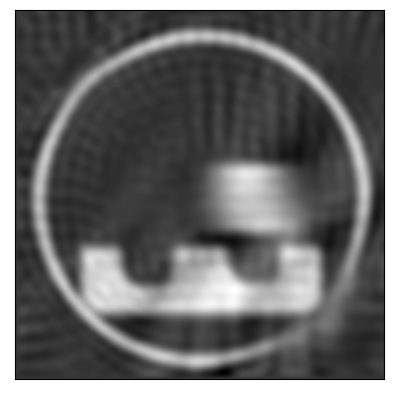}
\includegraphics[width=0.19\textwidth,height=.18\textwidth]{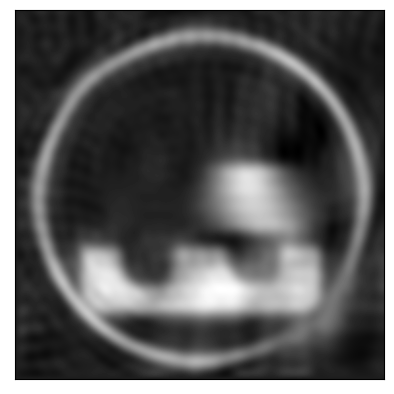}
\includegraphics[width=0.19\textwidth,height=.18\textwidth]{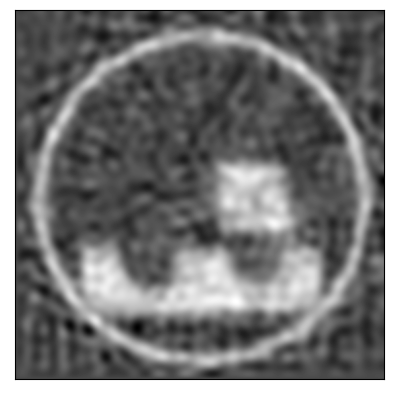}
\end{tabular}
\vspace{-5pt}
\begin{tabular}{ccccc}
\includegraphics[width=0.19\textwidth,height=.18\textwidth]{STEMPO/truth/stempo_19.png}
\includegraphics[width=0.19\textwidth,height=.18\textwidth]{STEMPO/sino/stempo_19.png}
\includegraphics[width=0.19\textwidth,height=.18\textwidth]{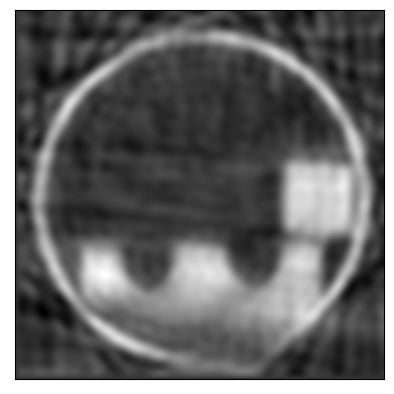}
\includegraphics[width=0.19\textwidth,height=.18\textwidth]{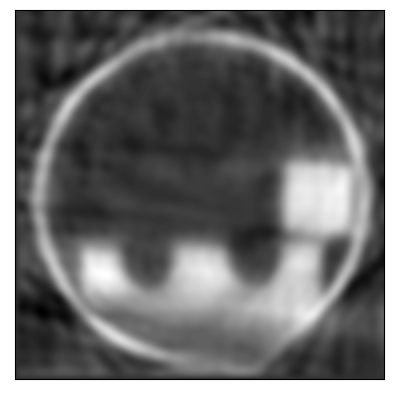}
\includegraphics[width=0.19\textwidth,height=.18\textwidth]{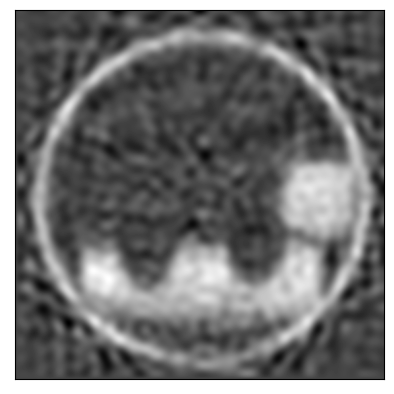}
\end{tabular}
\vspace{-5pt}
\caption{Reconstruction results of dynamic STEMPO tomography in the original space. Columns from left to right: true images, sinograms, MAP estimates by STBP, STGP and time-uncorrelated models respectively. Rows from top to bottom: time step $j = 0, 6, 13, 19$.}
\label{fig:STEMPO_MAP}
\end{figure}
\subsubsection{STEMPO Tomography}

In Figure \ref{fig:STEMPO_MAP}, MAP estimates of the dynamic STEMPO tompography obtained by optimizing the log-posterior \eqref{eq:logpost} in the original space of $\bXi$ are compared among STBP, STGP, and time-uncorrelated models. Although we still observe the better reconstruction by STBP (the third column) compared with the other two (the forth and the last columns), these results are generally more noisy with larger errors compared with those obtained by optimizing \eqref{eq:logpost_whiten} in the whitened space of $\bZeta$, as illustrated in Figure \ref{fig:STEMPO_MAP_whiten}. Such comparison not only supports the superior performance of STBP, but also highlights the benefit of the white noise representation \eqref{eq:Lambda}, which is also verified in Figure \ref{fig:STEMPO_err}.

\begin{figure}[t]
\centering
\includegraphics[width=.49\textwidth,height=.185\textwidth]{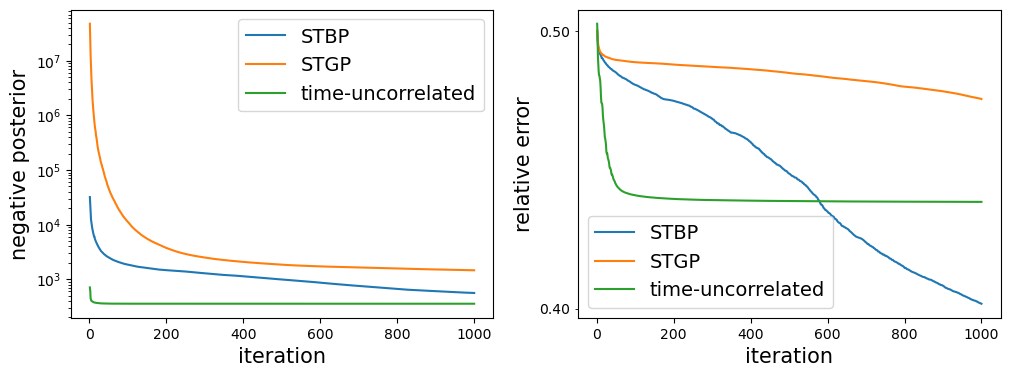}
\includegraphics[width=.49\textwidth,height=.185\textwidth]{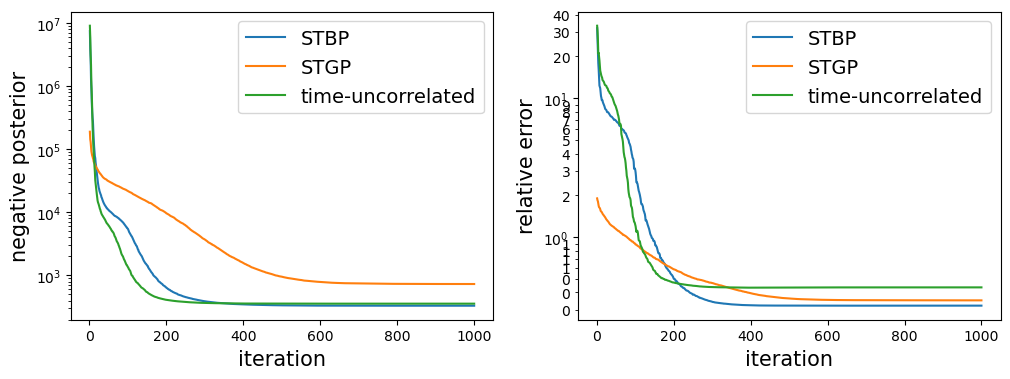}
\caption{Dynamic STEMPO tomography reconstruction: negative posterior densities and relative errors for the optimization in the original space (left two) and in the whitened space (right two) as functions of iterations in the BFGS algorithm used to obtain MAP estimates. Early termination is implemented if the error falls below some threshold or the maximal iteration (1000) is reached.}
\label{fig:STEMPO_err}
\end{figure}

Figure \ref{fig:STEMPO_err} compares minimizing the negative log-posterior \eqref{eq:logpost} in the original space (the left two panels) with minimizing the negative log-posterior \eqref{eq:logpost_whiten} in the whitened space (the right two panels). The speed-up of the convergence in the whitened space may be explained by the de-correlated coordinates. Though time-uncorrelated model converges faster, STBP and STGP could achieve lower relative errors by accounting for time correlations.

\renewcommand{\arraystretch}{0.6}
\begin{table}[htbp]
\caption{Comparison of MAP estimates for STEMPO tomography generated by STBP, STGP and time-uncorrelated prior models
in terms of RLE, log-likelihood, PSNR, and SSIM 
measures. Standard deviations (in bracket) are obtained by repeating the experiments for 10 times with different random seeds for initialization.}
\centering
\begin{tabular}{|c|c|c|c|}
\toprule
 & time-uncorrelated & STGP & STBP  \\
 \midrule
RLE &  0.4354 (2.91e-5)  & 0.3512(1.42e-4) & \cellcolor{lightgray} 0.3217 (2.72e-5) \\
log-likelihood & -39190.72 (0.65) & -39085.37 (5.49) & -39697.93 (0.71)\\
PSNR & 16.6235 (5.80e-4) & 18.4896 (3.50e-3) & \cellcolor{gray} 19.2532 (7.33e-4) \\
SSIM & 0.1469 (3.50e-5) & \cellcolor{gray} 0.3486 (3.47e-4) & 0.2318 (7.10e-5) \\
\bottomrule
\end{tabular}
\label{tab:STEMPO}
\end{table}

Table \ref{tab:STEMPO} compares the three models in terms of relative error and other image reconstruction metrics like PSNR, and SSIM. The STBP model performs the best and generates the best reconstruction with the lowest error. The same conclusion can be drawn with these high values of the image quality measurements.

\begin{figure}[!ht]
\begin{tabular}{ccccc}
\qquad Truth & \quad \qquad STBP (mean) & \quad STGP (mean) & \quad STBP (std) & \quad STGP (std)  \\
\end{tabular}
\vspace{-5pt}
\begin{tabular}{ccccc}
\includegraphics[width=0.19\textwidth,height=.18\textwidth]{STEMPO/truth/stempo_00.png}
\includegraphics[width=0.19\textwidth,height=.18\textwidth]{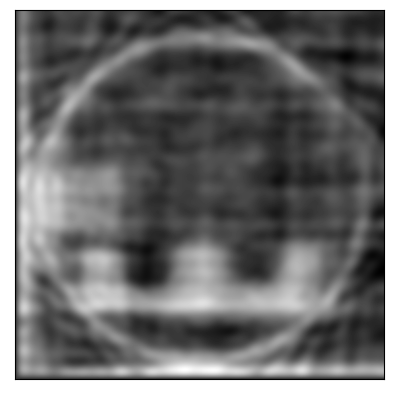}
\includegraphics[width=0.19\textwidth,height=.18\textwidth]{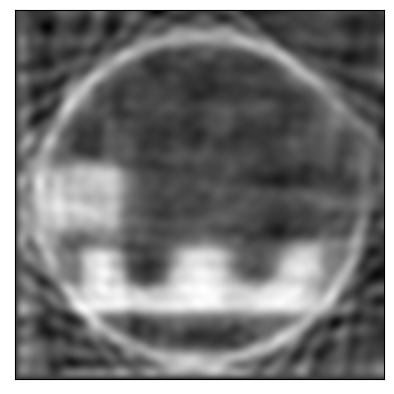}
\includegraphics[width=0.19\textwidth,height=.18\textwidth]{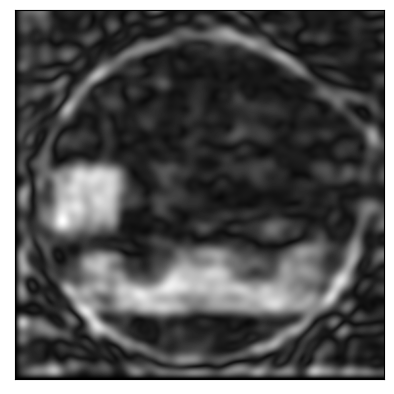}
\includegraphics[width=0.19\textwidth,height=.18\textwidth]{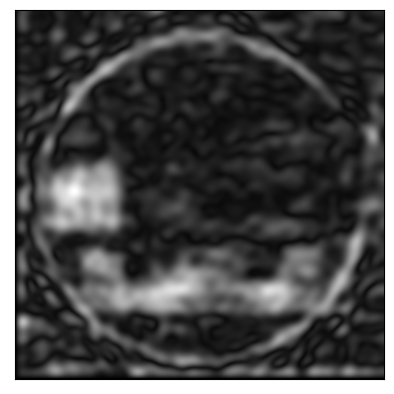}
\end{tabular}
\vspace{-5pt}
\begin{tabular}{ccccc}
\includegraphics[width=0.19\textwidth,height=.18\textwidth]{STEMPO/truth/stempo_06.png}
\includegraphics[width=0.19\textwidth,height=.18\textwidth]{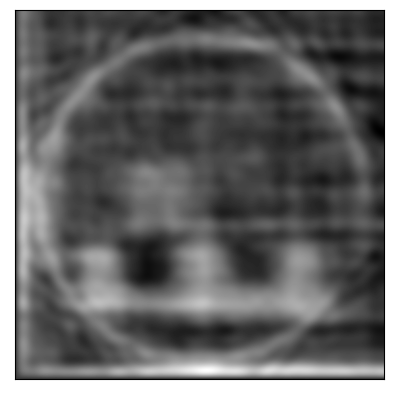}
\includegraphics[width=0.19\textwidth,height=.18\textwidth]{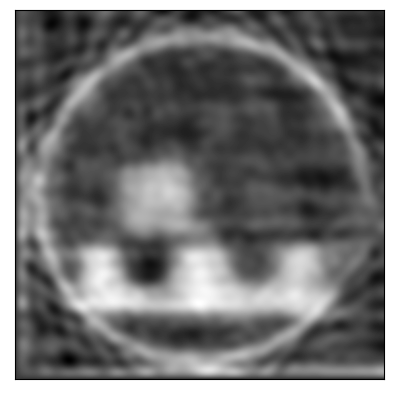}
\includegraphics[width=0.19\textwidth,height=.18\textwidth]{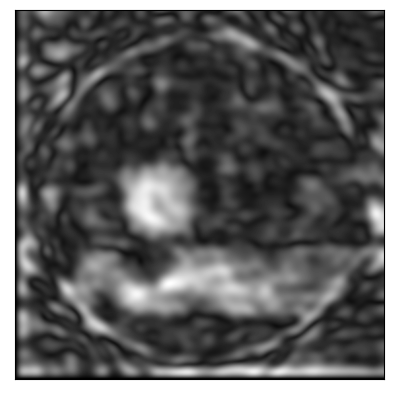}
\includegraphics[width=0.19\textwidth,height=.18\textwidth]{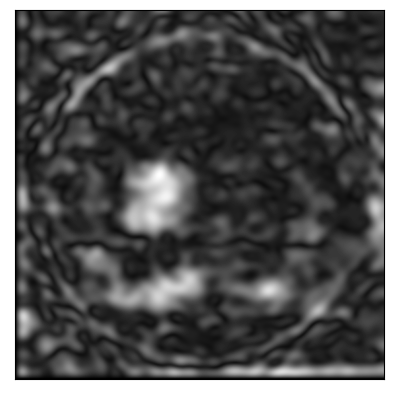}
\end{tabular}
\vspace{-5pt}
\begin{tabular}{ccccc}
\includegraphics[width=0.19\textwidth,height=.18\textwidth]{STEMPO/truth/stempo_13.png}
\includegraphics[width=0.19\textwidth,height=.18\textwidth]{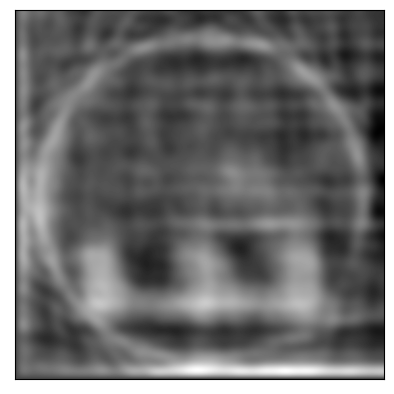}
\includegraphics[width=0.19\textwidth,height=.18\textwidth]{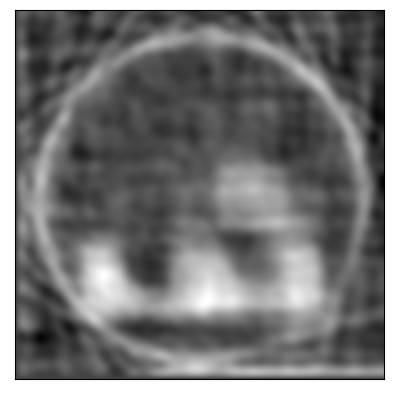}
\includegraphics[width=0.19\textwidth,height=.18\textwidth]{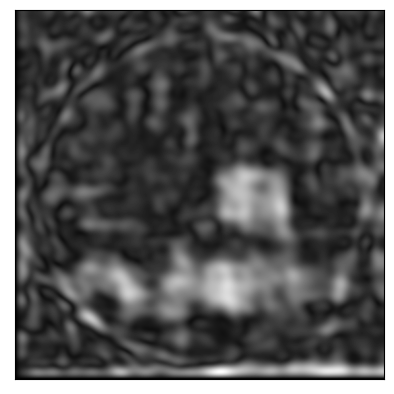}
\includegraphics[width=0.19\textwidth,height=.18\textwidth]{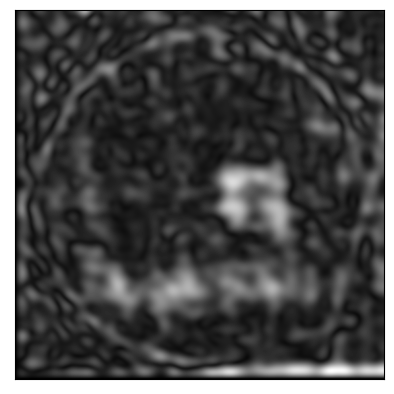}
\end{tabular}
\vspace{-5pt}
\begin{tabular}{ccccc}
\includegraphics[width=0.19\textwidth,height=.18\textwidth]{STEMPO/truth/stempo_19.png}
\includegraphics[width=0.19\textwidth,height=.18\textwidth]{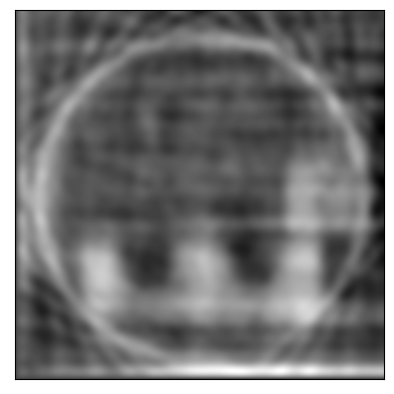}
\includegraphics[width=0.19\textwidth,height=.18\textwidth]{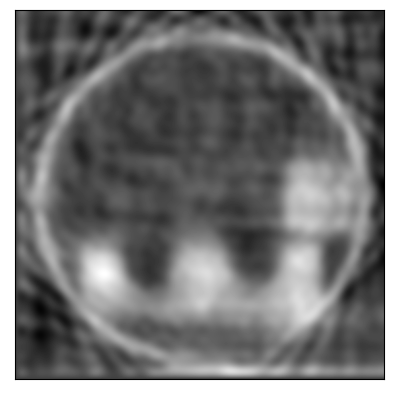}
\includegraphics[width=0.19\textwidth,height=.18\textwidth]{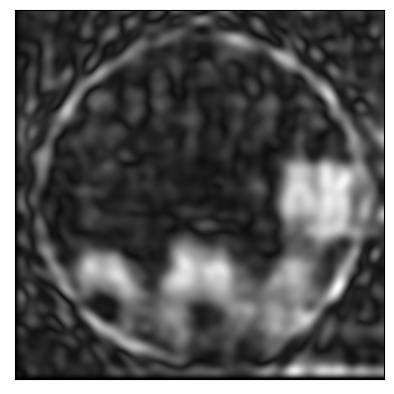}
\includegraphics[width=0.19\textwidth,height=.18\textwidth]{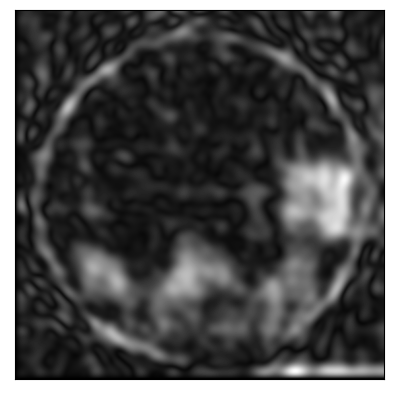}
\end{tabular}
\vspace{-5pt}
\caption{MCMC reconstruction of dynamic STEMPO tomography in the whitened space. Columns from left to right: true images, posterior mean estimates by STBP and STGP, posterior standard deviation estimates by STBP and STGP models respectively. Rows from top to bottom: time step $j = 0, 6, 13, 19$.}
\label{fig:STEMPO_MCMC}
\end{figure}

Figure \ref{fig:STEMPO_MCMC} compares the posterior estimates of the dynamic STEMPO tompography given by STBP (the second and forth columns) and STGP (the third and last columns) models.
Note although the posterior mean estimates are not as good as their MAP estimates, the posterior standard deviations by STBP (the forth column) have clearer spatial features compared with those by STGP model (the last column).

\begin{figure}[!ht]
\begin{tabular}{ccccc}
\quad Truth & \qquad \qquad Observations & \quad \qquad STBP & \qquad \qquad STGP & \quad \quad time-uncorrelated \\
\end{tabular}
\vspace{-5pt}
\begin{tabular}{ccccc}
\raisebox{2pt}{\includegraphics[width=0.185\textwidth,height=.17\textwidth]{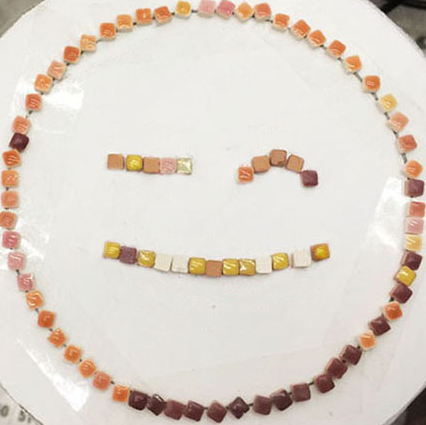}}
\includegraphics[width=0.19\textwidth,height=.18\textwidth]{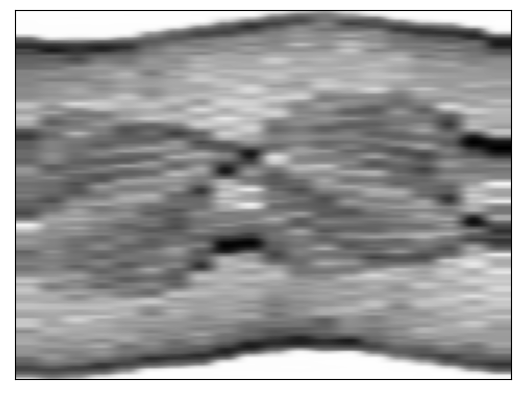}
\includegraphics[width=0.19\textwidth,height=.18\textwidth]{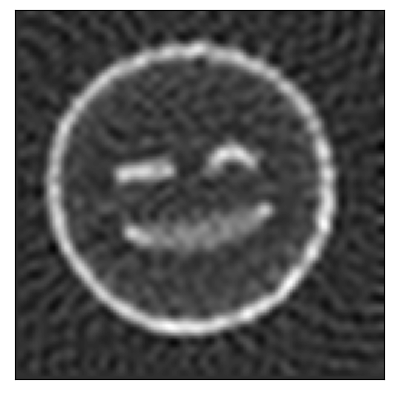}
\includegraphics[width=0.19\textwidth,height=.18\textwidth]{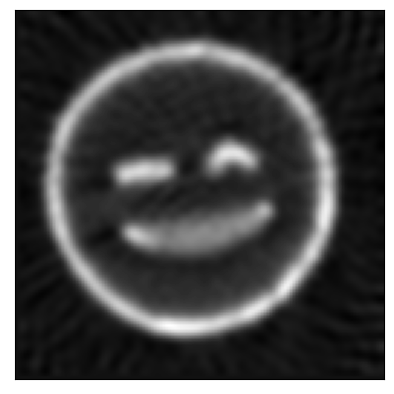}
\includegraphics[width=0.19\textwidth,height=.18\textwidth]{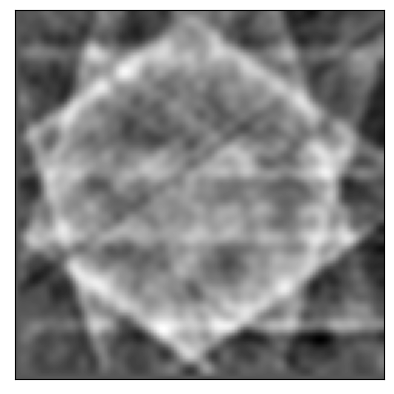}
\end{tabular}
\vspace{-5pt}
\begin{tabular}{ccccc}
\raisebox{2pt}{\includegraphics[width=0.185\textwidth,height=.17\textwidth]{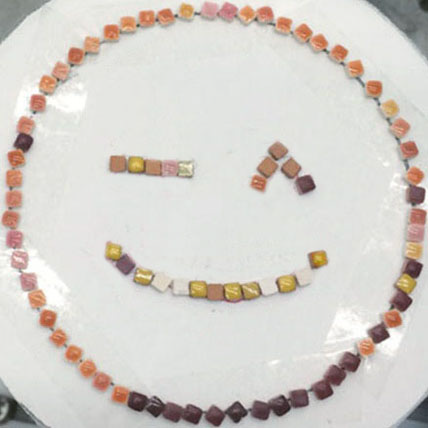}}
\includegraphics[width=0.19\textwidth,height=.18\textwidth]{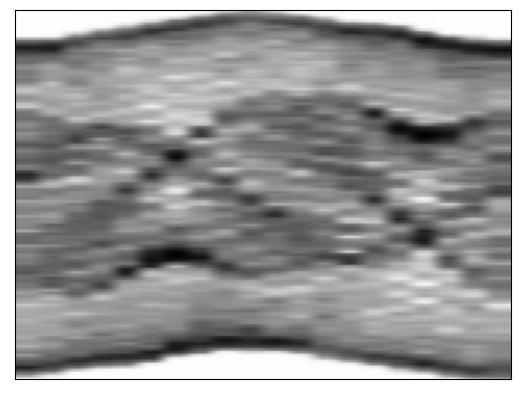} 
\includegraphics[width=0.19\textwidth,height=.18\textwidth]{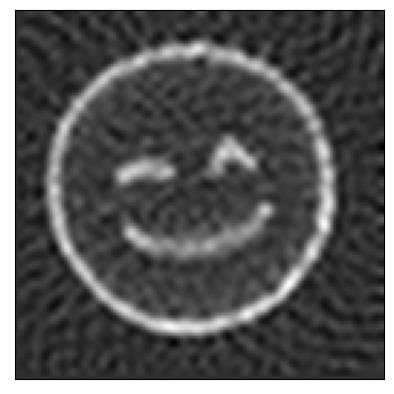} 
\includegraphics[width=0.19\textwidth,height=.18\textwidth]{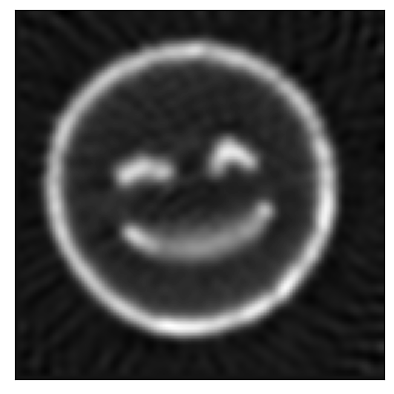} 
\includegraphics[width=0.19\textwidth,height=.18\textwidth]{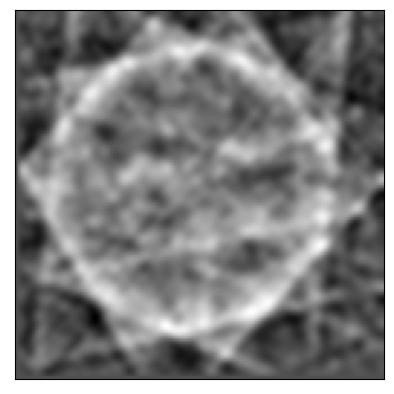} 
\end{tabular}
\vspace{-5pt}
\begin{tabular}{ccccc}
\raisebox{2pt}{\includegraphics[width=0.185\textwidth,height=.17\textwidth]{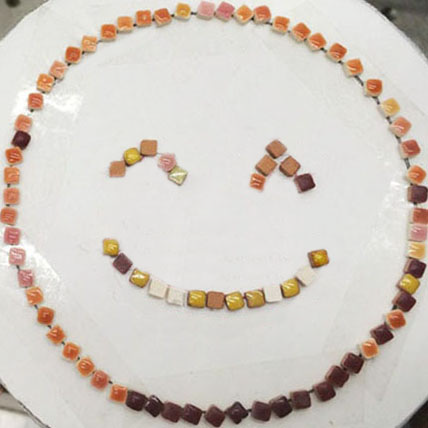}}
\includegraphics[width=0.19\textwidth,height=.18\textwidth]{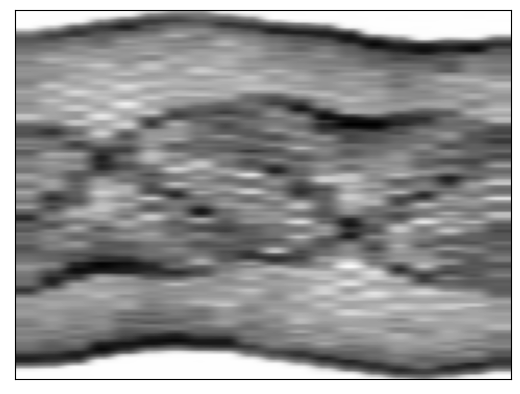}
\includegraphics[width=0.19\textwidth,height=.18\textwidth]{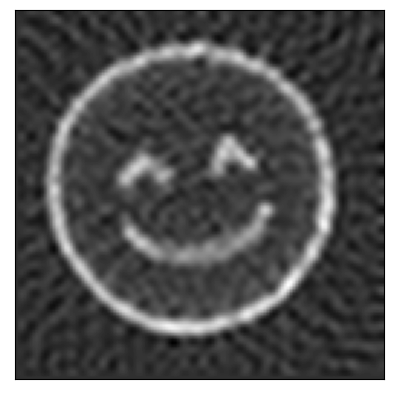}
\includegraphics[width=0.19\textwidth,height=.18\textwidth]{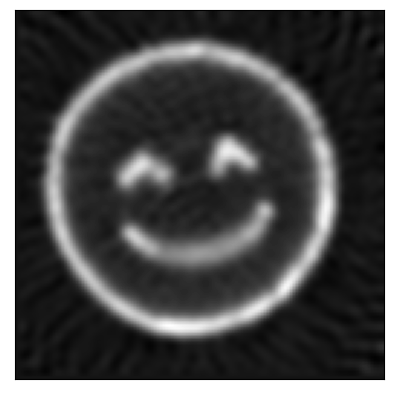}
\includegraphics[width=0.19\textwidth,height=.18\textwidth]{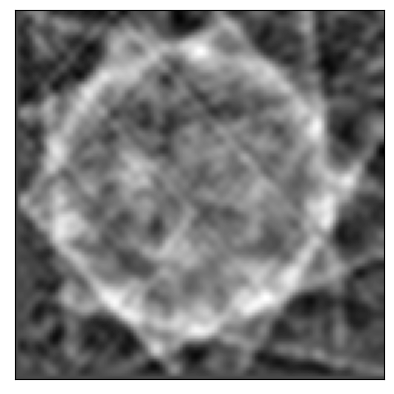}
\end{tabular}
\vspace{-5pt}
\begin{tabular}{ccccc}
\raisebox{2pt}{\includegraphics[width=0.185\textwidth,height=.17\textwidth]{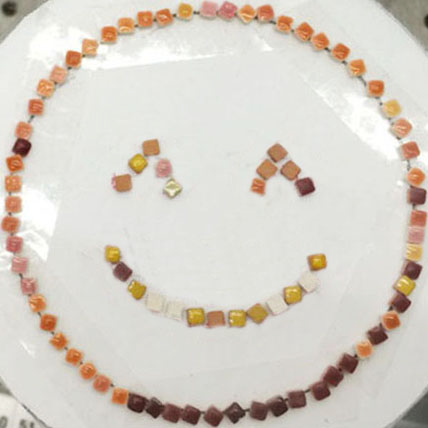}}
\includegraphics[width=0.19\textwidth,height=.18\textwidth]{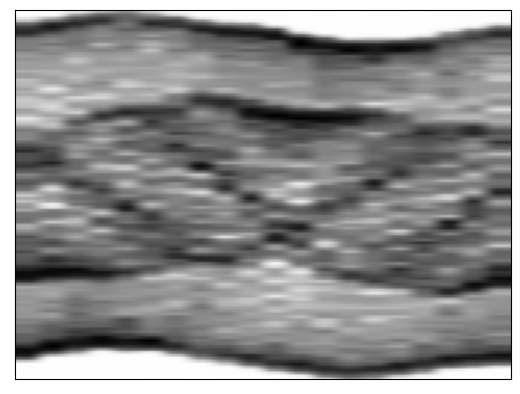}
\includegraphics[width=0.19\textwidth,height=.18\textwidth]{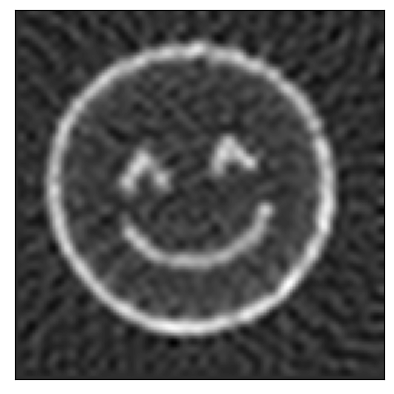}
\includegraphics[width=0.19\textwidth,height=.18\textwidth]{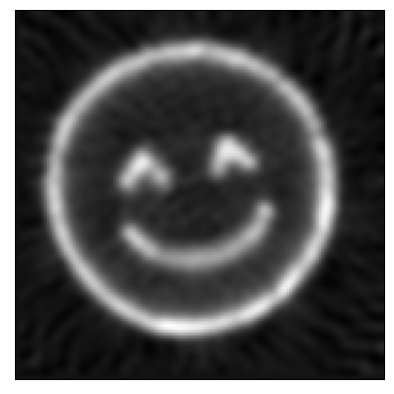}
\includegraphics[width=0.19\textwidth,height=.18\textwidth]{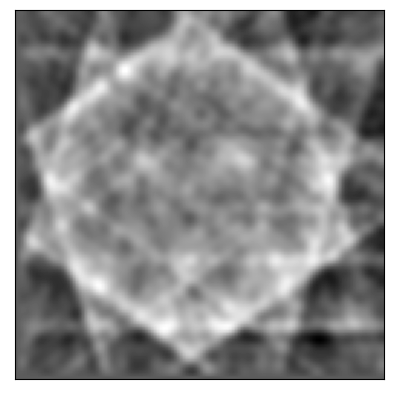}
\end{tabular}
\vspace{-5pt}
\caption{MAP reconstruction for the dynamic emoji tomography in the whitened space. Columns from left to right: true images, sinograms, MAP estimates by STBP, STGP and time-uncorrelated models respectively. Rows from top to bottom: time step $j = 6, 14, 22, 30$.
}
\label{fig:emoji_30proj_whiten_MAP}
\end{figure}
\subsubsection{Emoji Tomography}\label{apx:emoji}
Next, we test our methods on a real data of dynamic ``emoji'' phantom measured at the
University of Helsinki \citep[See more details in][about the machine (forward operator) set-up and data collection]{meaney2018tomographic}. 
The available spatiotemporal data represent $J=33$ time steps of a series of the X-ray sinogram of emojis made of small ceramic stones obtained by shining $n_s=217$ X-ray projections from $n_a=10$ angles.

\begin{figure}[ht]
\begin{tabular}{ccccc}
\quad Truth & \qquad \qquad Observations & \quad \qquad STBP & \qquad \qquad STGP & \quad \quad time-uncorrelated \\
\end{tabular}
\vspace{-5pt}
\begin{tabular}{ccccc}
\raisebox{2pt}{\includegraphics[width=0.185\textwidth,height=.17\textwidth]{emoji/true_pics/smile_true_3.png}}
\includegraphics[width=0.19\textwidth,height=.18\textwidth]{emoji/sino_60proj/emoji_06.png}
\includegraphics[width=0.19\textwidth,height=.18\textwidth]{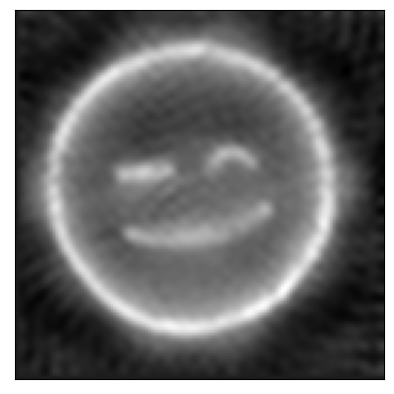}
\includegraphics[width=0.19\textwidth,height=.18\textwidth]{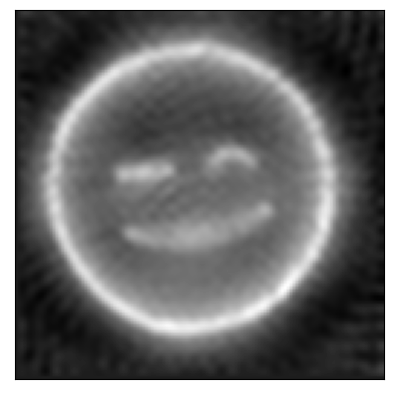}
\includegraphics[width=0.19\textwidth,height=.18\textwidth]{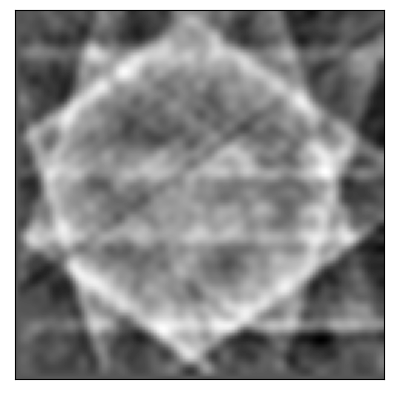}
\end{tabular}
\vspace{-5pt}
\begin{tabular}{ccccc}
\raisebox{2pt}{\includegraphics[width=0.185\textwidth,height=.17\textwidth]{emoji/true_pics/smile_true_7.png}}
\includegraphics[width=0.19\textwidth,height=.18\textwidth]{emoji/sino_60proj/emoji_14.png} 
\includegraphics[width=0.19\textwidth,height=.18\textwidth]{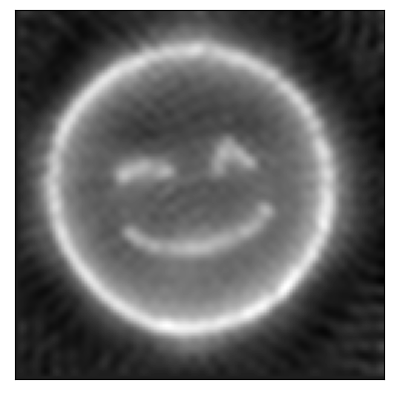} 
\includegraphics[width=0.19\textwidth,height=.18\textwidth]{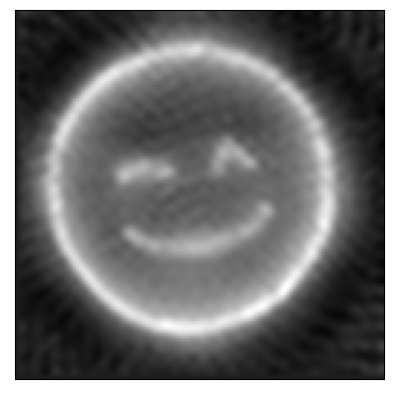} 
\includegraphics[width=0.19\textwidth,height=.18\textwidth]{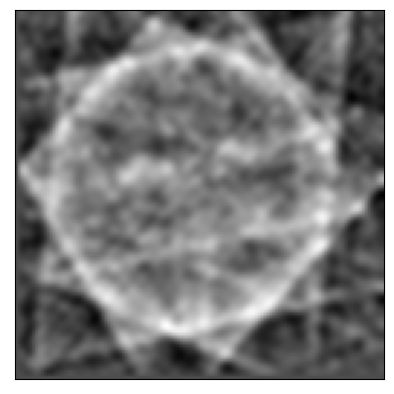} 
\end{tabular}
\vspace{-5pt}
\begin{tabular}{ccccc}
\raisebox{2pt}{\includegraphics[width=0.185\textwidth,height=.17\textwidth]{emoji/true_pics/smile_true_11.png}}
\includegraphics[width=0.19\textwidth,height=.18\textwidth]{emoji/sino_60proj/emoji_22.png}
\includegraphics[width=0.19\textwidth,height=.18\textwidth]{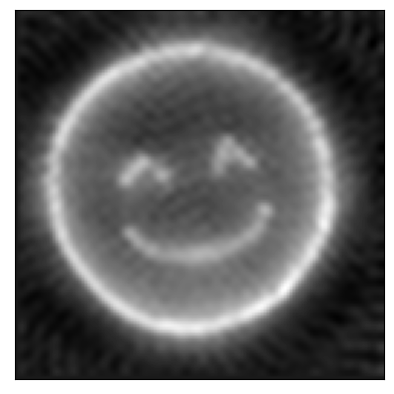}
\includegraphics[width=0.19\textwidth,height=.18\textwidth]{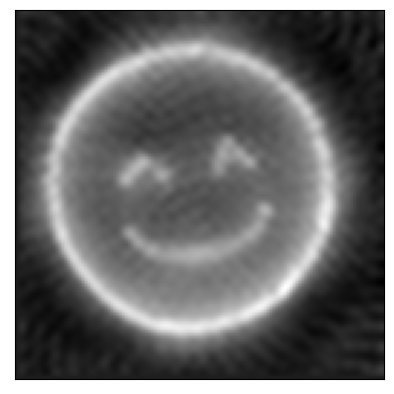}
\includegraphics[width=0.19\textwidth,height=.18\textwidth]{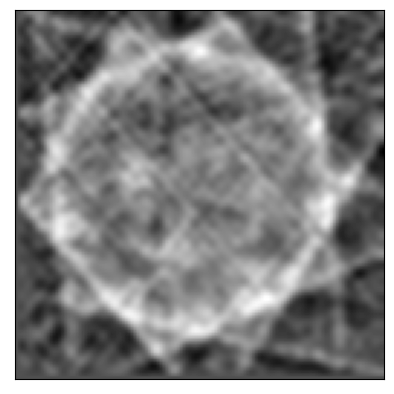}
\end{tabular}
\vspace{-5pt}
\begin{tabular}{ccccc}
\raisebox{2pt}{\includegraphics[width=0.185\textwidth,height=.17\textwidth]{emoji/true_pics/smile_true_15.png}}
\includegraphics[width=0.19\textwidth,height=.18\textwidth]{emoji/sino_60proj/emoji_30.png}
\includegraphics[width=0.19\textwidth,height=.18\textwidth]{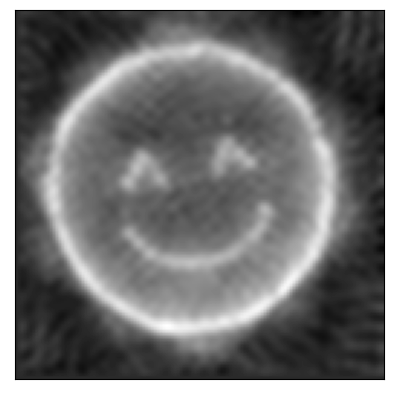}
\includegraphics[width=0.19\textwidth,height=.18\textwidth]{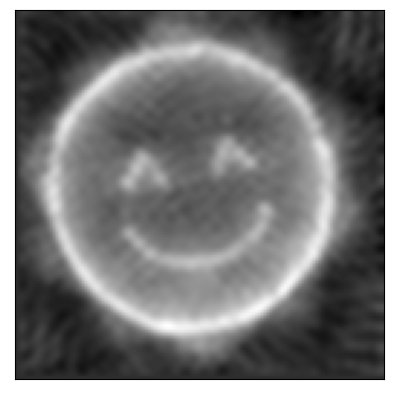}
\includegraphics[width=0.19\textwidth,height=.18\textwidth]{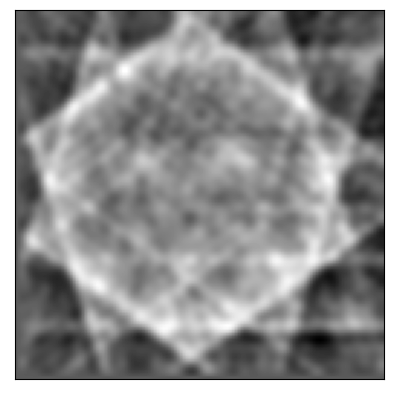}
\end{tabular}
\vspace{-5pt}
\caption{MAP reconstruction for the dynamic emoji tomography in the original space. Columns from left to right: true images, sinograms, MAP estimates by STBP, STGP and time-uncorrelated models respectively. Rows from top to bottom: time step $j = 6, 14, 22, 30$.
}
\label{fig:emoji_30proj_MAP}
\end{figure}

The inverse problem involves reconstructing a sequence of images $u(\bX, t_j)$, $t=1,2,\dots, J$, each of size $I = 128\times 128$, from low-dose observations measured at a limited number of $n_a$ angles. Hence, the unknown images are collected in $\bU = u(\bX, \bt) \in \mbR^{16,384\times 33}$, representing the dynamic sequence of the emoji images changing from an expressionless face with closed eyes and a straight mouth to a face with smiling eyes and mouth, where the outmost circular shape does not change. We refer to Figure \ref{fig:emoji_30proj_whiten_MAP} for a sample of 4 setup images (first column) and sinograms (second column) at time steps $t= 6,14,22,30$.
The low-dose observations are modeled as in the model \eqref{eq:STBP_model}: $\bY_{j}\sim \mN_{2170}(\mG_j(u^\dagger(\bX, t_j)), \Gamma_\mathrm{noise})$ with $\Gamma_\mathrm{noise}$ being the empirical covariance obtained from $J=33$ images, and measurement matrix $\mG_j$ being the result of the same Radon transform as above that represents line integrals \citep{meaney2018tomographic}.
Although the ground truth is not available, we can qualitatively compare the visual outputs from STBP, STGP and time-uncorrelated models.

\begin{figure}[htbp]
\begin{tabular}{ccccc}
\quad Truth & \qquad \qquad STBP (mean) & \quad STGP (mean) & \quad STBP (std) & \quad STGP (std)  \\
\end{tabular}
\vspace{-5pt}
\begin{tabular}{ccccc}
\raisebox{2pt}{\includegraphics[width=0.185\textwidth,height=.17\textwidth]{emoji/true_pics/smile_true_3.png}}
\includegraphics[width=0.19\textwidth,height=.18\textwidth]{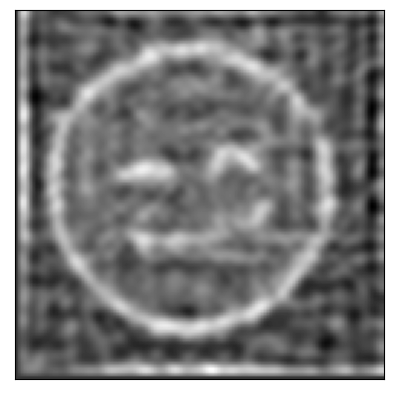}
\includegraphics[width=0.19\textwidth,height=.18\textwidth]{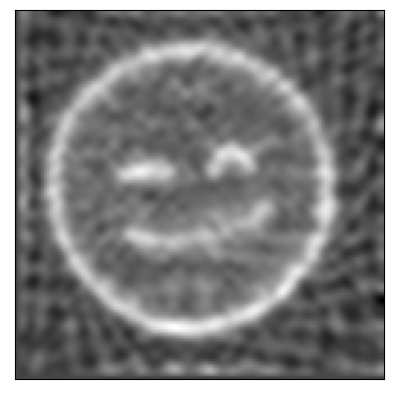}
\includegraphics[width=0.19\textwidth,height=.18\textwidth]{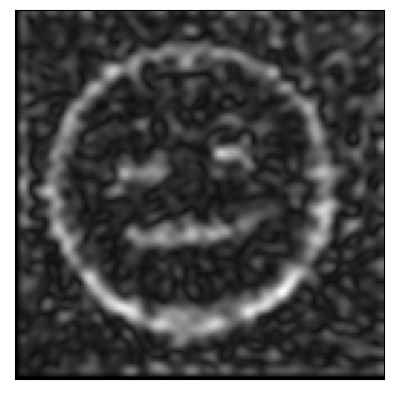}
\includegraphics[width=0.19\textwidth,height=.18\textwidth]{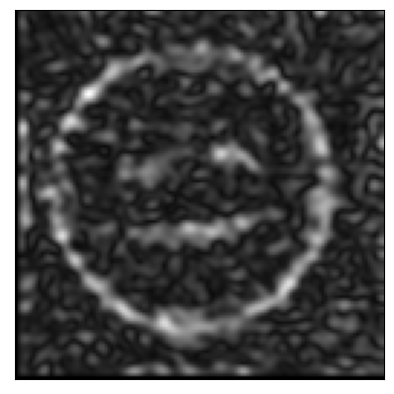}
\end{tabular}
\vspace{-5pt}
\begin{tabular}{ccccc}
\raisebox{2pt}{\includegraphics[width=0.185\textwidth,height=.17\textwidth]{emoji/true_pics/smile_true_7.png}}
\includegraphics[width=0.19\textwidth,height=.18\textwidth]{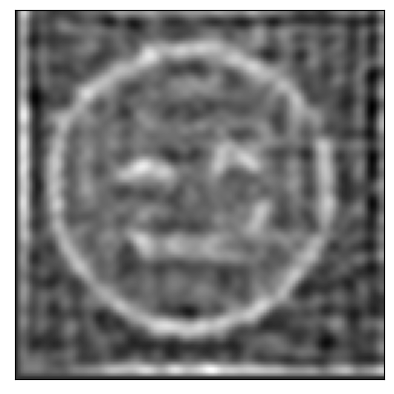}
\includegraphics[width=0.19\textwidth,height=.18\textwidth]{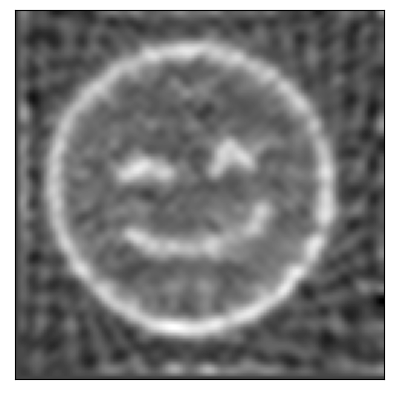}
\includegraphics[width=0.19\textwidth,height=.18\textwidth]{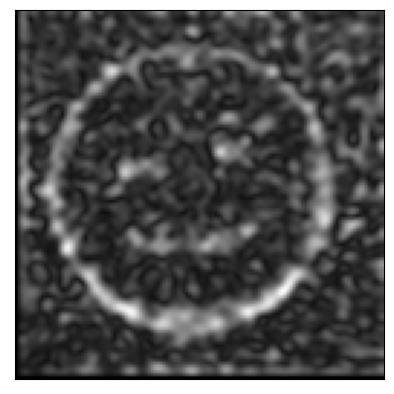}
\includegraphics[width=0.19\textwidth,height=.18\textwidth]{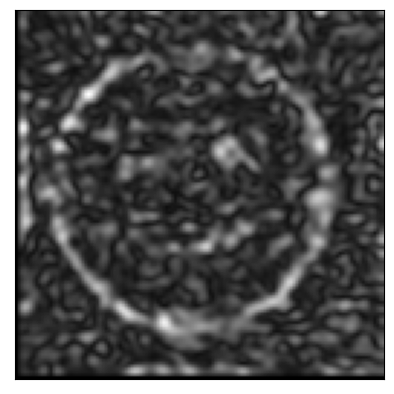}
\end{tabular}
\vspace{-5pt}
\begin{tabular}{ccccc}
\raisebox{2pt}{\includegraphics[width=0.185\textwidth,height=.17\textwidth]{emoji/true_pics/smile_true_11.png}}
\includegraphics[width=0.19\textwidth,height=.18\textwidth]{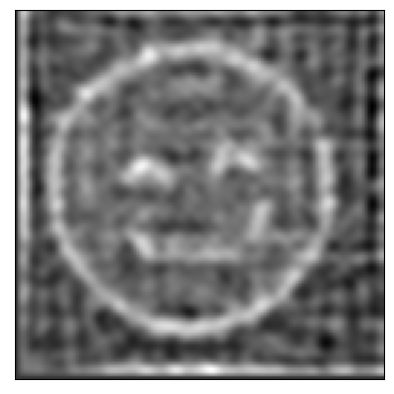}
\includegraphics[width=0.19\textwidth,height=.18\textwidth]{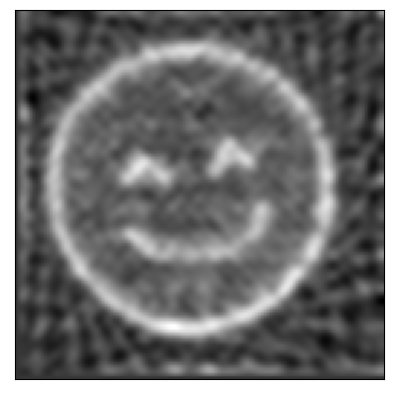}
\includegraphics[width=0.19\textwidth,height=.18\textwidth]{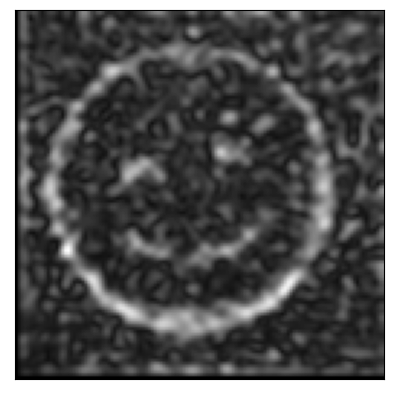}
\includegraphics[width=0.19\textwidth,height=.18\textwidth]{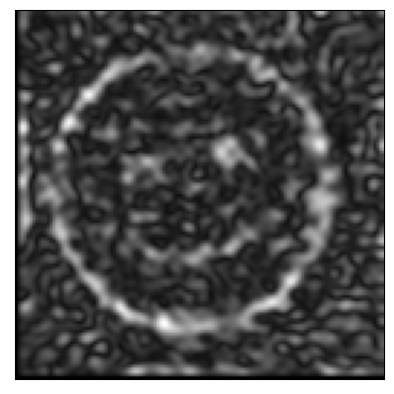}
\end{tabular}
\vspace{-5pt}
\begin{tabular}{ccccc}
\raisebox{2pt}{\includegraphics[width=0.185\textwidth,height=.17\textwidth]{emoji/true_pics/smile_true_15.png}}
\includegraphics[width=0.19\textwidth,height=.18\textwidth]{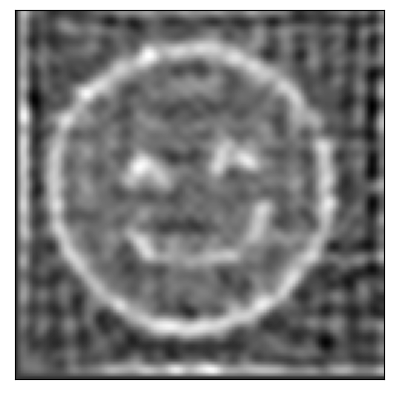}
\includegraphics[width=0.19\textwidth,height=.18\textwidth]{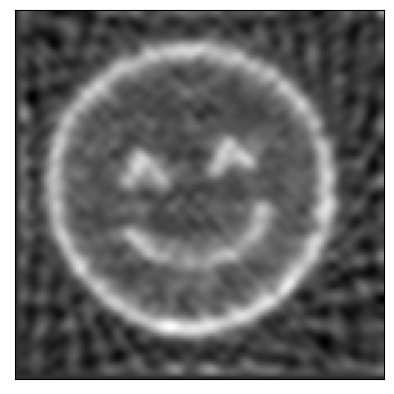}
\includegraphics[width=0.19\textwidth,height=.18\textwidth]{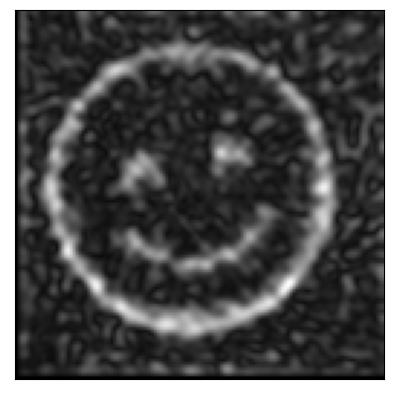}
\includegraphics[width=0.19\textwidth,height=.18\textwidth]{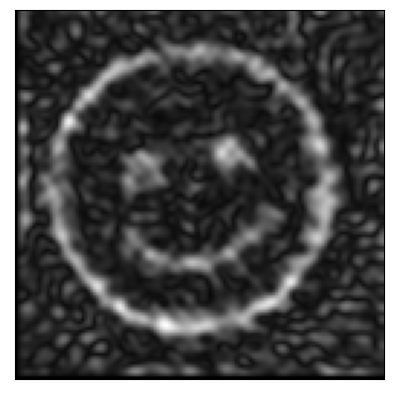}
\end{tabular}
\vspace{-5pt}
\caption{MCMC reconstruction of dynamic emoji tomography in the whitened space. Columns from left to right: true images, posterior mean estimates by STBP and STGP, posterior standard deviation estimates by STBP and STGP models respectively. Rows from top to bottom: time step $j = 6, 14, 22, 30$.}
\label{fig:emoji_60proj_MCMC}
\end{figure}

Figure \ref{fig:emoji_30proj_whiten_MAP} compares the MAP estimates by STBP (the third column), STGP (the forth row) and the time-uncorrelated (the last column) prior models in the whitened space. Again we observe similar advantage in reconstructing a sequence of sharper tomography images by STBP compared with those more blurry results by STGP. Note, due to the absence of temporal correlation, the time-uncorrelated prior model reconstructs images that are noisy and difficult to recognize.

Similarly in the example of dynamic reconstruction of STEMPO tomography, MAP estimates in the original space shown in Figure \ref{fig:emoji_30proj_MAP} demonstrate similar contrast among the three models, STBP, STGP and time-uncorrelated; but are also more blurry compared with those obtain in the whitened space as in Figure \ref{fig:emoji_30proj_whiten_MAP}.

We also compare the UQ results generated by wn-$\infty$-mMALA (See Algorithm \ref{alg:wn-infMC}) for the two models, STBP and STGP, respectively in Figure \ref{fig:emoji_60proj_MCMC}. Again we observe noisy posterior mean estimates (the second and the third columns) for both models compared with MAP estimates plotted in Figure \ref{fig:emoji_30proj_whiten_MAP} due to the limited samples. However, the posterior standard deviation estimates by STBP (the forth column) are slightly clearer than those by STGP (the last column) in characterizing the uncertainty field representing the changing smiling faces.


\subsection{Navier-Stokes Inverse Problems}

\begin{figure}[t]
\begin{tabular}{ccccc}
\quad Truth & \qquad \qquad Observations & \quad \qquad STBP & \qquad \qquad STGP & \quad \quad time-uncorrelated  \\
\end{tabular}
\vspace{-5pt}
\begin{tabular}{ccccc}
\includegraphics[width=0.19\textwidth,height=.18\textwidth]{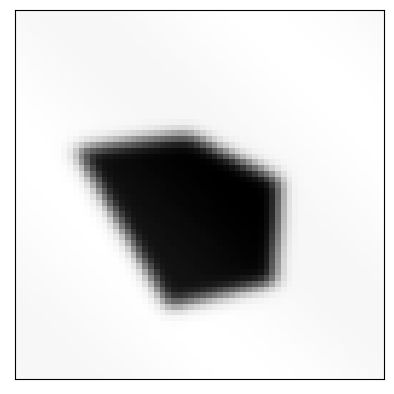}
\includegraphics[width=0.19\textwidth,height=.18\textwidth]{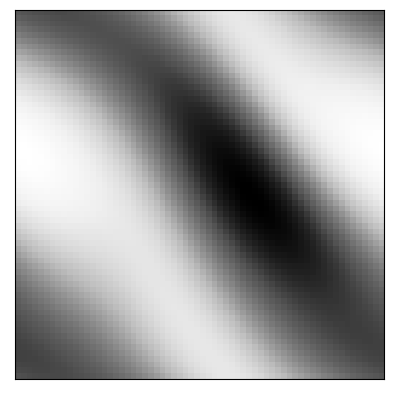}
\includegraphics[width=0.19\textwidth,height=.18\textwidth]{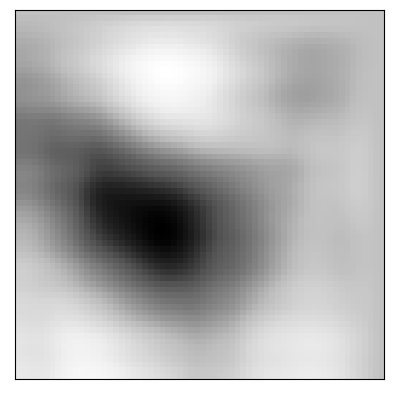}
\includegraphics[width=0.19\textwidth,height=.18\textwidth]{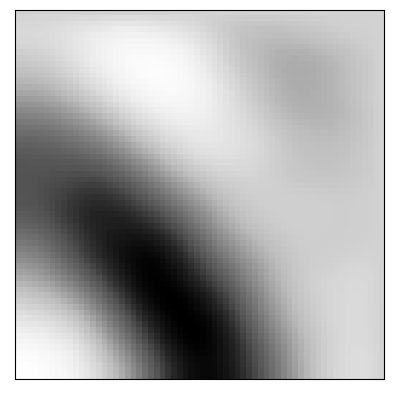}
\includegraphics[width=0.19\textwidth,height=.18\textwidth]{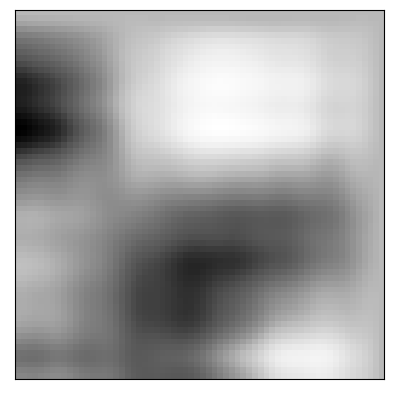}
\end{tabular}
\vspace{-5pt}
\begin{tabular}{ccccc}
\includegraphics[width=0.19\textwidth,height=.18\textwidth]{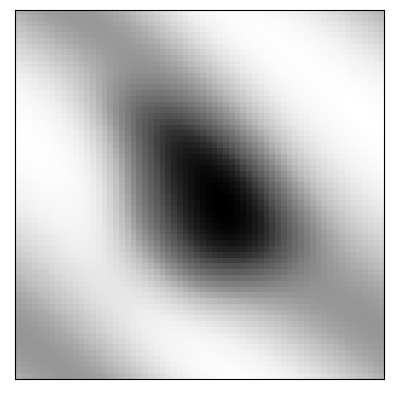}
\includegraphics[width=0.19\textwidth,height=.18\textwidth]{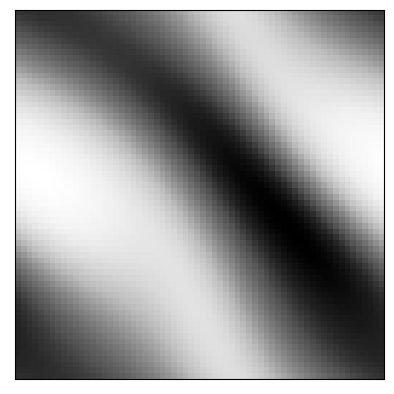} 
\includegraphics[width=0.19\textwidth,height=.18\textwidth]{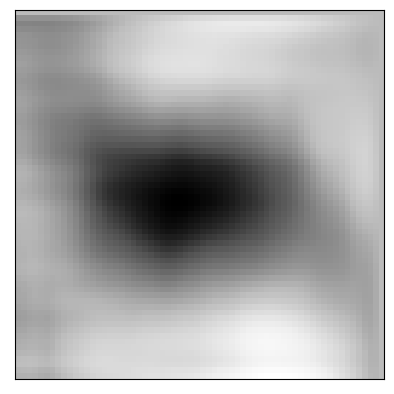} 
\includegraphics[width=0.19\textwidth,height=.18\textwidth]{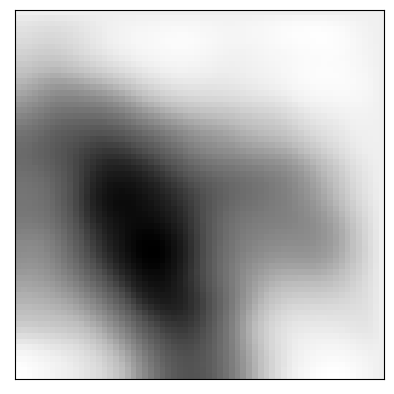} 
\includegraphics[width=0.19\textwidth,height=.18\textwidth]{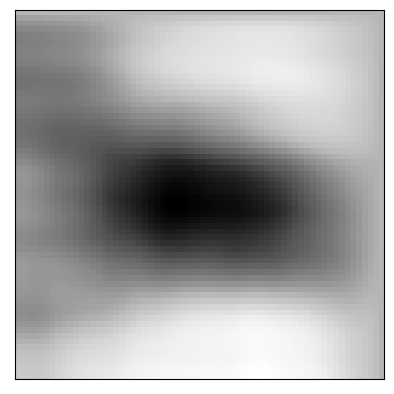} 
\end{tabular}
\vspace{-5pt}
\begin{tabular}{ccccc}
\includegraphics[width=0.19\textwidth,height=.18\textwidth]{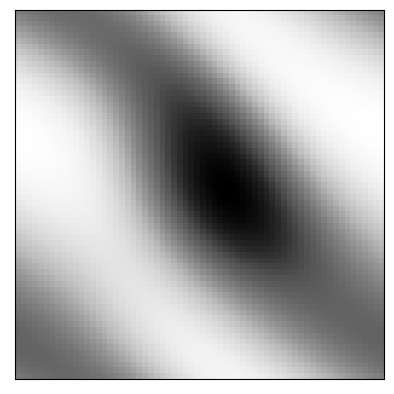}
\includegraphics[width=0.19\textwidth,height=.18\textwidth]{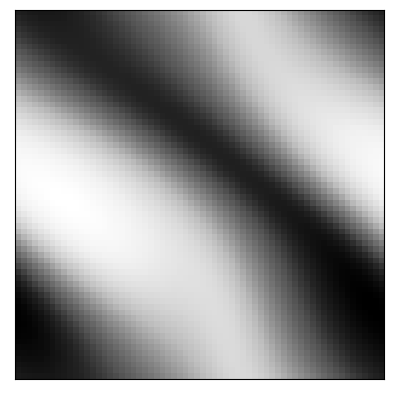}
\includegraphics[width=0.19\textwidth,height=.18\textwidth]{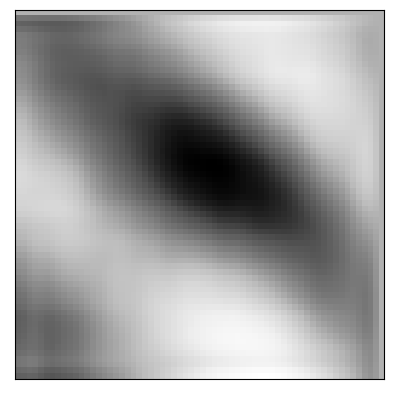}
\includegraphics[width=0.19\textwidth,height=.18\textwidth]{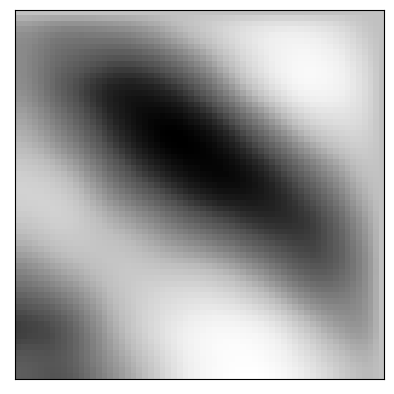}
\includegraphics[width=0.19\textwidth,height=.18\textwidth]{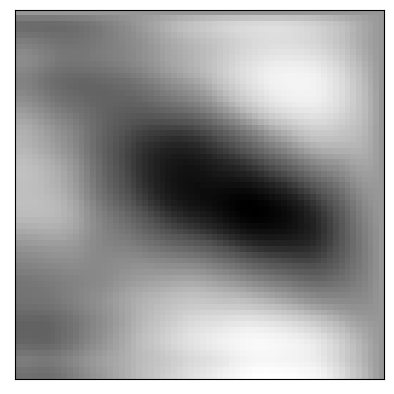}
\end{tabular}
\vspace{-5pt}
\begin{tabular}{ccccc}
\includegraphics[width=0.19\textwidth,height=.18\textwidth]{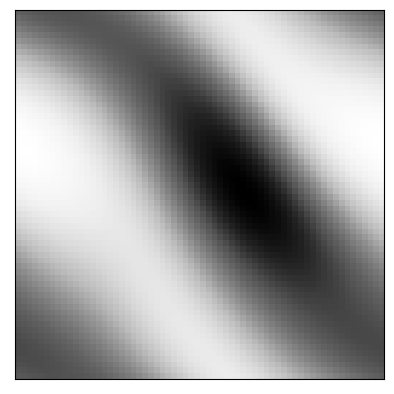}
\includegraphics[width=0.19\textwidth,height=.18\textwidth]{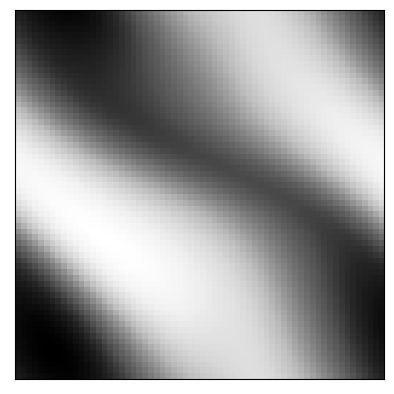}
\includegraphics[width=0.19\textwidth,height=.18\textwidth]{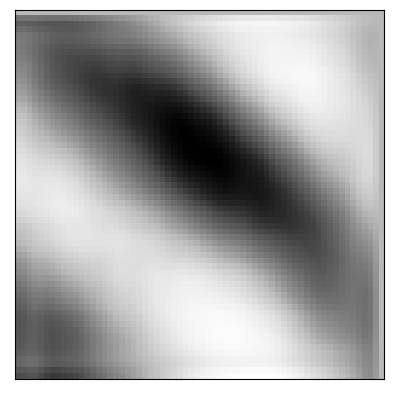}
\includegraphics[width=0.19\textwidth,height=.18\textwidth]{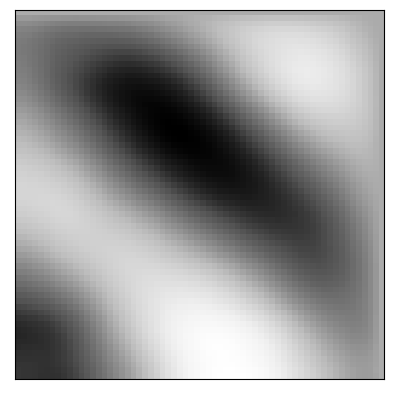}
\includegraphics[width=0.19\textwidth,height=.18\textwidth]{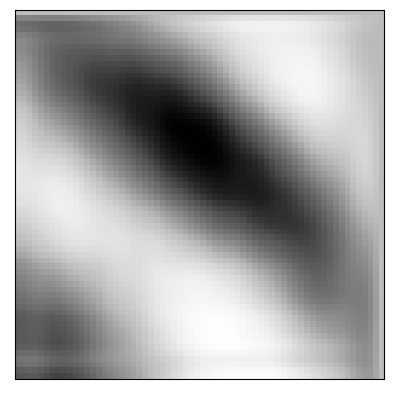}
\end{tabular}
\vspace{-5pt}
\caption{MAP inverse solutions of Navier-Stokes equation in the whitened space. Columns from left to right: true vorticity, observations, MAP estimates by STBP, STGP and time-uncorrelated models respectively. Rows from top to bottom: time step $j = 0, 3, 6, 9$.}
\label{fig:NSE_whiten_MAP}
\end{figure}

In the inverse problem governed by the Navier-Stokes equation (NSE),
we seek a spatiotemporal solution demonstrated in the first column of Figure \ref{fig:NSE_whiten_MAP}. It is more challenging than the traditional inverse problem for just the initial condition based on the same amount of downstream observations, illustrated in the second column.
The last three columns compare the MAP estimates the three models:
STBP (the third column) has the inverse solution closest to the truth, especially the initial condition at $t=0$ (the first row) which is the most challenging due the furthest distance to the observation window $(T_0,T]$. 
The time-uncorrelated prior model bears a solution trajectory that appears excessively erratic due to the lack of temporal correlation.

\begin{figure}[t]
\centering
\begin{tabular}{ccccc}
\quad $j = 0 $ & \qquad \qquad  $j =6$ & \quad \qquad \qquad  $j = 12$ &   \qquad \qquad  $j = 18$  &   \qquad \qquad  $j = 24$  \\
\end{tabular}
\vspace{-5pt}
{\includegraphics[width=0.19\textwidth,height=.19\textwidth]{NSE/observation/NSE_00.png}
\includegraphics[width=0.19\textwidth,height=.19\textwidth]{NSE/observation/NSE_06.png} 
\includegraphics[width=0.19\textwidth,height=.19\textwidth]{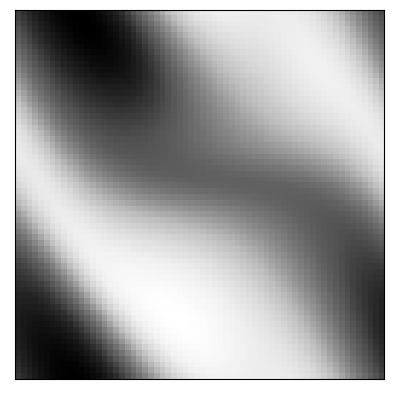}
\includegraphics[width=0.19\textwidth,height=.19\textwidth]{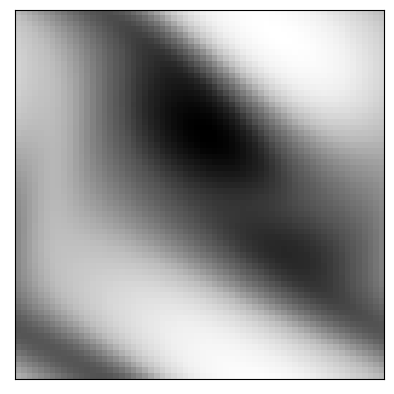}
\includegraphics[width=0.19\textwidth,height=.19\textwidth]{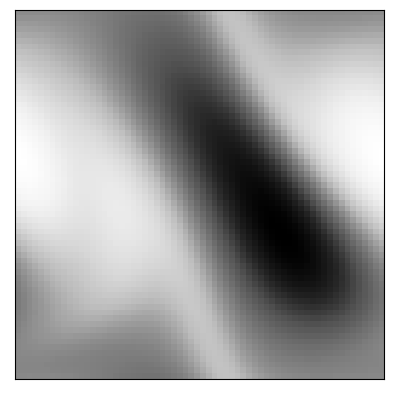}}
{\includegraphics[width=0.19\textwidth,height=.19\textwidth]{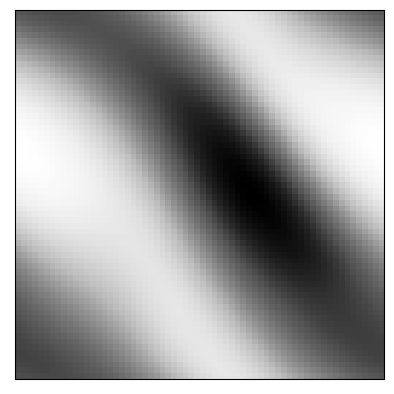}
\includegraphics[width=0.19\textwidth,height=.19\textwidth]{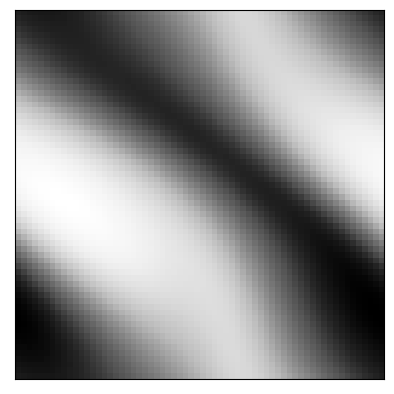} 
\includegraphics[width=0.19\textwidth,height=.19\textwidth]{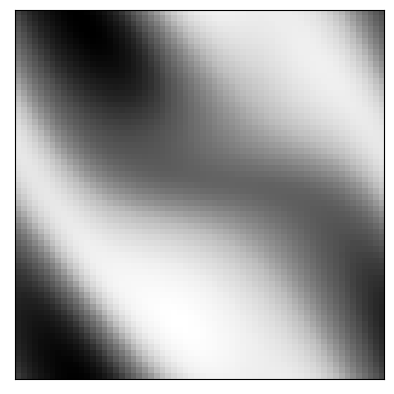}
\includegraphics[width=0.19\textwidth,height=.19\textwidth]{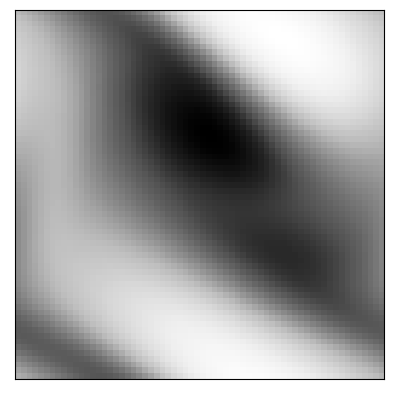}
\includegraphics[width=0.19\textwidth,height=.19\textwidth]{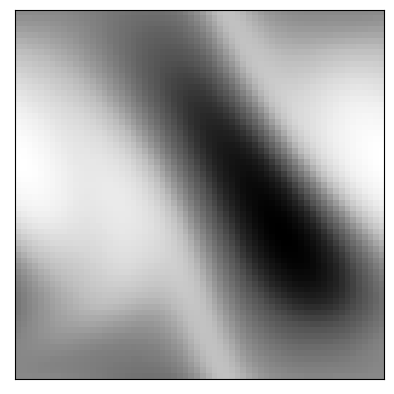}}
\caption{Observed and emulated Navier-Stokes equation solutions. Upper row: true NSE trajectory solved by classical PDE solver; Lower row: emulated NSE trajectory by FNO. Left to right: time step $j = 0, 6,12,18,24$ ($t_j\in(10,40]$).}
\label{fig:NSE_FNO}
\end{figure}

Figure \ref{fig:NSE_FNO} compares the true trajectory of the vorticity solved by the classical PDE solver (upper row) for the time window $(10,40]$ with the one emulated by FNO network (lower row). The negligible difference indicates that the trained FNO network serves as a very precise emulator of the PDE solver that can facilitate the Bayesian inference, which requires expensive repeated PDE forward solving but it is now replaced by cheap emulation.

\begin{figure}[ht]
\centering
\includegraphics[width=1\textwidth,height=.3\textwidth]{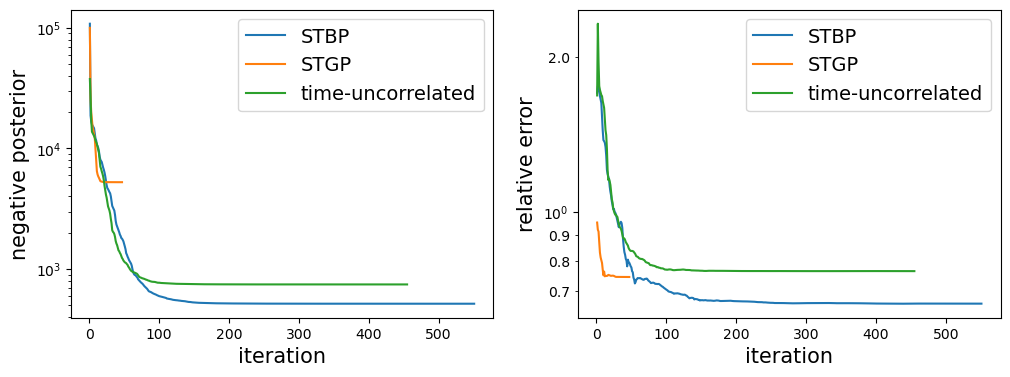}
\caption{Navier-Stokes inverse problem: negative posterior densities (left) and relative errors (right) for the optimization in the whitened space as functions of iterations in the BFGS algorithm used to obtain MAP estimates. Early termination is implemented if the error falls below some threshold or the maximal iteration (1000) is reached.}
\label{fig:NSE_err}
\end{figure}

Similarly as previous examples, Figure \ref{fig:NSE_err} shows the optimization in the whitened space converges faster to solutions with lower errors compared with that done in the original space. This confirms the benefit by our proposed white noise representation \eqref{eq:Lambda}.

\begin{figure}[ht]
\begin{tabular}{ccccc}
\quad Truth & \qquad \qquad STBP (mean) & \quad STGP (mean) & \quad STBP (std) & \quad STGP (std)  \\
\end{tabular}
\vspace{-5pt}
\begin{tabular}{ccccc}
\includegraphics[width=0.19\textwidth,height=.18\textwidth]{NSE/truth/NSE_00.png}
\includegraphics[width=0.19\textwidth,height=.18\textwidth]{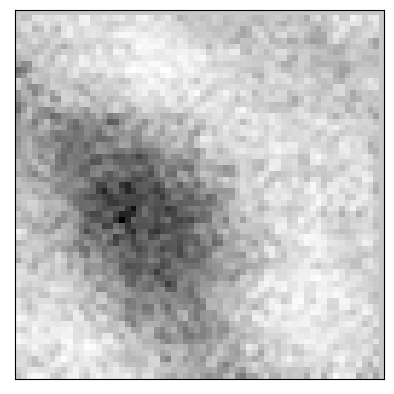}
\includegraphics[width=0.19\textwidth,height=.18\textwidth]{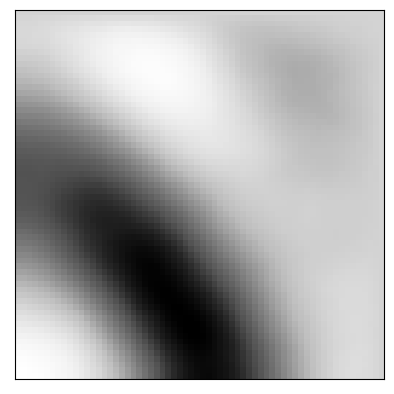}
\includegraphics[width=0.19\textwidth,height=.18\textwidth]{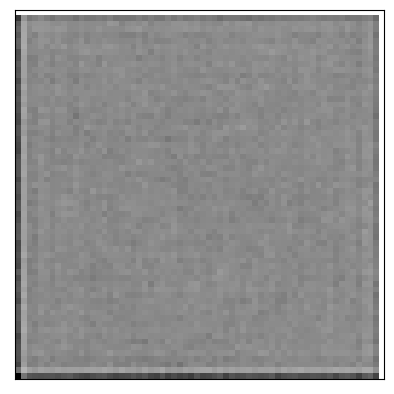}
\includegraphics[width=0.19\textwidth,height=.18\textwidth]{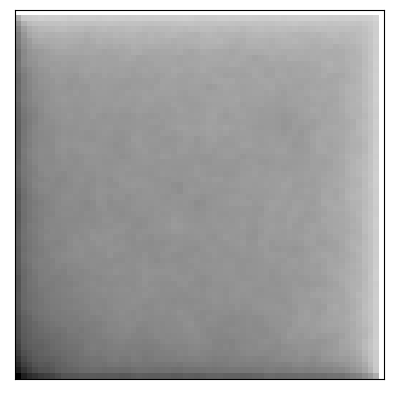}
\end{tabular}
\vspace{-5pt}
\begin{tabular}{ccccc}
\includegraphics[width=0.19\textwidth,height=.18\textwidth]{NSE/truth/NSE_03.png}
\includegraphics[width=0.19\textwidth,height=.18\textwidth]{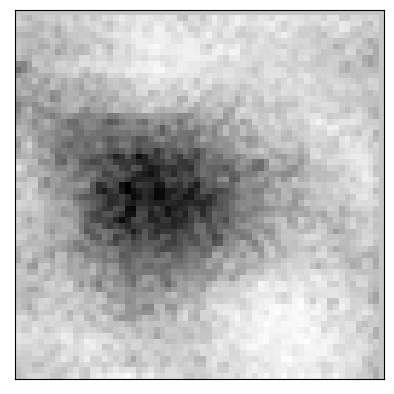}
\includegraphics[width=0.19\textwidth,height=.18\textwidth]{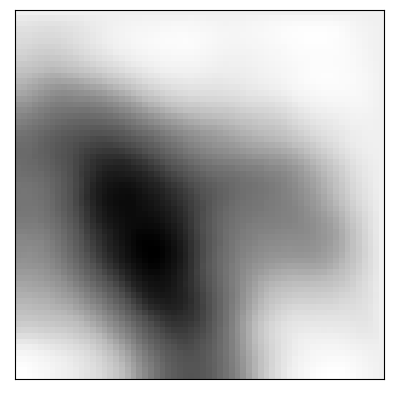}
\includegraphics[width=0.19\textwidth,height=.18\textwidth]{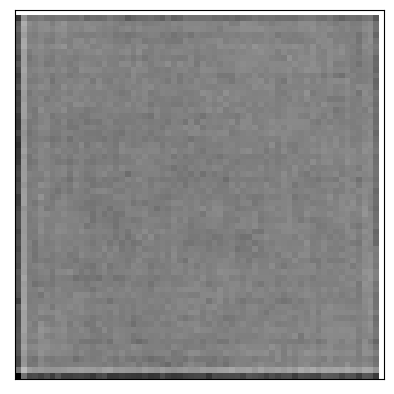} 
\includegraphics[width=0.19\textwidth,height=.18\textwidth]{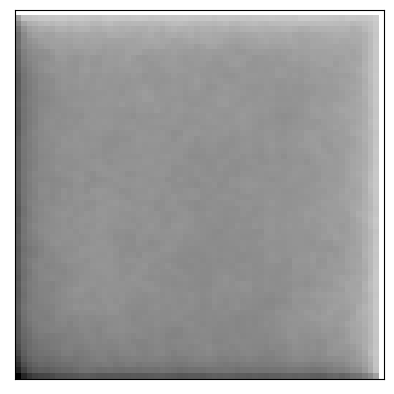} 
\end{tabular}
\vspace{-5pt}
\begin{tabular}{ccccc}
\includegraphics[width=0.19\textwidth,height=.18\textwidth]{NSE/truth/NSE_06.png}
\includegraphics[width=0.19\textwidth,height=.18\textwidth]{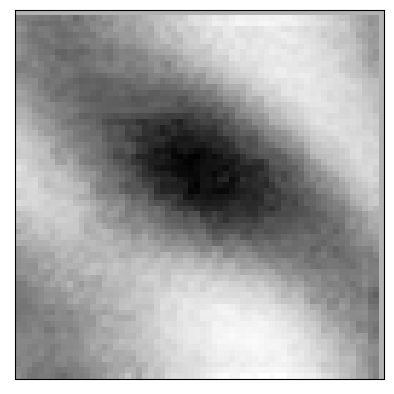}
\includegraphics[width=0.19\textwidth,height=.18\textwidth]{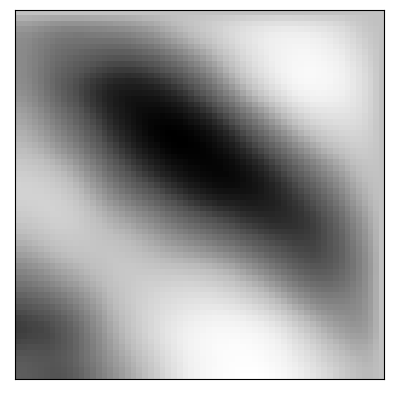}
\includegraphics[width=0.19\textwidth,height=.18\textwidth]{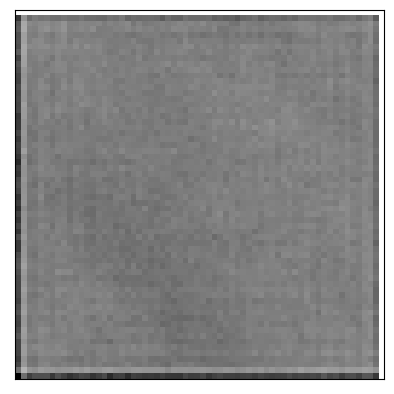}
\includegraphics[width=0.19\textwidth,height=.18\textwidth]{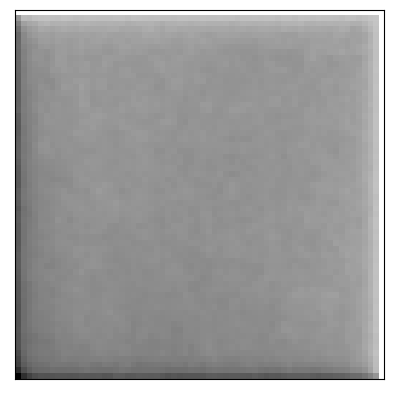}
\end{tabular}
\vspace{-5pt}
\begin{tabular}{ccccc}
\includegraphics[width=0.19\textwidth,height=.18\textwidth]{NSE/truth/NSE_09.png}
\includegraphics[width=0.19\textwidth,height=.18\textwidth]{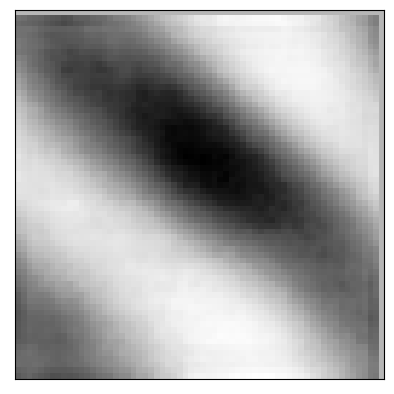}
\includegraphics[width=0.19\textwidth,height=.18\textwidth]{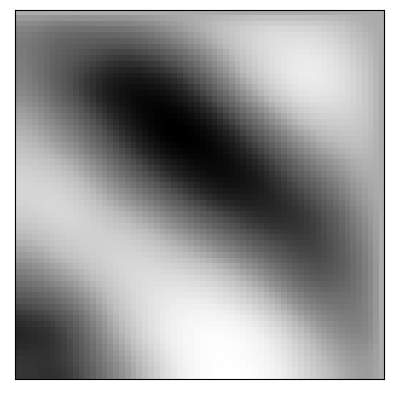}
\includegraphics[width=0.19\textwidth,height=.18\textwidth]{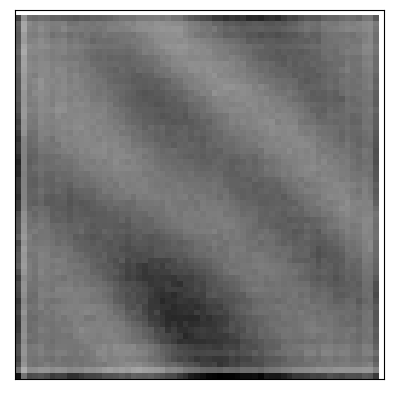}
\includegraphics[width=0.19\textwidth,height=.18\textwidth]{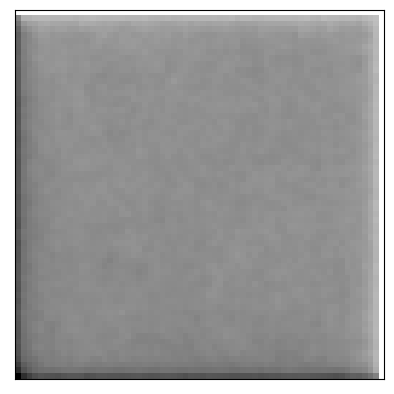}
\end{tabular}
\vspace{-5pt}
\caption{MCMC inverse solutions of Navier-Stokes equation in the whitened space. Columns from left to right: true vorticity, posterior mean estimates by STBP and STGP, posterior standard deviation estimates by STBP and STGP models respectively. Rows from top to bottom: time step $j = 0, 3, 6, 9$.}
\label{fig:NSE_MCMC}
\end{figure}

Figure \ref{fig:NSE_MCMC} compares the MCMC estimates for the first 10 unit time vorticity of NSE obtained in the whitened space by STBP and STGP models. Albeit noisy, the posterior mean estimates by STBP (the second column) manifest spatial features closer to the truth compared with the those by STGP (the thrid column) model. For the UQ, posterior standard deviation results are not very informative for both models, but a few of them at later time ($j=6, 9$) show more spatial traits in the STBP model than in the STGP model.

\subsection{NOAA Temperature Anomalies}

In the spatiotemporal imputation of NOAA temperature anomalies, Table \ref{tab:noaatmp} compares the three models in terms of relative error (RLE), negative log-posterior (NLP), and time. STBP yields the best result with the smallest RLE within time comparable to the other two models.

\renewcommand{\arraystretch}{0.6}
\begin{table}[htbp]
\caption{Comparison of MAP estimates of NOAA temperature anomalies generated by STBP, STGP and time-uncorrelated prior models in terms of RLE, negative log-posterior (NLP), and time. Standard deviations (in bracket) are obtained by repeating the experiments for 10 times with different random seeds for initialization.}
\centering
\begin{tabular}{|c|c|c|c|}
\toprule
 & time-uncorrelated  & STGP  & STBP \\
 \midrule
RLE  & 0.6356 (0.0137) &   0.3903 (0.0199) & \cellcolor{lightgray} 0.3008 (0.0008) \\
NLP  &   143172.66 (741.80) & 311310.7 (8776.36) &  146420.38 (746.16)\\
time &  \cellcolor{lightgray} 510.41 (1.51) &  555.77 (4.19) & 513.79 (4.95) \\
\bottomrule
\end{tabular}
\label{tab:noaatmp}
\end{table}

Posterior median estimates by applying wn-$\infty$-MALA (Algorithm \ref{alg:wn-infMC}) to these models are compared in Figure \ref{fig:noaatmp_MCMC}. The RLEs are $32.4\%$ for STBP, $32.43\%$ for STGP, and $40.83\%$ for time-uncorrelated model tested respectively on the $10\%$ held-out data. Similarly to Figure \ref{fig:noaatmp_MAP_whiten}, STBP generates the best imputation on both held-out data and missing values.

\begin{figure}[t]
\begin{tabular}{cccc}
\qquad Observations & \qquad \qquad STBP & \quad \qquad \qquad STGP & \quad \qquad \qquad time-uncorrelated  \\
\end{tabular}
\vspace{-5pt}
\begin{tabular}{ccccc}
\includegraphics[width=0.24\textwidth,height=.18\textwidth]{NOAATMP/obs/noaa_obs1999-1.png}
\includegraphics[width=0.24\textwidth,height=.18\textwidth]{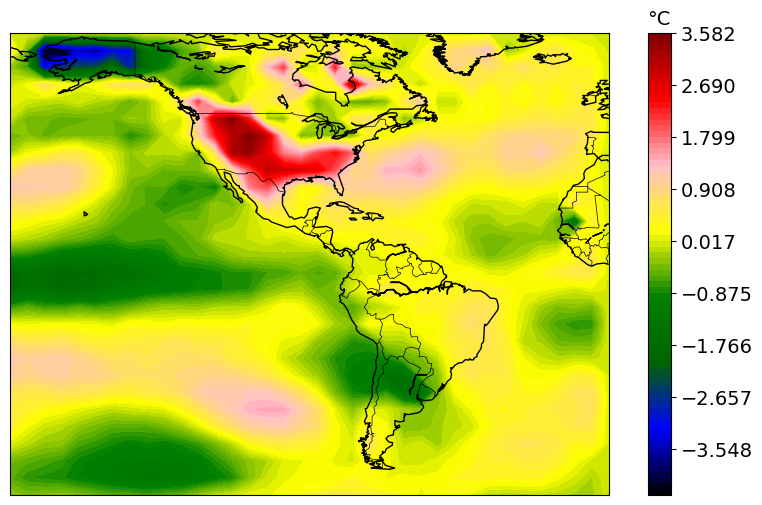}
\includegraphics[width=0.24\textwidth,height=.18\textwidth]{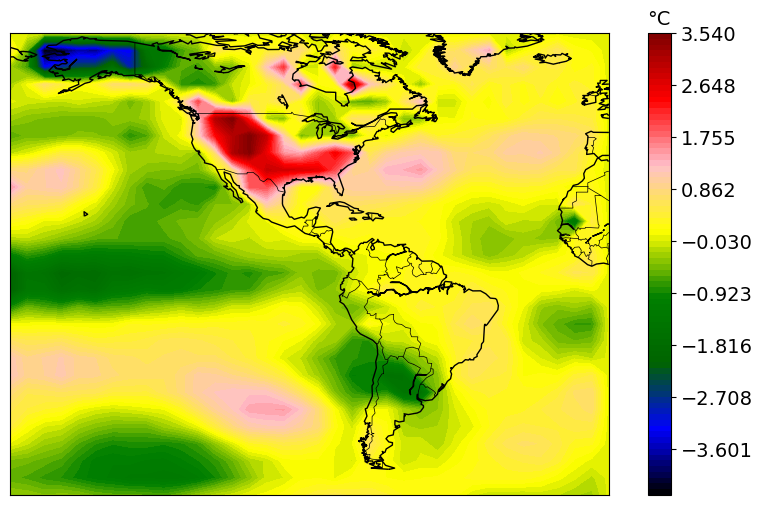}
\includegraphics[width=0.24\textwidth,height=.18\textwidth]{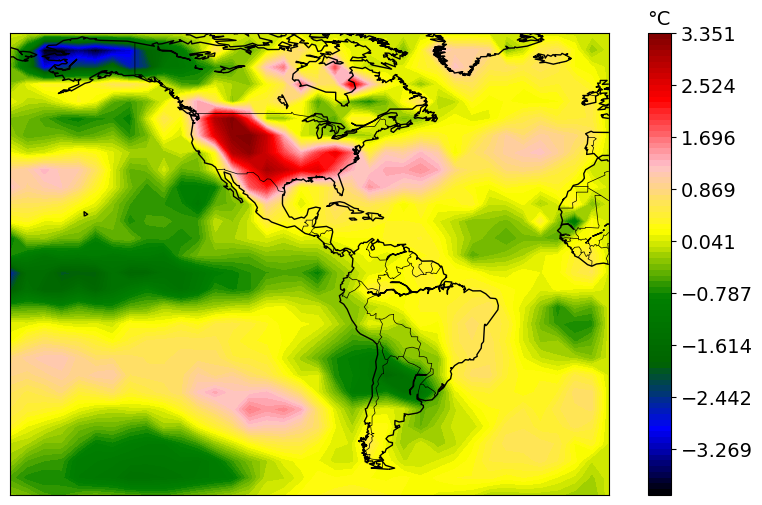}
\end{tabular}
\vspace{-5pt}
\begin{tabular}{ccccc}
\includegraphics[width=0.24\textwidth,height=.18\textwidth]{NOAATMP/obs/noaa_obs2005-1.png}
\includegraphics[width=0.24\textwidth,height=.18\textwidth]{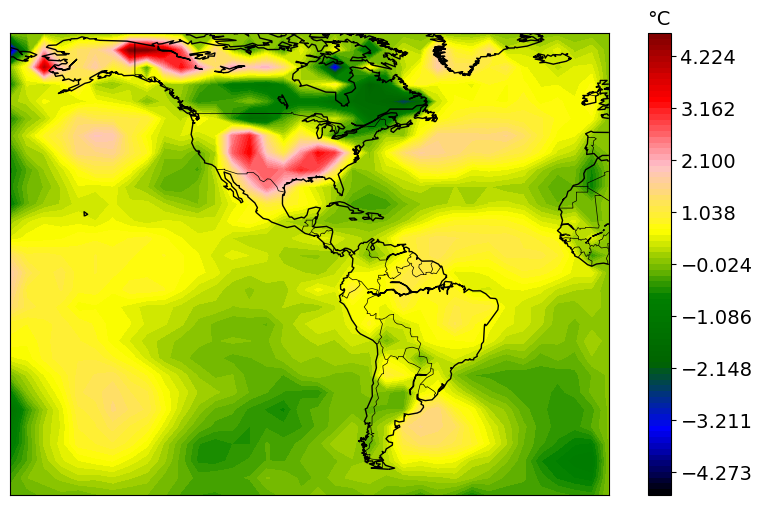}
\includegraphics[width=0.24\textwidth,height=.18\textwidth]{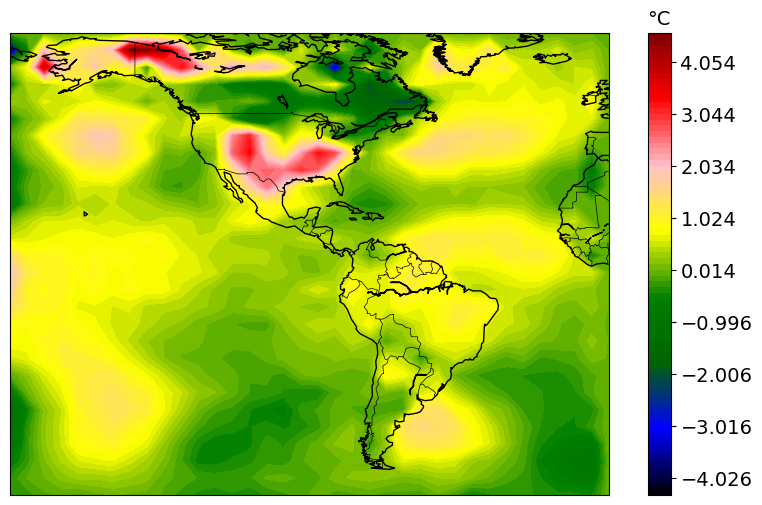}
\includegraphics[width=0.24\textwidth,height=.18\textwidth]{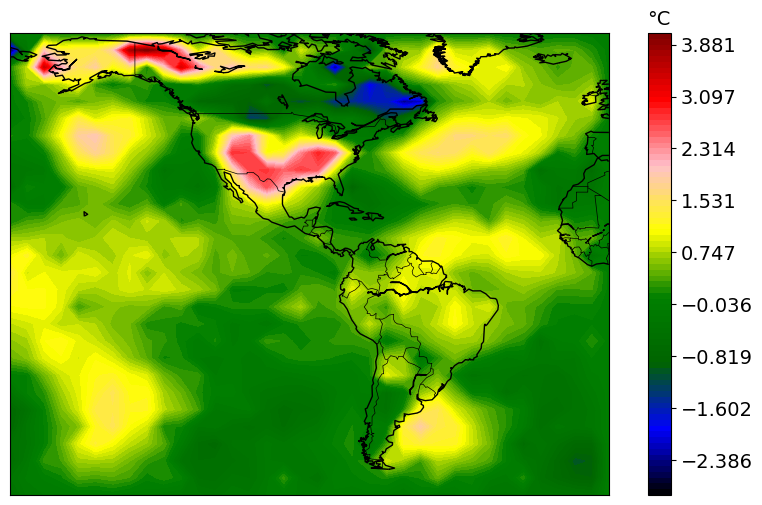}
\end{tabular}
\vspace{-5pt}
\begin{tabular}{ccccc}
\includegraphics[width=0.24\textwidth,height=.18\textwidth]{NOAATMP/obs/noaa_obs2011-1.png}
\includegraphics[width=0.24\textwidth,height=.18\textwidth]{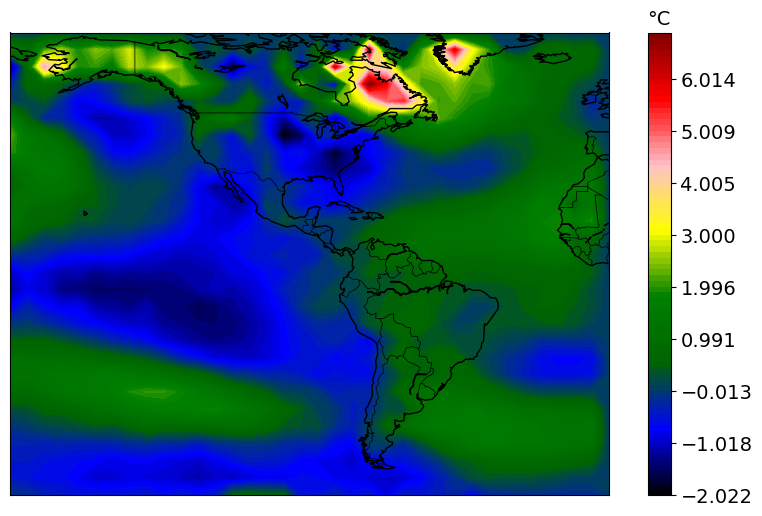}
\includegraphics[width=0.24\textwidth,height=.18\textwidth]{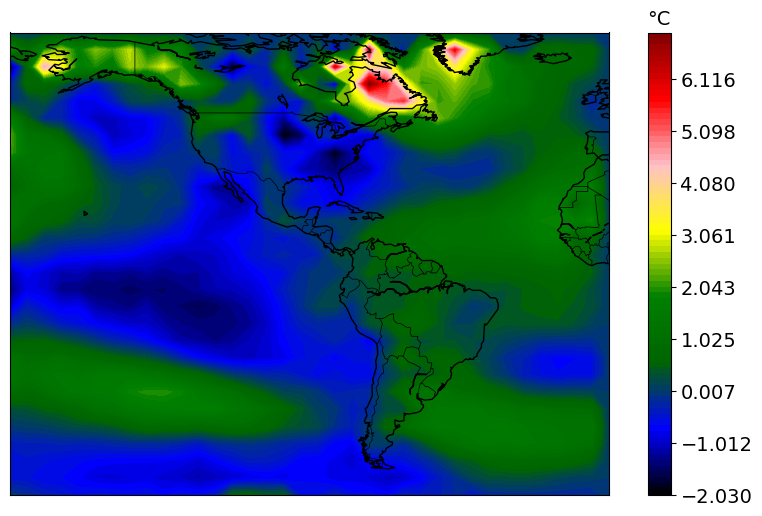}
\includegraphics[width=0.24\textwidth,height=.18\textwidth]{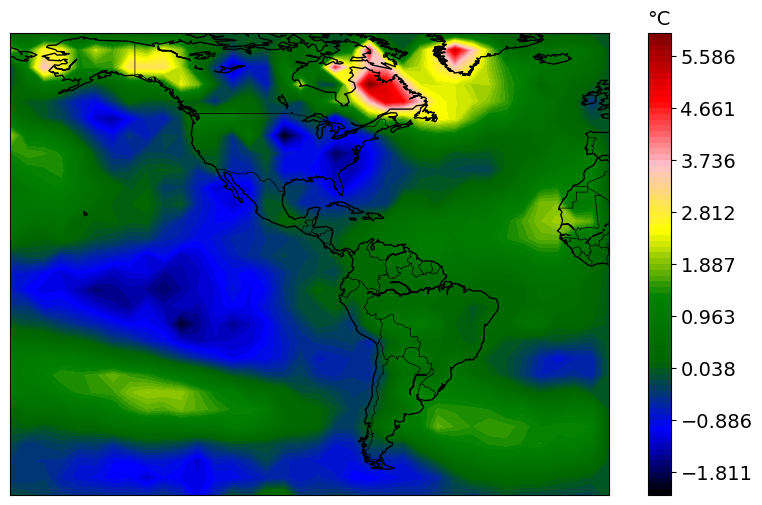}
\end{tabular}
\vspace{-5pt}
\begin{tabular}{ccccc}
\includegraphics[width=0.24\textwidth,height=.18\textwidth]{NOAATMP/obs/noaa_obs2017-1.png}
\includegraphics[width=0.24\textwidth,height=.18\textwidth]{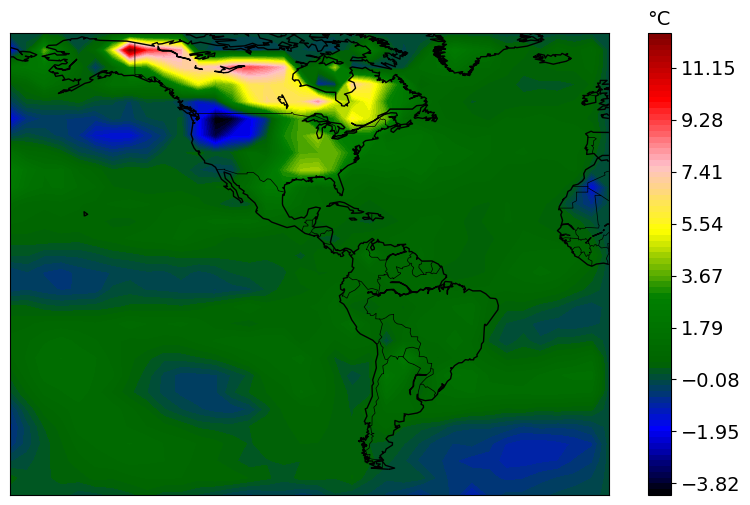}
\includegraphics[width=0.24\textwidth,height=.18\textwidth]{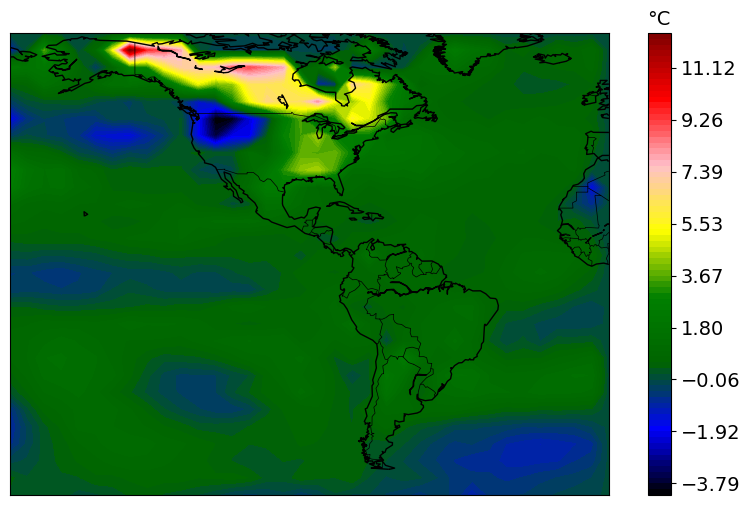}
\includegraphics[width=0.24\textwidth,height=.18\textwidth]{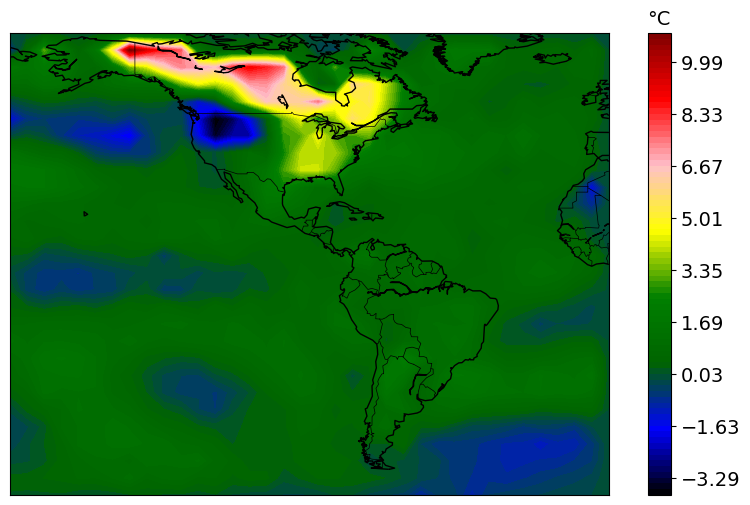}
\end{tabular}
\vspace{-5pt}
\caption{MCMC reconstruction of NOAA temperature anomalies. Columns: observations with missing values, posterior median estimates by STBP, STGP and time-uncorrelated models, respectively. Rows from top to bottom: time step $t_j = 1999, 2005, 2011$, and $2017$ (Januaries).}
\label{fig:noaatmp_MCMC}
\end{figure}

\end{document}